\documentclass[a4paper,onecolumn,11pt,accepted = 2023-02-23]{quantumarticle}
\pdfoutput=1
\usepackage[utf8]{inputenc}
\usepackage[english]{babel}
\usepackage[T1]{fontenc}
\usepackage{amsmath}
\usepackage{hyperref}

\usepackage[numbers,sort&compress]{natbib}

\usepackage{tikz}
\usepackage{lipsum}

\usepackage{graphicx}
\usepackage{dcolumn}
\usepackage{bm}

\usepackage{mathtools}
\usepackage{amsmath}
\usepackage{amssymb}
\usepackage{amsthm}
\usepackage[normalem]{ulem}
\usepackage{dsfont}

\newtheorem{Def}{Definition}
\newtheorem{Thm}[Def]{Theorem}

\newtheorem{Lem}[Def]{Lemma}

\usepackage{xcolor}

\newcommand{\1}{\mathds{1}}
\newcommand{\bra}[1]{\langle {#1} \vert}

\newcommand{\ket}[1]{\vert {#1} \rangle}
\newcommand{\dket}[1]{\vert {#1} \rangle\!\rangle}
\newcommand{\braket}[2]{\langle {#1} \vert {#2} \rangle}

\newcommand{\ketbra}[2]{\vert {#1} \rangle\!\langle {#2} \vert}

\newcommand{\dketbra}[1]{\vert {#1} \rangle\!\rangle\!\langle\!\langle {#1} \vert}
\newcommand{\Tr}[0]{{\mathrm{Tr}}}
\newcommand{\dd}[0]{{\mathrm{d}}}
\renewcommand{\Im}[0]{{\mathrm{Im}}}

\makeatletter
\newcommand{\doublewidetilde}[1]{{%
  \mathpalette\double@widetilde{#1}%
}}
\newcommand{\double@widetilde}[2]{%
  \sbox\z@{$\m@th#1\widetilde{#2}$}%
  \ht\z@=.9\ht\z@
  \widetilde{\box\z@}%
}
\makeatother

\newcommand{\map}[1]{\widetilde{\mathcal{#1}}}
\newcommand{\supermap}[1]{\doublewidetilde{\mathcal{#1}^{}}}

\begin{document}

\title{Universal construction of decoders from encoding black boxes}
\author{Satoshi Yoshida}
\email{satoshi.yoshida@phys.s.u-tokyo.ac.jp}
\affiliation{Department of Physics, Graduate School of Science, The University of Tokyo, Hongo 7-3-1, Bunkyo-ku, Tokyo 113-0033, Japan}
\orcid{0000-0002-0521-5209}
\author{Akihito Soeda}
\affiliation{Department of Physics, Graduate School of Science, The University of Tokyo, Hongo 7-3-1, Bunkyo-ku, Tokyo 113-0033, Japan}
\affiliation{Principles of Informatics Research Division, National Institute of Informatics, 2-1-2 Hitotsubashi, Chiyoda-ku, Tokyo 101-8430, Japan}
\affiliation{Department of Informatics, School of Multidisciplinary Sciences, SOKENDAI (The Graduate University for Advanced Studies), 2-1-2 Hitotsubashi, Chiyoda-ku, Tokyo 101-8430, Japan}
\orcid{0000-0002-7502-5582}
\author{Mio Murao}
\affiliation{Department of Physics, Graduate School of Science, The University of Tokyo, Hongo 7-3-1, Bunkyo-ku, Tokyo 113-0033, Japan}
\affiliation{Trans-scale Quantum Science Institute, The University of Tokyo, Bunkyo-ku, Tokyo 113-0033, Japan}
\orcid{0000-0001-7861-1774}

\maketitle

\begin{abstract}
Isometry operations encode the quantum information of the input system to a larger output system, while the corresponding decoding operation would be an inverse operation of the encoding isometry operation.  Given an encoding operation as a black box from a $d$-dimensional system to a $D$-dimensional system, we propose a universal protocol for isometry inversion that constructs a decoder from multiple calls of the encoding operation.  This is a probabilistic but exact protocol whose success probability is independent of $D$. For a qubit ($d=2$) encoded in $n$ qubits, our protocol achieves an exponential improvement over any tomography-based or unitary-embedding method, which cannot avoid $D$-dependence.
We present a quantum operation that converts multiple parallel calls of any given isometry operation to random parallelized unitary operations, each of dimension $d$.  Applied to our setup, it universally compresses the encoded quantum information to a $D$-independent space, while keeping the initial quantum information intact.  This compressing operation is combined with a unitary inversion protocol to complete the isometry inversion.
We also discover a fundamental difference between our isometry inversion protocol and the known unitary inversion protocols by analyzing isometry complex conjugation and isometry transposition.  General protocols including indefinite causal order are searched using semidefinite programming for any improvement in the success probability over the parallel protocols. We find a sequential ``success-or-draw'' protocol of universal isometry inversion for $d = 2$ and $D = 3$, thus whose success probability exponentially improves over parallel protocols in the number of calls of the input isometry operation for the said case.
\end{abstract}

\section{Introduction}
Universal transformations of quantum states have played an essential role in the fundamental understanding of quantum information theory and its applications \cite{Nielsen2010Quantum}. Recently,  \emph{higher-order quantum transformations}, namely, universal transformations of quantum  operations given as black boxes, have been studied in the contexts of processing unitary operations \cite{Chiribella2005Optimal, Bisio2010Optimal, Sedlak2019Optimala, Yang2020Optimal, Sedlak2020Probabilistic, Chiribella2008Optimal, Bisio2014Optimal, Dur2015Deterministic, Chiribella2015Universal, Soleimanifar2016Nogo, Micuda2016Experimental, Bisio2009Optimalb, Bisio2010Informationdisturbance, Miyazaki2019Complex, Chiribella2016Optimal, Navascues2018Resetting, Quintino2019Reversing, Quintino2019Probabilistic, Quintino2021Deterministic,Bartlett2009Quantum, Araujo2014Quantum, Bisio2016Quantum, Dong2020Controlled}, non-Markovian quantum process \cite{Milz2018Reconstructing, Milz2018Entanglement, Pollock2018Operational, Pollock2018Tomographically, Pollock2018NonMarkovian, Sakuldee2018NonMarkovian, Jorgensen2019Exploiting, Taranto2019Quantum, Taranto2019Structure, Milz2019Completely, Milz2020When, Milz2021Quantum, Giarmatzi2021Witnessing}, and dynamical resource theory \cite{theurer2019quantifying, Chitambar2019Quantum, Gour2019How, Liu2019Resource, Gour2020Dynamical, Gour2020Dynamicala, Gour2021Entanglement, Liu2020Operational, Yuan2020Oneshot, Theurer2020Quantifying, Regula2021Fundamental, Chen2020Entanglementbreaking, Kristjansson2020Resource, Hsieh2021Communication, Gour2021Uniqueness}. They can also be interpreted as quantum functional programming \cite{Altenkirch2005Functional, Ying2016Foundations}.
The concept of higher-order quantum transformation was initially introduced as single-input quantum supermaps \cite{Chiribella2008Transforming},  and a multiple-input version was introduced as quantum combs \cite{Chiribella2009Theoretical}, which are realizable by quantum circuits with a fixed ordering of input operations. A similar concept was studied as channels with memory \cite{Kretschmann2005Quantum} and quantum strategies \cite{gutoski2007toward}.

In particular, higher-order quantum transformations of unitary operations have been extensively studied for aiming to utilize in quantum information processing  (e.g., estimation of group transformations \cite{Chiribella2005Optimal}, quantum learning of unitary operations \cite{Bisio2010Optimal, Sedlak2019Optimala, Yang2020Optimal, Sedlak2020Probabilistic}, cloning of unitary operations \cite{Chiribella2008Optimal, Bisio2014Optimal, Dur2015Deterministic, Chiribella2015Universal, Soleimanifar2016Nogo, Micuda2016Experimental}, process tomography \cite{Bisio2009Optimalb, Bisio2010Informationdisturbance}, unitary complex conjugation \cite{Miyazaki2019Complex}, unitary inversion \cite{Chiribella2016Optimal, Navascues2018Resetting, Quintino2019Reversing, Quintino2019Probabilistic, Quintino2021Deterministic} and unitary controllization \cite{Araujo2014Quantum, Bisio2016Quantum, Dong2020Controlled}).
Isometry operations are also frequently used in quantum protocols, as they preserve information of the input state similarly to unitary operations. For instance, encoding of quantum information is represented by an isometry operation \cite{Nielsen2010Quantum}. In addition, many quantum algorithms use fixed auxiliary states and such algorithms can also be considered to be utilizing isometry operations (e.g.,  Harrow-Hassidim-Lloyd (HHL) algorithm \cite{Harrow2009Quantum}).  Despite their importance, higher-order quantum transformations of isometry operations are not well investigated yet.    

In this work, we study one of the fundamental higher-order quantum transformation tasks for isometry operations, namely,  \emph{isometry inversion}. 
Isometry inversion is a task to implement the inverse map of an input isometry operation, interpreted as retrieving quantum information encoded by the isometry operation. Such  a retrieval of quantum information is widely studied in the context of quantum error correction \cite{Nielsen2010Quantum}, quantum secret sharing \cite{Gottesman2000Theory}, quantum communication \cite{Wilde2017Classical}, and uncomputation \cite{Bennett1973Logical, Aaronson2015Classification}. Such studies usually assume that the complete descriptions of encoding operations are given.  Contrary, this work considers a {\it universal} protocol without knowing the descriptions except for the dimensions of the input system and the output system.
The universal protocol implements the inverse operation \emph{probabilistically} but \emph{exactly}.

The special case of probabilistic exact isometry inversion, namely probabilistic exact unitary inversion \cite{Quintino2019Reversing, Quintino2019Probabilistic}, is known. This implementation relies on the existence of two protocols: (deterministic exact) unitary complex conjugation \cite{Miyazaki2019Complex} and (probabilistic exact) unitary transposition \cite{Quintino2019Probabilistic}. However, this strategy is not directly applicable for isometry inversion. The key idea for unitary complex conjugation protocol presented in Ref.~\cite{Miyazaki2019Complex} is to utilize the knowledge of representation theory of the unitary group, but this idea  cannot be applied to isometry complex conjugation since the set of isometry operations does not form a group. In fact, we show below a no-go theorem for isometry complex conjugation.

The most trivial way to implement isometry inversion in the black box setting is to obtain  a classical description of the isometry operation by quantum process tomography \cite{Nielsen2010Quantum} and then implement the inverse map based on the description, namely, using the ``measure-and-prepare'' strategy \cite{horodecki2003entanglement, Bisio2010Optimal}.  However, known quantum process tomography protocols \cite{Mohseni2008Quantum} require  a  $D$-dependent number of experiments to obtain an approximate description of an isometry operation $\widetilde{\mathcal{V}}_{d,D}$ from a $d$-dimensional system to a $D$-dimensional system  (see the discussion in Section~\ref{subsec:comparison}). Another straightforward way is to embed an isometry operation  $\widetilde{\mathcal{V}}_{d,D}$ in a $D$-dimensional unitary operation and then apply a universal probabilistic exact unitary inversion protocol \cite{Quintino2019Reversing}. However, the success probability of such a protocol cannot avoid the dependence on $D$, either.   In particular, for isometry operations encoding quantum information of a qudit (a $d$-dimensional system) into $n$ qudits (a $d^n$-dimensional system), the exponential cost in $n$ due to the dimensionality of $D=d^n$  may seem to be inevitable for implementing isometry inversion. 

Nevertheless, we present a probabilistic but exact protocol for isometry inversion of which success probability does not depend on $D$. Due to this property, our protocol can significantly outperform the protocols based on the two strategies mentioned above.   To compare with the protocol based on unitary inversion, we consider an isometry operation that encodes a qubit into five qubits, i.e., $d=2$ and $D=2^5$. The unitary inversion requires at least $D-1=31$ calls to obtain a non-zero success probability  \cite{Quintino2019Reversing}, but our protocol achieves a success probability $p=87\%$ by $20$ calls.  Compared with the protocol based on quantum process tomography, our protocol can implement isometry inversion approximately within a fixed error $\epsilon$ by a $D$-independent number of calls.  This comparison exhibits the potential power of a higher-order quantum transformation that directly transforms a black box operation without evaluating its classical description.

We also clarify a crucial difference between the unitary inversion protocols presented in Ref.~\cite{Quintino2019Reversing} and the isometry inversion protocol. Reference \cite{Quintino2019Reversing} presents a systematic construction of unitary inversion protocols concatenating unitary transposition \cite{Quintino2019Probabilistic} and unitary complex conjugation \cite{Miyazaki2019Complex}. The unitary complex conjugation protocol utilizes the fact that the complex conjugate representation of the unitary group is unitarily equivalent to the antisymmetric subspace of the tensor representation of the unitary group \cite{Miyazaki2019Complex}. The unitary transposition protocol presented in Ref.~\cite{Quintino2019Probabilistic} uses a variant of the gate teleportation \cite{Gottesman1999Demonstrating} or the probabilistic port-based teleportation \cite{Ishizaka2008Asymptotica, Studzinski2017Portbased}. However, we show that isometry inversion cannot be implemented by concatenating the corresponding tasks since no probabilistic exact isometry complex conjugation is  possible for $D\geq 2d$. We also show that any isometry inversion protocol  transposing a ``pseudo complex conjugate'' map (see Section \ref{subsec:isometry_pseudo_cc}) of an isometry operation by a variant of gate teleportation \cite{Gottesman1999Demonstrating} is less efficient than our protocol. 

The key idea of our isometry inversion protocol is a quantum operation that universally  compresses the $D$-dimensional output spaces of the isometry operations into $d$-dimensional quantum systems. 
We first extend the irreducible decomposition of the tensor product of {unitary operators} known as the Schur-Weyl duality to isometry operators. We show that the tensor product of  isometry operators also admits a block diagonal decomposition, despite  the isometry operators not forming a group. This decomposition identifies relevant and irrelevant components to retrieve quantum information encoded by the isometry operation. The compressing quantum operation discards the irrelevant component. We utilize the compressing quantum operation to convert unitary inversion protocols to isometry inversion protocols avoiding the $D$-dependence of the success probability and the no-go theorem for isometry complex conjugation.

Our isometry inversion protocol uses input operations in parallel. Such parallel protocols form an essential class of higher-order quantum transformation because parallelization is a common technique to reduce the circuit depth \cite{Gyongyosi2020Circuit}.
However, more general protocols than parallel ones can be helpful to improve the success probability. In this work, we consider sequential protocols using input operations in a fixed ordering and general protocols including indefinite causal order \cite{Oreshkov2012Quantum, Chiribella2013Quantum, Araujo2015Witnessing, Wechs2019Definition, Bisio2019Theoretical, Yokojima2021Consequences, Vanrietvelde2021Routed} in addition to parallel ones. To see the performance improvement in our setting, we conduct semidefinite programming (SDP) to obtain the optimal success probability of parallel, sequential, and general protocols. 

The rest of this paper is organized as follows. Section \ref{sec:problem_setting} states the problem setting for implementing isometry transposition, isometry complex conjugation, and isometry inversion.
Section \ref{sec:isometry_inversion} presents the main result of this paper, constructing a parallel protocol for isometry inversion by investigating the compressing quantum operation. Section \ref{sec:difference} discusses the difference between our isometry inversion protocol and the previous {work} on unitary inversion. Section \ref{sec:SDP} shows numerical results on the optimal success probability of parallel, sequential, and general protocols including indefinite causal order for isometry inversion, isometry  (pseudo) complex conjugation, and isometry transposition. Section \ref{sec:conclusion} concludes the paper.

\section{Problem Setting}
\label{sec:problem_setting}
\subsection{Inverse maps of isometry operations}
\label{sec:decoding}
A $d$-dimensional quantum system is represented by a Hilbert space $\mathcal{H}=\mathbb{C}^d$ and a state of the system is represented by a density operator (a positive semi-definite operator with unit trace) $\rho$ on $\mathcal{H}$.  A state is called pure if its density operator has rank 1, i.e., $\rho=\ketbra{\psi}{\psi}$ for $\ket{\psi}\in \mathcal{H}$. The set of linear operators on $\mathcal{H}$ is denoted by $\mathcal{L}(\mathcal{H})$ and the set of linear operators from $\mathcal{H}$ to $\mathcal{H^\prime}$ is denoted by $\mathcal{L}(\mathcal{H} \to \mathcal{H}^\prime)$. When we explicitly specify the dimension of the set of linear operators, we denote the corresponding sets by  $\mathcal{L}(\mathbb{C}^d)$ and $\mathcal{L}(\mathbb{C}^d \to \mathbb{C}^D)$, respectively, for $\mathcal{H}=\mathbb{C}^d$ and $\mathcal{H}^\prime=\mathbb{C}^d$.
We only consider quantum systems represented by finite-dimensional Hilbert spaces in this paper.   

An isometry operation transforms a pure input state in a $d$-dimensional system $\mathcal{H}=\mathbb{C}^d$ to a pure output state in a $D$-dimensional system $\mathcal{H}'=\mathbb{C}^D$ with $D \geq d$ where the transformation preserves the inner product of two input states.  An isometry operation is regarded to encode (and spread for the case of $d<D$) quantum information represented by a $d$-dimensional quantum state into a $D$-dimensional state.  Unitary operations are special cases of isometry operations with $d=D$.  

Formally, an isometry operation $\widetilde{\mathcal{V}}: \mathcal{L}(\mathcal{H}) \to \mathcal{L}(\mathcal{H}^\prime)$ for $\mathcal{H}=\mathbb{C}^d$ and $\mathcal{H}^\prime=\mathbb{C}^D$ is a completely positive trace preserving (CPTP) map given as $\widetilde{\mathcal{V}}(\rho)=V\rho V^{\dagger}$ in terms of an isometry operator $V: \mathcal{H} \to \mathcal{H}^\prime$, an element of the set of isometry operators $\mathbb{V}_{\mathrm{iso}} (d, D)$ defined by
\begin{align}
    \mathbb{V}_\mathrm{iso} (d,D)
    &\coloneqq \{V \in \mathcal{L}  {(\mathbb{C}^d \to \mathbb{C}^D)}  |V^\dagger V=\1_d\},
\end{align}
where  $\1_{d}$ is the identity operator on {$\mathcal{H}=\mathbb{C}^d$}, and $V^\dagger: \mathcal{H}^\prime \to \mathcal{H}$ is the adjoint of $V$.  In this notation, the tilde symbol on top of $\mathcal{V}$ represents a linear map\footnote{This convention is adopted from Ref.~\cite{Quintino2019Probabilistic}.}.   We denote a set of $d$-dimensional unitary operators by $\mathbb{U}(d)$, which is equivalent to $\mathbb{V}_{\mathrm{iso}} (d, d)$, and a unitary operation corresponding to a unitary operator $U$ by $\widetilde{\mathcal{U}}$.

We define the most general map which can decode the states encoded by an isometry operation $\widetilde{\mathcal{V}}$.  We consider a completely positive  (CP) map $\widetilde{\mathcal{V}}_{\mathrm{inv}}: \mathcal{L}(\mathcal{H}^\prime)\to \mathcal{L}(\mathcal{H})$ that satisfies 
\begin{align}
    \widetilde{\mathcal{V}}_{\mathrm{inv}}\circ \widetilde{\mathcal{V}}=\widetilde{\1}_{d},
    \label{eq:inverse}
\end{align}
for an isometry operation $\widetilde{\mathcal{V}}$ corresponding to $V\in \mathbb{V}_\mathrm{iso} (d,D)$, where $\widetilde{\1}_d$ is the identity operation on $\mathcal{L}(\mathcal{H})$, which is defined by $\widetilde{\1}_d (\rho)=\rho$ for all density matrices $\rho$ of a $d$-dimensional system. 
We refer to a map $\widetilde{\mathcal{V}}_{\mathrm{inv}}$  satisfying Eq.~(\ref{eq:inverse}) as an {\it inverse map} of $\widetilde{\mathcal{V}}$.  Note that $\widetilde{\mathcal{V}}_{\mathrm{inv}}$ is not necessarily a trace preserving map, while the composition of $\widetilde{\mathcal{V}}_{\mathrm{inv}}$ and $\widetilde{\mathcal{V}}$ is trace preserving.  An inverse map is not necessarily the adjoint map $\widetilde{\mathcal{V}}^{\dagger}$ defined as $\widetilde{\mathcal{V}}^{\dagger}(\rho)=V^{\dagger}\rho V$,  either. For instance, a  CP map $\widetilde{\mathcal{V}}^\prime_\alpha$ defined by
\begin{align}
    \widetilde{\mathcal{V}}^\prime_{\alpha}(\rho)\coloneqq V^{\dagger} \rho V+\alpha  \frac{I_d}{d}
    \Tr
    \left[
    \Pi_{(\Im V)^{\perp}}\rho
    \right]
    \label{eq:def_f(V)}
\end{align}
is also an inverse map of $\widetilde{\mathcal{V}}$ for $\alpha\geq 0$, where $\Pi_{(\Im V)^{\perp}}$ is a projector onto the orthogonal subspace of  the image of $V: \mathcal{H} \to \mathcal{H^\prime}$ denoted by $\Im V\coloneqq V(\mathcal{H})$. The map  $\widetilde{\mathcal{V}}^\prime_\alpha$ is not trace preserving for $\alpha\neq 1$, but the composition of  $\widetilde{\mathcal{V}}$ and $\widetilde{\mathcal{V}}^\prime_\alpha$ is trace preserving because the second term of Eq.~(\ref{eq:def_f(V)}) vanishes for all $\rho \in \mathcal{L}(\Im V)$.

\subsection{Higher-order quantum transformations of isometry operations: parallel protocols for probabilistic exact tasks}
\label{sec:supermap_superinstrument_higher-order}

In this paper, we present a probabilistic but universal and exact protocol to construct an inverse map $\widetilde{\mathcal{V}}_{\mathrm{inv}}$ from {\it finite calls} of an unknown isometry operation $\widetilde{\mathcal{V}}$ given as a black box.  Such a protocol can be regarded as implementing a {\it higher-order quantum transformation} of an isometry operation, similarly to the preceding works considered higher-order quantum transformations of a unitary operation \cite{Quintino2019Reversing, Quintino2019Probabilistic}, of which formulation is based on the notion of {\it quantum supermaps and superinstrument} \cite{Chiribella2008Transforming}.  

We first introduce the notations of  quantum supermaps \cite{Chiribella2008Transforming} describing higher-order deterministic quantum transformations.   A quantum supermap is a linear completely CPTP preserving transformation from an input map to an output map.  We consider a $k$-input supermap $\doublewidetilde{\mathcal{C}^{}}$ that transforms $k$  input maps $\widetilde{\Lambda}^{(i)}_{\mathrm{in}}:\mathcal{L}(\mathcal{I}_i)\to \mathcal{L}(\mathcal{O}_i)$ for $i\in \{1, \cdots, k\}$, where $\mathcal{I}_i$ and $\mathcal{O}_i$ represent the input Hilbert space and the output Hilbert space of the $i$-th input map $\widetilde{\Lambda}^{(i)}_{\mathrm{in}}$, respectively, to  an output map $\widetilde{\Lambda}_{\mathrm{out}}:\mathcal{L}(\mathcal{P})\to \mathcal{L}(\mathcal{F})$, where $\mathcal{P}$ and $\mathcal{F}$ represent the input Hilbert space (also referred to as the past space) and the output Hilbert space (also referred to as the future space), respectively, of $\widetilde{\Lambda}_{\mathrm{out}}$.   The double tilde symbol on top of $\mathcal{C}$ represents a linear supermap (or a linear superinstrument, which will be introduced in the next paragraph)\footnote{This convention is also adopted from Ref.~\cite{Quintino2019Probabilistic}.}.

\begin{figure}[tbh]
    \centering
    \includegraphics[width=0.5\linewidth]{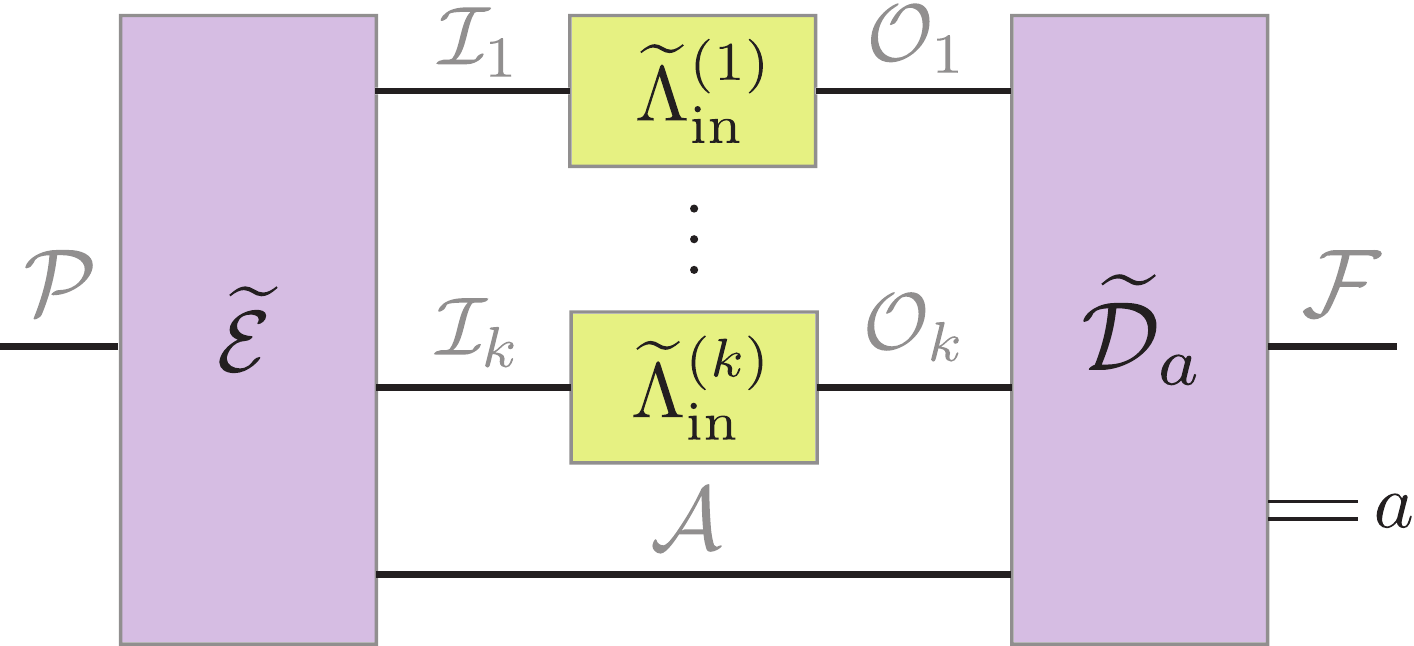}
    \caption{A quantum circuit representation of a parallel superinstrument  $\{\doublewidetilde{\mathcal{C}^{}}_a\}$ defined in Eq.~(\ref{eq:parallel_superinstrument}), where $\widetilde{\Lambda}_{\mathrm{in}}^{(i)}\;(i\in \{1, \cdots, k\})$ are input maps, $\widetilde{\mathcal{E}}$ is a CPTP map, and $\{\widetilde{\mathcal{D}}_a\}$ is a quantum instrument. A wire corresponds to a  Hilbert space and a box corresponds to an operation.  The double line in this figure represents the classical outcome of the measurement.}
    \label{fig:parallel_comb}
\end{figure}

Since we mainly focus on probabilistic parallel protocols, we introduce the notions of a {\it parallel} superinstrument representing a higher-order probabilistic quantum transformation where all input maps are called in parallel \cite{Quintino2019Probabilistic}.   For a parallel protocol, it is convenient to define the joint input Hilbert space $\mathcal{I}\coloneqq \bigotimes_{i=1}^{k}\mathcal{I}_i$ and the joint output Hilbert space as $\mathcal{O}\coloneqq \bigotimes_{i=1}^{k}\mathcal{O}_i$, as well as the joint input map $\widetilde{\Lambda}_{\mathrm{in}}\coloneqq \bigotimes_{i=1}^{k}\widetilde{\Lambda}^{(i)}_{\mathrm{in}}$.   Using these notations, a parallel superinstrument is a set of $k$-input supermaps $\{\doublewidetilde{\mathcal{C}^{}}_{a}\}: [\mathcal{L}(\mathcal{I}) \to \mathcal{L}(\mathcal{O})] \to [\mathcal{L}(\mathcal{P}) \to \mathcal{L}(\mathcal{F})]$  given by
\begin{align}
    \doublewidetilde{\mathcal{C}^{}}_{a}(\widetilde{\Lambda}_{\mathrm{in}})=\widetilde{\mathcal{D}}_a\circ 
    \left(
    \widetilde{\Lambda}_{\mathrm{in}} \otimes \widetilde{\1}_{\mathcal{A}}
    \right)\circ \widetilde{\mathcal{E}},\label{eq:parallel_superinstrument}
\end{align}
where $\widetilde{\mathcal{E}}: \mathcal{L}(\mathcal{P})\to \mathcal{L}(\mathcal{I}\otimes \mathcal{A})$ is a  CPTP map, $\{\widetilde{\mathcal{D}}_a\}: \mathcal{L}(\mathcal{O}\otimes \mathcal{A})\to \mathcal{L}(\mathcal{F})$ is a quantum instrument\footnote{ A quantum instrument $\{\widetilde{\Psi}_a\}$ is a set of CP maps such that $\sum_a \widetilde{\Psi}_a$ is a CPTP map \cite{Wilde2017Classical}.} and $\mathcal{A}$ is an auxiliary Hilbert space (see Figure~\ref{fig:parallel_comb}). We note that a parallel protocol using $k$ copies of an input isometry operation does not require the $k$ copies to be available simultaneously during the protocol. We may substitute by using a single black box isometry operation  repeatedly $k$ times.

This work also considers constructions of probabilistic exact parallel protocols for the complex conjugate map and the transposition map of an isometry operation to compare the protocol with that of unitary inversion.  For this sake, we introduce a notation of general higher-order probabilistic quantum transformation of an isometry operation $\widetilde{\mathcal{V}}$ to another map which is a function of $\widetilde{\mathcal{V}}$ denoted as $ \widetilde{f(V)}$ \footnote{We write the argument of the function $f$ as $V$ instead of $\map{V}$ to avoid a duplicate use of a tilde.}.  Note that the function $f$ is not necessarily linear in terms of a single $\widetilde{\mathcal{V}}$ , but has to be linear in terms of $\map{V}^{\otimes k}$, and not necessary to be trace preserving either.

A probabilistic exact  parallel protocol of a higher-order quantum transformation of  $k$ calls of an unknown isometry operation $\widetilde{\mathcal{V}}$ to $\widetilde{f(V)}$ is formulated as follows.
Let $\{\doublewidetilde{\mathcal{S}^{}}, \doublewidetilde{\mathcal{F}^{}}\}: [\mathcal{L}(\mathcal{I}) \to \mathcal{L}(\mathcal{O})]\to [\mathcal{L}(\mathcal{P}) \to \mathcal{L}(\mathcal{F}{)]}$ be a parallel superinstrument, where $\doublewidetilde{\mathcal{S}^{}}$ denotes the successful transformation and $\doublewidetilde{\mathcal{F}^{}}$ denotes the failure transformation. We say that  a parallel protocol $\doublewidetilde{\mathcal{S}^{}}$ is a probabilistic exact protocol to implement $ \widetilde{f(V)}$ from  $k$ calls of $\widetilde{\mathcal{V}}$ if
\begin{align}
    \doublewidetilde{\mathcal{S}^{}}(\widetilde{\mathcal{V}}^{\otimes k})=p_{\mathrm{succ}} \widetilde{f(V)}\;\;\;(\forall V\in  \mathbb{V}_\mathrm{iso} (d,D))\label{eq:def_success_probability}
\end{align}
holds.  We  require that $p_\mathrm{succ}$ is independent of the input isometry operation $\map{V}$ and the input quantum state $\rho_\mathrm{in}$, and call it the \emph{success probability} of $ \widetilde{f(V)}$  for the following reasons.
 Note that the probability to obtain the successful measurement outcome is $p_{\text{succ}} \Tr [ \widetilde{f(V)}( \rho_\mathrm{in})]$, where $ \rho_\mathrm{in}$ is the input state. The probability to obtain the successful measurement can be divided into two terms; one is the success probability of the protocol denoted by $p_{\text{succ}}$, and the other is the success probability of the map $ \widetilde{f(V)}$ denoted by $\Tr [ \widetilde{f(V)}( \rho_\mathrm{in})]$.
For the case of $ \widetilde{f(V)}=\widetilde{\mathcal{V}}_{\mathrm{inv}}$ (isometry inversion), $p_{\mathrm{succ}}$ coincides with the probability to obtain the successful measurement outcome when the input quantum state is in the image $\Im \map{V}$, since $\widetilde{\mathcal{V}}_{\mathrm{inv}}\circ \widetilde{\mathcal{V}}$ is trace preserving. Then, the success probability $p_{\mathrm{succ}}$ for isometry inversion represents the probability to obtain the quantum state $\rho$ when the input state is $\widetilde{\mathcal{V}}(\rho)$.

We summarize a list of $ \widetilde{f(V)}$ discussed in this paper.  Isometry inversion (\ref{item:isometry_inversion}) is the main topic of this work, and other tasks (\ref{item:isometry_cc}, \ref{item:isometry_pcc}, \ref{item:isometry_transposition}) are analyzed to compare isometry inversion with the previous works \cite{Quintino2019Probabilistic, Quintino2019Reversing} on unitary inversion.

\begin{enumerate}
    \item {\it Isometry inversion}:\label{item:isometry_inversion} $ \widetilde{f(V)}=\widetilde{\mathcal{V}}_{\mathrm{inv}},$\\
    such that $\widetilde{\mathcal{V}}_{\mathrm{inv}}\circ \widetilde{\mathcal{V}}= \widetilde{\1}_{d}$.
    \item {\it Isometry complex conjugation}:\label{item:isometry_cc} $ \widetilde{f(V)}=\widetilde{\mathcal{V}}^*,$\\
    where $\widetilde{\mathcal{V}}^*(\rho)=V^*\rho (V^*)^\dagger$. Here, $V^*$ denotes the complex conjugate of $V$ in  the computational basis.
     \item {\it Isometry pseudo complex conjugation} (see Section~\ref{subsec:isometry_pseudo_cc}): \label{item:isometry_pcc} $ \widetilde{f(V)}=\widetilde{\mathcal{V}}_{\mathrm{p cc}},$\\
    such that $\widetilde{\mathcal{V}}_{\mathrm{p cc}}^T\circ \widetilde{\mathcal{V}}= \widetilde{\1}_{d}$, where the transposed map $\widetilde{\Lambda}^T$ for a CP map $\widetilde{\Lambda}$ given by its action as $\widetilde{\Lambda}(\rho) = \sum_{k} K_k \rho K_k^\dagger$ in terms of the Kraus operators $\{K_k\}$ is defined as $\widetilde{\Lambda}^T(\rho) \coloneqq \sum_{k}K_k^T \rho (K_k^T)^\dagger$.
    \item {\it Isometry transposition}:\label{item:isometry_transposition} $ \widetilde{f(V)}=\widetilde{\mathcal{V}}^T,$\\
    where $\widetilde{\mathcal{V}}^{T}(\rho)=V^T\rho (V^T)^\dagger$. Here, $V^T$ denotes the transpose of $V$ in  the computational basis.
\end{enumerate}

In quantum circuits shown in the figures in the rest of the paper, we sometimes write $V \in  \mathbb{V}_\mathrm{iso} (d,D)$ as $V_{d,D}$ to represent the dimensions of its input and output  Hilbert spaces explicitly in quantum circuits.   To illustrate the dimensions of the Hilbert spaces represented by the wires of quantum circuits in the figures, we use the following color coding of wires:  a red wire corresponds to a $d$-dimensional  Hilbert space, a blue wire corresponds to a $D$-dimensional  Hilbert space, and a black wire corresponds to a  Hilbert space with an arbitrary dimension.   The dual lines in the quantum circuits represent classical information transmissions.

\section{The parallel isometry inversion protocol}
\label{sec:isometry_inversion}

\subsection{Main result: Parallel isometry inversion with $D$-independent success probability}

We present our main theorem on the optimal success probability of the probabilistic exact parallel protocol for isometry inversion. The success probability only depends on the dimension of the input Hilbert space of the isometry, and thus significantly outperforms the probabilistic exact parallel protocols based on unitary inversion  (see {Figure~\ref{fig:untiary_embedding_strategy_and_graph}~(b)}).

\begin{Thm}
\label{thm:optimal_isometry_inversion_and_PBT}
The optimal success probability of probabilistic parallel protocols that transform $k$ calls of an isometry operation $\widetilde{\mathcal{V}}: \mathcal{L}(\mathbb{C}^d)\to\mathcal{L}(\mathbb{C}^D)$ into its inverse map $\widetilde{\mathcal{V}}_{\mathrm{inv}}$ does not depend on $D$. Moreover, a parallel protocol shown in Figure~\ref{fig:isometry_inversion_protocols}~(a) achieves a success probability $p_{\mathrm{succ}}=\lfloor k/(d-1) \rfloor/[d^2+\lfloor k/(d-1)\rfloor -1]$, which is optimal for $d=2$.
\end{Thm}

{Before proceeding the proof of Theorem~\ref{thm:optimal_isometry_inversion_and_PBT} and showing the detail of the protocol shown in Figure~\ref{fig:isometry_inversion_protocols}~(a),} we show how an isometry inversion protocol is implemented in the case of $k=d-1$ calls, as the protocol is shown by a quantum circuit represented in Figure~\ref{fig:isometry_inversion_protocols}~(b).
In the quantum circuit, $\ket{A_d}\in \mathcal{I}\otimes \mathcal{F}= (\mathbb{C}^d)^{\otimes d}$ is the totally antisymmetric state defined by
\begin{align}
    \ket{A_d}\coloneqq \sum_{\vec{j}\in\{1, \cdots, d\}^{d}}\frac{\epsilon_{\vec{j}}}{\sqrt{d!}}\ket{j_1j_2\cdots j_d},\label{eq:def_Ad}
\end{align}
where $\{\ket{j_i}\}\;(i=1, \cdots, d)$ is an orthonormal basis of $\mathbb{C}^d$ and $\epsilon_{\vec{j}}$ is the antisymmetric tensor with rank $d$.
The POVM $\mathcal{M}$ is a projective measurement $\{\Pi_1=  \Pi^{\mathrm{a.s.}},  \Pi_0=  I_{\mathcal{P}\mathcal{O}}-\Pi^{\mathrm{a.s.}}\}$, where $\Pi^{\mathrm{a.s.}}$ is the orthogonal projector on $\mathcal{P}\otimes \mathcal{O}_1\otimes \cdots \otimes \mathcal{O}_{d-1}= (\mathbb{C}^D)^{\otimes d}$ onto the subspace spanned by the totally antisymmetric states. This protocol succeeds when the measurement outcome of  $\mathcal{M} = \{\Pi_a \}$ is $a=1$. This quantum circuit implements an isometry inversion protocol as shown in the following theorem. See Appendix \ref{sec:appendix_isometry_inverse} for the proof.

\begin{Thm}
\label{theorem:isometry_inverse}
A parallel protocol shown in Figure~\ref{fig:isometry_inversion_protocols}~(b) transforms $d-1$ calls of an isometry operation {$\widetilde{\mathcal{V}}$} corresponding to $V \in \mathbb{V}_\mathrm{iso} (d, D)$ into its inverse map $\widetilde{\mathcal{V}}_{\mathrm{inv}}$ with a success probability $p_{\mathrm{succ}}=1/d^2$. Moreover, this protocol implements {the inverse map} $\widetilde{\mathcal{V}}_{\mathrm{inv}}=\widetilde{\mathcal{V}}^\prime$ {given by}
\begin{align}
    \widetilde{\mathcal{V}}^\prime(\rho_\mathrm{in})\coloneqq V^{\dagger} \rho_{\mathrm{in}} V+ \1_d 
    \Tr
    \left[
    \Pi_{(\Im V)^{\perp}}\rho_\mathrm{in}
    \right]
    ,\label{eq:isometry_inverse}
\end{align}
where $\Pi_{(\Im V)^{\perp}}$ is the orthogonal projector onto the orthogonal complement $(\Im V)^{\perp}$ of  $\Im V \subset \mathbb{C}^D$.
\end{Thm}

\begin{figure}[th]
    \centering
    \includegraphics[width=\linewidth]{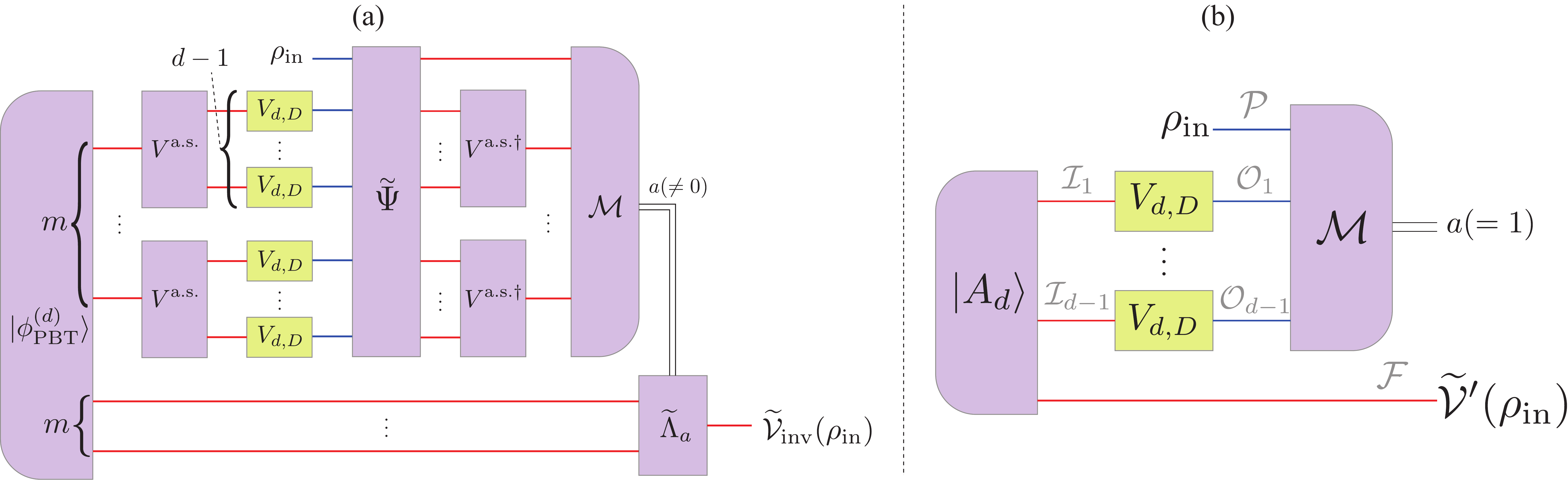}
    \caption{(a) A quantum circuit representation of a parallel delayed input-state protocol for isometry inversion that achieves the success probability $p_{\mathrm{succ}}=\lfloor k/(d-1) \rfloor/[d^2+\lfloor k/(d-1)\rfloor -1]$.   In the quantum circuit, $m(d-1)$ calls of the input isometry operation $\map{V}_{d,D}$ for $m=\lfloor k/(d-1) \rfloor$ are used, and the rest $k-m(d-1)$ calls of $\map{V}_{d,D}$ are discarded (not shown in the figure). The CPTP map $\widetilde{\Psi}$ is defined in Eq.~(\ref{eq:def_Psi}). The quantum state $ \ket{\phi^{(d)}_{\mathrm{PBT}}}$ and the POVM $\mathcal{M}=\{\Gamma_a^{(d)}\}_{a=0}^{k}$ are the optimal resource state and the POVM for the probabilistic port-based teleportation \cite{Ishizaka2008Asymptotica,Studzinski2017Portbased}, which are defined in Eqs.~(\ref{eq:def_phi_PBT}) and (\ref{eq:def_gamma_a^d}), respectively. The conditional CPTP map $\widetilde{\Lambda}_a$ is {the} post-processing operation used in port-based teleportation defined in Eq.~(\ref{eq:def_V_a}), which selects the quantum state in $\mathcal{A}_a$ corresponding to the measurement outcome $a$ of $\mathcal{M}$ as the output state for $a\neq 0$. The isometry operator $V^{\mathrm{a.s.}}$ represents an encoding of quantum information on a totally antisymmetric state defined in Eq.~(\ref{eq:def_V^as}). This protocol succeeds when the measurement outcome $a$ is $a\neq 0$. (b) {The isometry inversion protocol shown in (a) reduces to the circuit shown in this figure for $m=1$, i.e., $k=d-1$.} The quantum state $\ket{A_d}$ is the totally antisymmetric state defined in Eq.~(\ref{eq:def_Ad}) and the POVM $\mathcal{M}$ is a projective measurement  $\{\Pi_1= \Pi^{\mathrm{a.s.}}, \Pi_0=\1_{\mathcal{P}\mathcal{O}}-\Pi^{\mathrm{a.s.}}\}$, where $\Pi^{\mathrm{a.s.}}$ is the orthogonal projector on $\mathcal{P}\otimes \mathcal{O}_1\otimes \cdots \otimes \mathcal{O}_{d-1}= (\mathbb{C}^D)^{\otimes d}$ onto the subspace spanned by the totally antisymmetric states.
    This protocol {implements the inverse map $\widetilde{\mathcal{V}}_{\mathrm{inv}}=\widetilde{\mathcal{V}}'$ of $V$ defined in Eq.~(\ref{eq:isometry_inverse})} when the measurement outcome of $\mathcal{M} = \{\Pi_a \}$ is $a=1$.}
    \label{fig:isometry_inversion_protocols}
\end{figure}

Theorem~\ref{theorem:isometry_inverse} states the existence of a parallel isometry inversion protocol whose success probability does not depend on $D$ for $k=d-1$. Now, we go back to Theorem~\ref{thm:optimal_isometry_inversion_and_PBT}, which is a generalization of Theorem~\ref{theorem:isometry_inverse}. The first part of Theorem~\ref{thm:optimal_isometry_inversion_and_PBT} states the $D$-independence of the optimal success probability of parallel isometry inversion.
 In general, an isometry operation $V\in \mathbb{V}_\mathrm{iso}(d,D')$ can be embedded in $\mathbb{V}_\mathrm{iso}(d,D)$ for $D\geq D'\geq d$. Using this embedding, we can show the optimal success probability of parallel isometry inversion using $k$ calls of the input isometry operation $V\in \mathbb{V}_\mathrm{iso} (d,D)$, denoted by $p_\mathrm{opt}(d,D,k)$, is monotonically non-increasing in the output dimension $D$ of $V$, i.e.,
\begin{align}
    p_\mathrm{opt}(d,D',k) \geq p_\mathrm{opt}(d,D,k)\label{eq:opt_succ_monotonic}
\end{align}
holds for all $d$, $k$, $D$, and $D'$ such that $D\geq D'\geq d$.  We present a proof of this statement for $D'=d$ in the proof of Theorem~\ref{thm:optimal_isometry_inversion_and_PBT} (see Figure~\ref{fig:untiary_inversion_to_isometry_inversion}~(b)) and a similar construction is possible for any $D\geq D'\geq d$.  Thus, it is enough to show the converse of Eq.~(\ref{eq:opt_succ_monotonic}) given by
\begin{align}
    p_\mathrm{opt}(d,d,k) \leq p_\mathrm{opt}(d,D,k).
\end{align}
We show this relation by
constructing a CPTP map compressing the output space of isometry operations to $d$-dimensional spaces. Using this CPTP map, we show  the construction of a parallel isometry inversion protocol of $V\in \mathbb{V}_\mathrm{iso}(d,D)$ from a corresponding parallel unitary inversion protocol of $U\in \mathbb{U}(d)(=\mathbb{V}_{\mathrm{iso}}(d,d))$. The second part shows the series of parallel isometry inversion protocols with $D$-independent success probability approaching 1 for $k\to \infty$.  This asymptotic behavior is natural since one can perform the process tomography of $V\in\mathbb{V}_{\text{iso}}(d,D)$ to construct the inverse operation deterministically for $k\to \infty$. We note that the optimality of the success probability shown in Theorem~\ref{thm:optimal_isometry_inversion_and_PBT} is still open since the optimal parallel protocol for unitary inversion is not analytically known except for the case of $d=2$.

 \subsection{Construction of isometry inversion protocol from a given parallel delayed-input state protocol for unitary inversion}

We consider a subclass of parallel protocols to proceed the proof of Theorem~\ref{thm:optimal_isometry_inversion_and_PBT}.
We say that a superinstrument is realized by a parallel delayed input-state protocol if an input state is inserted after applying black box operations similarly to the one given by the quantum circuit shown in Figure~\ref{fig:isometry_inversion_protocols}~(a) (see Ref.~\cite{Quintino2019Probabilistic} for the detail).
 As shown in the following Lemma,  an isometry inversion protocol can be constructed from a given parallel delayed-input state protocol for unitary inversion.

Since the input quantum state in the Hilbert space $\mathcal{P}$ and the output state of $\map{V}_{d,D}^{\otimes k}$ in the Hilbert space $\mathcal{O}=\bigotimes_{i=1}^{k} \mathcal{O}_i$ are available at some point of the parallel delayed input-state protocol, we can apply a CPTP map $\widetilde{\Psi}$ on $\mathcal{P}\otimes \mathcal{O}$.  Using this idea, we construct an isometry inversion protocol (see the proof of Lemma \ref{lem:isometry_inversion_construction}).  Note that this construction is not valid for a general (i.e., non delayed input-state) parallel protocol since the input quantum state in the Hilbert space $\mathcal{P}$ is not available after applying the input isometry operation (see Figure~\ref{fig:parallel_comb}).

\begin{Lem}
\label{lem:isometry_inversion_construction}
Suppose we are given a parallel delayed input-state protocol that transforms $k$ calls of a unitary operation $\widetilde{\mathcal{U}}: \mathcal{L}(\mathbb{C}^d)\to\mathcal{L}(\mathbb{C}^d)$ into its inverse map $\widetilde{\mathcal{U}}^{\dagger}$ with a success probability ${p'_{\mathrm{succ}}}$. Then, we can construct a parallel delayed input-state protocol that transforms $k$ calls of an isometry operation $\widetilde{\mathcal{V}}: \mathcal{L}(\mathbb{C}^d)\to\mathcal{L}(\mathbb{C}^D)$ into its inverse map $\widetilde{\mathcal{V}}_{\mathrm{inv}}$ with the same success probability $p_{\mathrm{succ}}={p'_{\mathrm{succ}}}$.
\end{Lem}
\begin{proof}
 We construct a CPTP map $\widetilde{\Psi}$ transforming $k+1$ parallel calls of isometry $V\in\mathbb{V}_{\text{iso}}(d,D)$ into parallel calls of  a random unitary $U\in\mathbb{U}(d)$ (see Eq.~(\ref{eq:Psi_U})). This CPTP map $\widetilde{\Psi}$ compresses the output state of isometry operations on $d$-dimensional systems keeping the relevant component to retrieve the input state, which contributes  to the $D$- independence of the success probability. Then, we construct an isometry inversion protocol by inserting the CPTP map $\widetilde{\Psi}$ into a given unitary inversion protocol (see  Figure~\ref{fig:untiary_inversion_to_isometry_inversion}~(a)).  
In the following, we derive $\widetilde{\Psi}$ and show the protocol and the achievability of $p_{\mathrm{succ}}$.

We  first introduce a parallel delayed input-state $k$-input superinstrument $\{\doublewidetilde{\mathcal{S}'}, \doublewidetilde{\mathcal{F}'}\}: [\mathcal{L}(\mathcal{I})\to \mathcal{L}(\mathcal{O}')] \to [\mathcal{L}(\mathcal{P}')\to \mathcal{L}(\mathcal{F})]$ defined by
\begin{align}
    \doublewidetilde{\mathcal{S}'}(\widetilde{\Lambda}_{\mathrm{in}})(\cdot)&=
    \left[
    \widetilde{\mathcal{D}}'_{S}\circ \left(
    \widetilde{\1}_{\mathcal{P}'}\otimes \widetilde{\Lambda}_{\mathrm{in}}\otimes \widetilde{\1}_{\mathcal{A}}
    \right)
    \right]
    (\cdot \otimes {\phi'}_{\mathcal{IA}}),\label{eq:def_S'}\\
    \doublewidetilde{\mathcal{F}'}(\widetilde{\Lambda}_{\mathrm{in}})(\cdot)&=
    \left[
    {\widetilde{\mathcal{D}}'_{F}}\circ \left(
    \widetilde{\1}_{\mathcal{P}'}\otimes \widetilde{\Lambda}_{\mathrm{in}}\otimes \widetilde{\1}_{\mathcal{A}}
    \right)
    \right]
    (\cdot \otimes {\phi'}_{\mathcal{IA}}),\label{eq:def_F'}
\end{align}
where the joint Hilbert space are defined by $\mathcal{I}\coloneqq \bigotimes_{i=1}^{k}\mathcal{I}_i$ and $\mathcal{O}'\coloneqq \bigotimes_{i=1}^{k}\mathcal{O}'_i$, $\mathcal{A}$ is an auxiliary system, ${\phi'}_{\mathcal{IA}}\in \mathcal{L}(\mathcal{I}\otimes \mathcal{A})$ is a quantum state and $\{{{\widetilde{\mathcal{D}}'_S, \widetilde{\mathcal{D}}'_F}}\}: \mathcal{L}(\mathcal{P}'\otimes \mathcal{O}'\otimes \mathcal{A})\to \mathcal{L}(\mathcal{F})$ is a quantum instrument (see the left panel of Figure~\ref{fig:untiary_inversion_to_isometry_inversion}~(a)).  Note that a prime is added on the quantum state $\phi'$, the quantum instrument $\{{{\widetilde{\mathcal{D}}'_S, \widetilde{\mathcal{D}}'_F}}\}$, and the Hilbert spaces $\mathcal{O}_i'$ and $\mathcal{P}'$ to distinguish $\{\doublewidetilde{\mathcal{S}'}, \doublewidetilde{\mathcal{F}'}\}$ from another superinstrument $\{\doublewidetilde{\mathcal{S}^{}}, \doublewidetilde{\mathcal{F}^{}}\}$ defined later (see Eqs.~(\ref{eq:def_S}) and (\ref{eq:def_F})). Suppose  $\mathcal{P}'=\mathbb{C}^d$, $\mathcal{F}=\mathbb{C}^d$, $\mathcal{I}_i=\mathbb{C}^d$ and $\mathcal{O}'_i=\mathbb{C}^d$ for $i\in \{1, \cdots, k\}$.  We assume that the superinstrument  $\{\doublewidetilde{\mathcal{S}^\prime}, \doublewidetilde{\mathcal{F}^\prime}\}$ implements unitary inversion with a success probability ${p'_{\mathrm{succ}}}$, i.e.,
\begin{align}
    \doublewidetilde{\mathcal{S}'}(\widetilde{\mathcal{U}}^{\otimes k})={p'_{\mathrm{succ}}}\;\widetilde{\mathcal{U}}^{\dagger}\;\;\;(\forall U\in \mathbb{U}(d)).\label{eq:S'_unitary_inversion}
\end{align}

Next, we consider the condition that a parallel delayed input-state superinstrument $\{\doublewidetilde{\mathcal{S}^{}}, \doublewidetilde{\mathcal{F}^{}}\}$ implements isometry inversion.
Suppose $\mathcal{P}= \mathbb{C}^d$ and $\mathcal{O}_i= \mathbb{C}^{D}$ for $i\in\{1, \cdots, k\}$ and define the joint Hilbert space by $\mathcal{O}\coloneqq \bigotimes_{i=1}^{k}\mathcal{O}_i$. {We consider a parallel delayed input-state superinstrument $\{\doublewidetilde{\mathcal{S}^{}}, \doublewidetilde{\mathcal{F}^{}}\}:[\mathcal{L}(\mathcal{I})\to\mathcal{L}(\mathcal{O})]\to [\mathcal{L}(\mathcal{P})\to\mathcal{L}(\mathcal{F})]$ given by
\begin{align}
    \doublewidetilde{\mathcal{S}^{}}(\widetilde{\Lambda}_{\mathrm{in}})(\cdot)&=
    \left[
    {\widetilde{\mathcal{D}}_{S}}\circ \left(
    \widetilde{\1}_{\mathcal{P}}\otimes \widetilde{\Lambda}_{\mathrm{in}}\otimes \widetilde{\1}_{\mathcal{A}}
    \right)
    \right]
    (\cdot \otimes \phi_{\mathcal{IA}}),\label{eq:def_S}\\
    \doublewidetilde{\mathcal{F}^{}}(\widetilde{\Lambda}_{\mathrm{in}})(\cdot)&=
    \left[
    {\widetilde{\mathcal{D}}_{F}}\circ \left(
    \widetilde{\1}_{\mathcal{P}}\otimes \widetilde{\Lambda}_{\mathrm{in}}\otimes \widetilde{\1}_{\mathcal{A}}
    \right)
    \right]
    (\cdot \otimes \phi_{\mathcal{IA}}),\label{eq:def_F}
\end{align}
where $\mathcal{A}$ is an auxiliary system, $\phi_{\mathcal{IA}}\in \mathcal{L}(\mathcal{I}\otimes \mathcal{A})$ is a quantum state and $\{{{\widetilde{\mathcal{D}}_S, \widetilde{\mathcal{D}}_F}}\}: \mathcal{L}(\mathcal{P}\otimes \mathcal{O}\otimes \mathcal{A})\to \mathcal{L}(\mathcal{F})$ is a quantum instrument (see the left panel of Figure~\ref{fig:isometry_neutralization}). The condition that superinstrument $\{\doublewidetilde{\mathcal{S}^{}}, \doublewidetilde{\mathcal{F}^{}}\}$ implements isometry inversion with the success probability $p_{\mathrm{succ}}$ is given by
\begin{align}
    \doublewidetilde{\mathcal{S}^{}}(\widetilde{\mathcal{V}})=p_{\mathrm{succ}}\widetilde{\mathcal{V}}_{\mathrm{inv}}\;\;\;(\forall V\in \mathbb{V}_{\mathrm{iso}}(d,D)).
\end{align}
By definition of the inverse map $\widetilde{\mathcal{V}}_{\mathrm{inv}}$, this condition is equivalent to the condition given by
\begin{align}
    \doublewidetilde{\mathcal{S}^{}}(\widetilde{\mathcal{V}})\circ \widetilde{\mathcal{V}}=p_{\mathrm{succ}}\widetilde{\1}_d,
\end{align}
i.e.,
\begin{align}
    \left[
    {\widetilde{\mathcal{D}}_{S}}\circ \left(
    \widetilde{\mathcal{V}}^{\otimes k+1}_{\mathcal{P}''\mathcal{I}\to \mathcal{P}\mathcal{O}} \otimes \widetilde{\1}_{\mathcal{A}}
    \right)
    \right]
    (\rho_{\mathcal{P}''} \otimes \phi_{\mathcal{IA}})
    =p_{\mathrm{succ}}\rho_{\mathcal{F}}
\end{align}
for all $V\in \mathbb{V}_{\mathrm{iso}}(d,D)$ and $\rho\in \mathcal{L}(\mathcal{P}'')$ (see the right panel of Figure~\ref{fig:isometry_neutralization}). Here, the Hilbert space $\mathcal{P}''$ is given by $\mathcal{P}''=\mathbb{C}^d$. 
This condition means that we have to retrieve the quantum state $\rho$ after applying the tensor product of an isometry  operation} $ \map{V}^{\otimes k+1}$.

\begin{figure}[t]
    \centering
    \includegraphics[width=\linewidth]{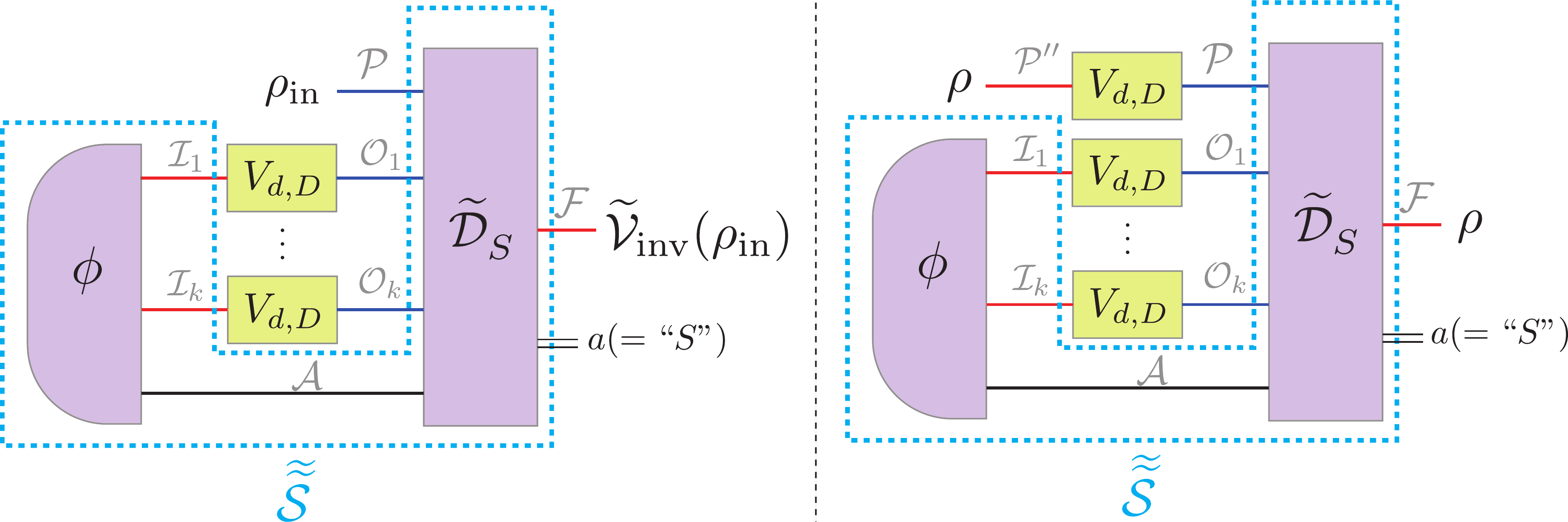}
    \caption{Left panel: A quantum circuit representation of a parallel delayed input-state protocol $\{\doublewidetilde{\mathcal{S}^{}}, \doublewidetilde{\mathcal{F}^{}}\}$.  This protocol succeeds when the successful measurement outcome is obtained ($a=\text{``}S\text{''}$). Right panel: The condition that the parallel delayed input-state protocol $\{\doublewidetilde{\mathcal{S}^{}}, \doublewidetilde{\mathcal{F}^{}}\}$ shown in the left panel implements isometry inversion with the success probability $p_{\mathrm{succ}}$ is equivalent to the condition that the output state of the quantum circuit shown in the right panel retrieves the input state $\rho$ with the success probability $p_{\mathrm{succ}}$.   This protocol succeeds when the successful measurement outcome is obtained ($a=\text{``}S\text{''}$).}
    \label{fig:isometry_neutralization}
\end{figure}

We show the decomposition of the tensor product {$V^{\otimes k+1}$} of an isometry {operator $V$} to consider how to retrieve the quantum state $\rho$ after the application of $V^{\otimes k+1}$.
The joint Hilbert spaces {$\mathcal{P}''\otimes \mathcal{I}=(\mathbb{C}^d)^{\otimes k+1}$,} $\mathcal{P}'\otimes \mathcal{O}' =(\mathbb{C}^d)^{\otimes k+1}$ and $\mathcal{P}\otimes \mathcal{O} =(\mathbb{C}^D)^{\otimes k+1}$ can be decomposed as
\begin{align}
    {\mathcal{P}''\otimes \mathcal{I}}&= {\bigoplus_{\mu\vdash k+1} \mathcal{U}_{\mu, \mathcal{P}''\mathcal{I}}^{(d)}\otimes \mathcal{S}_{\mu, \mathcal{P}''\mathcal{I}}^{(k+1)},}\label{eq:sw_p'o'}\\
    \mathcal{P}'\otimes \mathcal{O}'&= \bigoplus_{\mu\vdash k+1} \mathcal{U}_{\mu, \mathcal{P}'\mathcal{O}'}^{(d)}\otimes \mathcal{S}_{\mu, \mathcal{P}'\mathcal{O}'}^{(k+1)},\\
    \mathcal{P}\otimes \mathcal{O}&= \bigoplus_{\mu\vdash k+1} \mathcal{U}_{\mu, \mathcal{P}\mathcal{O}}^{(D)}\otimes \mathcal{S}_{\mu, \mathcal{P}\mathcal{O}}^{(k+1)}.\label{eq:sw_po}
\end{align}
by the Schur-Weyl duality , where $\mu$ is a Young diagram with $k+1$ boxes, denoted by $\mu\vdash k+1$ (see Appendix \ref{sec:schur-weyl_duality} for the detail).
Let $\{\ket{\mu, u_\mu, s_\mu}\}$ be the orthonormal basis of $\mathcal{U}_{\mu, \mathcal{P}'\mathcal{O}'}^{(d)}\otimes \mathcal{S}_{\mu, \mathcal{P}'\mathcal{O}'}^{(k+1)}$ for each $\mu$ and we call $\{\ket{\mu, u_\mu, s_\mu}\}$ the Schur basis (see also Appendix~\ref{sec:schur-weyl_duality}). We define the change of basis $U^{\mathrm{Sch}}_{\mathcal{P}'\mathcal{O}'}$ transforming the computational basis for $\mathcal{P}'\otimes \mathcal{O}'$ to the Schur basis, and suppose $\ket{\mu}$, $\ket{u_\mu}$ and $\ket{s_\mu}$ are stored in the Hilbert spaces $\mathcal{M}_{\mathcal{P}'\mathcal{O}'}$, $\mathcal{U}_{\mathcal{P}'\mathcal{O}'}$ and $\mathcal{S}_{\mathcal{P}'\mathcal{O}'}$, respectively. The unitary operator $U^{\mathrm{Sch}}_{\mathcal{P}'\mathcal{O}'}$ is called the quantum Schur transform \cite{harrow2005applications, bacon2006efficient, Krovi2019Efficient}. Using the decomposition of Hilbert spaces given by Eqs.~(\ref{eq:sw_p'o'}) and (\ref{eq:sw_po}), we can decompose the tensor product $V^{\otimes k+1}: \mathcal{P}''\otimes \mathcal{I}\to \mathcal{P}\otimes \mathcal{O}$ of an isometry operator $V\in \mathbb{V}_{\mathrm{iso}}(d,D)$ as
\begin{align}
    V^{\otimes k+1}=\bigoplus_{\substack{\mu\vdash k+1\\l(\mu)\leq d}} V_{\mu}\otimes I_{\mathcal{S}_{\mu, \mathcal{P}''\mathcal{I}}^{(k+1)}\to \mathcal{S}_{\mu, \mathcal{P}\mathcal{O}}^{(k+1)}},
\end{align}
where $l(\mu)$ is the number of rows of a Young diagram $\mu$, $V_{\mu}\in \mathcal{L}(\mathcal{U}_{\mu, \mathcal{P}''\mathcal{I}}^{(d)}\to \mathcal{U}_{\mu, \mathcal{P}\mathcal{O}}^{(D)})$ is an isometry operator depending on the isometry $V$ and $I_{\mathcal{S}_{\mu, \mathcal{P}''\mathcal{I}}^{(k+1)}\to \mathcal{S}_{\mu, \mathcal{PO}}^{(k+1)}}$ is the isomorphism between irreducible representations $\mathcal{S}_{\mu, \mathcal{P}'\mathcal{O}'}^{(k+1)}$ and $\mathcal{S}_{\mu, \mathcal{PO}}^{(k+1)}$. We prove this decomposition in Appendix~\ref{sec:schur-weyl_duality} (see Eq.~(\ref{eq:isometry_schur_weyl})). This decomposition shows that the quantum information encoded in $\mathcal{S}^{(k+1)}_{\mu, \mathcal{P}''\mathcal{I}}$ is unchanged by the action of the parallel calls of any isometry operator $V\in \mathbb{V}_{\mathrm{iso}}(d,D)$, while the quantum information encoded in $\mathcal{U}_{\mu, \mathcal{P}''\mathcal{I}}^{(d)}$ is affected by a $V$-dependent action. In addition, the Hilbert space $\mathcal{U}_{\mu, \mathcal{P}\mathcal{O}}^{(D)}\otimes \mathcal{S}_{\mu, \mathcal{P}\mathcal{O}}^{(k+1)}$ for $l(\mu)>d$ is out of the image $\Im V^{\otimes k+1}$. From these observations, the quantum information relevant to retrieve the quantum state $\rho$ after the application of  $\map{V}^{\otimes k+1}$ is considered to be encoded in the Hilbert space $\mathcal{S}_{\mu, \mathcal{P}\mathcal{O}}^{(k+1)}$ for $l(\mu)\leq d$.  Then, we can ``compress'' the output state of  $\map{V}^{\otimes k+1}$ to a smaller Hilbert space whose dimension is independent of $D$ by a CPTP map $\widetilde{\Psi}$ defined in the next paragraph.

We define a key element for constructing an isometry inversion protocol, a CPTP map $\widetilde{\Psi}: \mathcal{L}(\mathcal{P}\otimes \mathcal{O})\to \mathcal{L}(\mathcal{P}'\otimes \mathcal{O}')$, by
\begin{align}
    &\widetilde{\Psi}_{\mathcal{P}\mathcal{O}\to\mathcal{P}'\mathcal{O}'}(\rho)\nonumber\\
    &\coloneqq \bigoplus_{\substack{\mu\vdash k+1\\l(\mu)\leq d}}\frac{\1_{\mathcal{U}^{(d)}_{\mu, \mathcal{P}'\mathcal{O}'}}}{d_{\mathcal{U}_{\mu}^{(d)}}}\otimes 
    \left[
    \widetilde{\mathcal{I}}_{{\mathcal{S}^{(k+1)}_{\mu, \mathcal{PO}}\to \mathcal{S}^{(k+1)}_{\mu, \mathcal{P}'\mathcal{O}'}}}\Tr_{\mathcal{U}_{\mu, \mathcal{P}\mathcal{O}}^{(D)}}(\Pi_{\mu, \mathcal{PO}}\rho)
    \right]
    +\frac{\1_{\mathcal{P}'\mathcal{O}'}}{d_{\mathcal{P}'\mathcal{O}'}}\times \sum_{\substack{\mu\vdash k+1\\l(\mu)>d}}\Tr(\Pi_{\mu, \mathcal{PO}}\rho),\label{eq:def_Psi}
\end{align}
where $\Pi_{\mu, \mathcal{PO}}$ is a projector on the Hilbert space $\mathcal{P}\otimes \mathcal{O}$ onto its subspace $\mathcal{U}_{\mu, \mathcal{P}\mathcal{O}}^{(D)}\otimes \mathcal{S}_{\mu, \mathcal{P}\mathcal{O}}^{(k+1)}$. 
The CPTP map $\widetilde{\Psi}$ ``extracts'' the quantum information encoded in $\mathcal{S}^{(k+1)}_{\mu, \mathcal{P}\mathcal{O}}$ for $l(\mu)\leq d$ from $\mathcal{P}\otimes \mathcal{O}$ by ``discarding'' $\mathcal{U}_{\mu, \mathcal{P}\mathcal{O}}^{(D)}$ and ``embeds'' it onto $\mathcal{P}'\otimes \mathcal{O}'$. More precisely, the CPTP map $\widetilde{\Psi}$ can be implemented by using the quantum Schur transform and the ``measure-and-prepare'' strategy \cite{horodecki2003entanglement, Bisio2010Optimal} as follows (see Figure~\ref{fig:Psi_construction}). First, we apply the quantum Schur transform $U^{\mathrm{Sch}}_{\mathcal{P}\mathcal{O}}$ on $\mathcal{P}\otimes \mathcal{O}$ and measure $\ket{\mu}$. If $l(\mu)\leq d$, we replace $\ket{u_\mu}$ by the $d_{\mathcal{U}_{\mu}^{(d)}}$-dimensional maximally mixed state and apply the inverse of the quantum Schur transform $U^{\mathrm{Sch}\dagger}_{\mathcal{P}'\mathcal{O}'}$ on $\mathcal{P}'\otimes \mathcal{O}'$ to obtain the output state. Otherwise, we replace the entire quantum state by the $d_{\mathcal{P}'\mathcal{O}'}$-dimensional maximally mixed state to obtain the output state.

\begin{figure}[tb]
    \centering
    \includegraphics[width=0.8\linewidth]{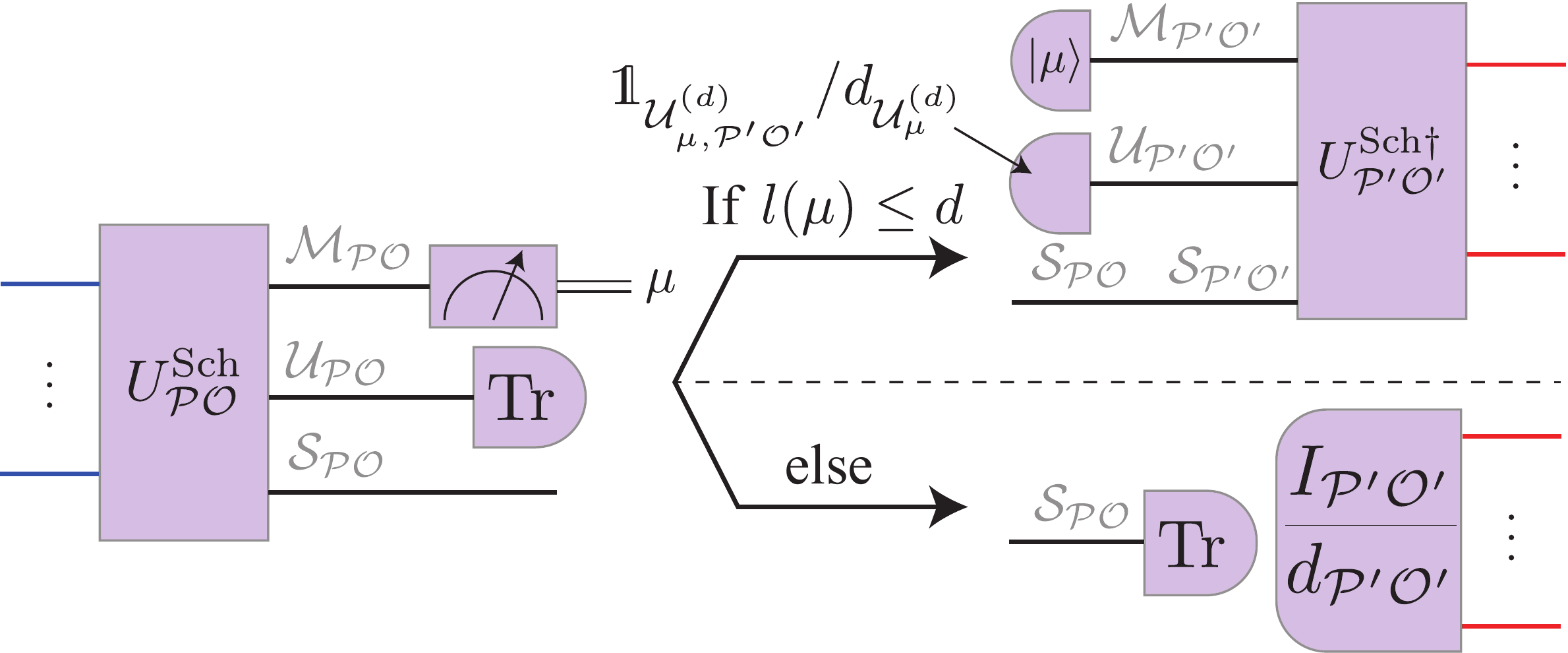}
    \caption{{Implementation of the CPTP map $\widetilde{\Psi}$ by using the quantum Schur transform \cite{harrow2005applications, bacon2006efficient, Krovi2019Efficient} and the ``measure-and-prepare'' strategy {\cite{horodecki2003entanglement, Bisio2010Optimal}}. The projective measurement is applied on the Hilbert space $\mathcal{M}_{\mathcal{PO}}$ after the quantum Schur transform. Depending on the measurement outcome $\mu$, we replace a quantum state by a fixed quantum state and apply the inverse of quantum Schur transform to obtain the output state. The wire between $\mathcal{S}_{\mathcal{PO}}$ and $\mathcal{S}_{\mathcal{P}'\mathcal{O}'}$ represents the identity map.}}
    \label{fig:Psi_construction}
\end{figure}

The CPTP map $\widetilde{\Psi}$ satisfies the following lemma, which will play a crucial role to construct a parallel isometry inversion protocol.

\begin{Lem}
\label{lem:Lambda}
For any isometry operation $\widetilde{\mathcal{V}}$ corresponding to $V\in \mathbb{V}_\mathrm{iso} (d,D)$,
\begin{align}
    \widetilde{\Psi}\circ \widetilde{\mathcal{V}}^{\otimes k+1}=\int \dd U \widetilde{\mathcal{U}}^{\otimes k+1}\label{eq:Psi_U}
\end{align}
holds, where $\widetilde{\mathcal{U}}$ is a unitary operation corresponding to $U \in \mathbb{U}(d) $, and $\dd U$ is the Haar measure on  $\mathbb{U}(d)$.
\end{Lem}

{This Lemma~shows the transformation from $k+1$ parallel calls of an isometry {operation $\widetilde{\mathcal{V}}$ corresponding to} $V\in \mathbb{V}_{\mathrm{iso}}(d,D)$ into $k+1$ parallel calls of a $d$-dimensional randomly (and independently of the isometry {operator} $V$) chosen unitary {operation}. {Note that, for $D=d$, the action of the CPTP map $\widetilde{\Psi}$ matches parallel calls of a randomly chosen unitary operation, i.e.,
\begin{align}
    \widetilde{\Psi}=\int \dd U \widetilde{\mathcal{U}}^{\otimes k+1}.
\end{align}
Then, Lemma~\ref{lem:Lambda} reduces to a well-known relation for the Haar measure $\dd U$ on the unitary group $\mathbb{U}(d)$ given by
\begin{align}
    \left( \int \dd U \widetilde{\mathcal{U}}^{\otimes k+1}\right)\circ \widetilde{\mathcal{U}}'^{\otimes k+1}=\int \dd U \widetilde{\mathcal{U}}^{\otimes k+1}
\end{align}
for all $U'\in \mathbb{U}(d)$.} In a sense, the CPTP map $\widetilde{\Psi}$ extends this relation to an arbitrary $D$ while keeping the independence on $D$ on the right hand side, which may be of independent interest.
}See Appendix \ref{sec:appendix_Lambda} for the proof of Lemma~\ref{lem:Lambda}.

\begin{figure*}[th]
    \centering
    \includegraphics[width=\linewidth]{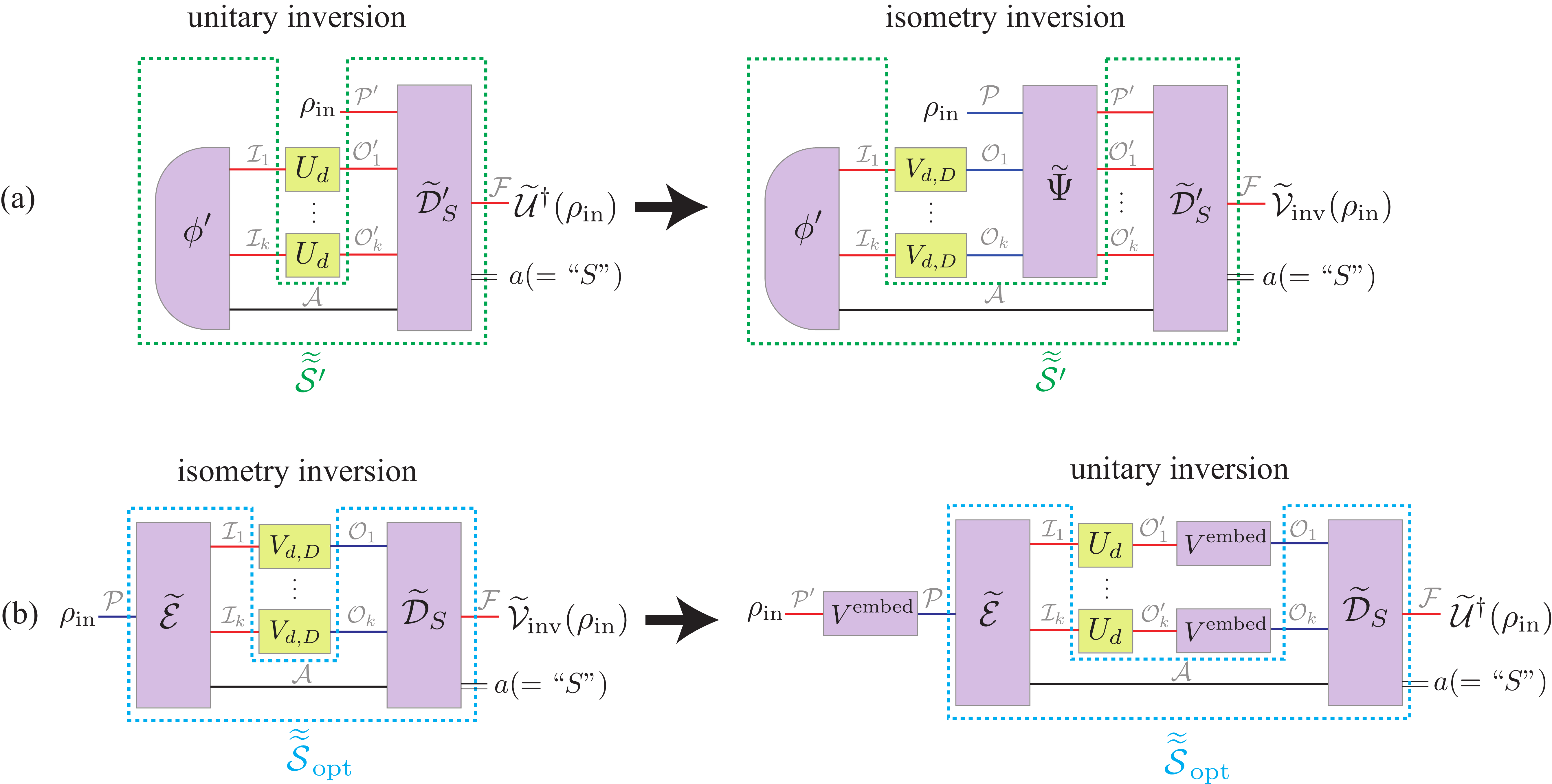}
    \caption{(a) Left panel: Parallel delayed input-state protocol $\doublewidetilde{\mathcal{S}'}$ for unitary inversion of $U_d\in \mathbb{U}(d)$ using a quantum state $\phi'$ and a quantum instrument $\{\widetilde{\mathcal{D}}'_S, \widetilde{\mathcal{D}}'_F\}$ (see also Eq.~(\ref{eq:def_S'})).
    Right panel: Construction of a parallel protocol for isometry inversion of $V_{d, D}\in \mathbb{V}_{\mathrm{iso}}(d, D)$ using a parallel delayed input-state protocol $\doublewidetilde{\mathcal{S}'}$ for unitary inversion of $U_d\in \mathbb{U}(d)$ and the CPTP map $\widetilde{\Psi}$ defined in Eq.~(\ref{eq:def_Psi}).
    This protocol achieves the same success probability as the parallel delayed-input state protocol for unitary inversion of $U_d\in \mathbb{U}(d)$ shown in the left panel. (b) Left panel: The parallel protocol $\doublewidetilde{\mathcal{S}^{}}_{\mathrm{opt}}$ for isometry inversion of $V_{d, D}\in \mathbb{V}_{\mathrm{iso}}(d, D)$ achieving the optimal success probability can be implemented using a CPTP map $\widetilde{\mathcal{E}}$ and a quantum instrument $\{\widetilde{\mathcal{D}}_S, \widetilde{\mathcal{D}}_F\}$ (see also Eq.~(\ref{eq:parallel_isometry_inversion_superinstrument})).
    Right panel: Construction of a parallel protocol for unitary inversion of $U_d\in \mathbb{U}(d)$ using the optimal parallel protocol $\doublewidetilde{\mathcal{S}^{}}_{\mathrm{opt}}$ for isometry inversion of $V_{d, D}\in \mathbb{V}_{\mathrm{iso}}(d,D)$ and the embedding isometry operator $V^{\mathrm{embed}}$ defined in Eq.~(\ref{eq:def_embedding}). This protocol achieves the same success probability as the optimal parallel protocol for isometry inversion shown in the left panel.}
    \label{fig:untiary_inversion_to_isometry_inversion}
\end{figure*}

Using the CPTP map $\widetilde{\Psi}$, we derive a parallel delayed input-state protocol for isometry inversion protocol as follows.  We define a parallel delayed input-state superinstrument $\{\doublewidetilde{\mathcal{S}^{}}, \doublewidetilde{\mathcal{F}^{}} \}: [\mathcal{L}(\mathcal{I})\to \mathcal{L}(\mathcal{O})] \to [\mathcal{L}(\mathcal{P})\to \mathcal{L}(\mathcal{F})]$ by inserting the CPTP map $\widetilde{\Psi}$ to the parallel delayed input-state superinstrument $\{\doublewidetilde{\mathcal{S}^{\prime}}, \doublewidetilde{\mathcal{F}^{\prime}} \}$ for unitary inversion as
\begin{align}
    \doublewidetilde{\mathcal{S}^{}}(\widetilde{\Lambda}_{\mathrm{in}})(\cdot)
    &=
    \left[
    {\widetilde{\mathcal{D}}'_{S}}\circ
    \left(
    \widetilde{\Psi}\otimes \widetilde{\1}_{\mathcal{A}}
    \right)
    \circ 
    \left(
    \widetilde{\1}_{\mathcal{PA}}\otimes \widetilde{\Lambda}_{\mathrm{in}}
    \right)
    \right]
    (\cdot \otimes {\phi'}_{\mathcal{IA}}), \\
    \doublewidetilde{\mathcal{F}^{}}(\widetilde{\Lambda}_{\mathrm{in}})(\cdot)&=
    \left[
    {\widetilde{\mathcal{D}}'_{F}}\circ
    \left(
    \widetilde{\Psi}\otimes \widetilde{\1}_{\mathcal{A}}
    \right)
    \circ 
    \left(
    \widetilde{\1}_{\mathcal{PA}}\otimes \widetilde{\Lambda}_{\mathrm{in}}
    \right)
    \right]
    (\cdot \otimes {\phi'}_{\mathcal{IA}}),
\end{align}
 where the quantum state $\phi'$ and the quantum instrument $\{\widetilde{\mathcal{D}}'_{S}, \widetilde{\mathcal{D}}'_{F}\}$ are in the superinstrument $\{\doublewidetilde{\mathcal{S}'}, \doublewidetilde{\mathcal{F}'}\}$ for unitary inversion (see Eqs.~(\ref{eq:def_S'}), (\ref{eq:def_F'}), and Figure~\ref{fig:untiary_inversion_to_isometry_inversion}~(a)). Note that the CPTP map $\widetilde{\Psi}$ can be inserted since $\{\doublewidetilde{\mathcal{S}'}, \doublewidetilde{\mathcal{F}'}\}$ is a parallel delayed-input state protocol.
If the input state of the parallel delayed-input state protocol $\doublewidetilde{\mathcal{S}^{}}$ is $\rho_{\mathrm{in}}=\widetilde{\mathcal{V}}(\rho)$, the output state is calculated as
\begin{align}
    {\doublewidetilde{\mathcal{S}^{}}(\widetilde{\mathcal{V}}^{\otimes k})(\rho_{\mathrm{in}})}
    &={[}\doublewidetilde{\mathcal{S}^{}}(\widetilde{\mathcal{V}}^{\otimes k})\circ \widetilde{\mathcal{V}}{]}({\rho})\\
    &=
    \left[
    {\widetilde{\mathcal{D}}'_{S}}\circ
    \left(
    \widetilde{\Psi}\circ \widetilde{\mathcal{V}}^{\otimes k+1}\otimes \widetilde{\1}_{\mathcal{A}}
    \right)
    \right]({\rho} \otimes {\phi'}_{\mathcal{IA}})\\
    &=\int \dd U 
    \left[
    {\widetilde{\mathcal{D}}'_{S}}\circ
    \left(\widetilde{\mathcal{U}}^{\otimes k+1}\otimes \widetilde{\1}_{\mathcal{A}}
    \right)
    \right]({\rho} \otimes {\phi'}_{\mathcal{IA}})\\
    &=\int \dd U \left[
    \doublewidetilde{\mathcal{S}'}(\widetilde{\mathcal{U}}^{\otimes k})\circ \widetilde{\mathcal{U}}
    \right]
    ({\rho})\\
    &={p'_{\mathrm{succ}}}\int \dd U 
    \left(
    \widetilde{\mathcal{U}}^{\dagger}\circ \widetilde{\mathcal{U}}
    \right)
    ({\rho})\\
    &={p'_{\mathrm{succ}}}{\rho}\\
    &={p'_{\mathrm{succ}}\widetilde{\mathcal{V}}_{\mathrm{inv}}(\rho_{\mathrm{in}})}.
\end{align}
from Eqs.~(\ref{eq:S'_unitary_inversion}) and (\ref{eq:Psi_U}). Therefore, the quantum superinstrument $\{\doublewidetilde{\mathcal{S}^{}}, \doublewidetilde{\mathcal{F}^{}}\}$ implements an isometry inversion protocol and its success probability is ${p'_{\mathrm{succ}}}$  when the input quantum state $\rho_\mathrm{in}$ is in the image $\mathrm{Im} \map{V}$ of the input isometry operation $\map{V}$.
\end{proof}

\subsection{Proof of the main result}

\begin{proof}[Proof of Theorem~\ref{thm:optimal_isometry_inversion_and_PBT}]
We show the first part of Theorem~\ref{thm:optimal_isometry_inversion_and_PBT}, i.e., the optimal success probability of parallel isometry inversion of $V\in\mathbb{V}_{\text{iso}}(d,D)$ coincides with the optimal success probability of parallel unitary inversion of $U\in \mathbb{U}(d)$. The achievability is shown by constructing an isometry inversion protocol using Lemma~\ref{lem:isometry_inversion_construction}. The optimality is shown by constructing a $d$-dimensional unitary inversion protocol from a given isometry inversion protocol by embedding $\mathbb{U}(d)$ to $\mathbb{V}_{\text{iso}}(d,D)$. The second part of Theorem~\ref{thm:optimal_isometry_inversion_and_PBT}  states the existence of a parallel isometry inversion protocol achieving the success probability $p_{\mathrm{succ}}=\lfloor k/(d-1) \rfloor/[d^2+\lfloor k/(d-1)\rfloor -1]$ and its optimality for $d=2$.  The existence is shown by constructing an isometry inversion protocol from a corresponding parallel delayed input-state protocol for unitary inversion of $U\in\mathbb{U}(d)$ with a success probability $p_{\mathrm{succ}}=\lfloor k/(d-1) \rfloor/[d^2+\lfloor k/(d-1)\rfloor -1]$ using Lemma~\ref{lem:isometry_inversion_construction}. Since this parallel unitary inversion protocol is known to be optimal for $d=2$, this success probability of parallel isometry inversion protocol is optimal for $d=2$ (see Appendix~\ref{sec:appendix_unitary_inversion} and Ref.~\cite{Quintino2019Reversing}). We prove the achievability and the optimality of the success probability of isometry inversion as follows.

 (Achievability)
The optimal success probability  $p'_\mathrm{opt} = p_\mathrm{opt}(d,d,k)$ of parallel unitary inversion  of $U\in \mathbb{U}(d)$ is  shown to be achieved by a delayed input-state protocol in Ref.~\cite{Quintino2019Probabilistic}.   Using the delayed input-state protocol achieving the optimal success probability $p'_\mathrm{opt}$ with Lemma~\ref{lem:isometry_inversion_construction}, we show the existence of a parallel protocol for isometry inversion of $V\in \mathbb{V}_\mathrm{iso} (d,D)$ with  the success probability $p'_\mathrm{opt}$.

 (Optimality)
We show the existence of a parallel protocol for unitary inversion of $U\in \mathbb{U}(d)$ with the same success probability as the optimal parallel protocol for isometry inversion of $V\in \mathbb{V}_\mathrm{iso} (d,D)$ to complete the proof of the first part of Theorem~\ref{thm:optimal_isometry_inversion_and_PBT}.
We consider Hilbert spaces $\mathcal{P}=\mathbb{C}^D$, $\mathcal{F}=\mathbb{C}^d$, $\mathcal{I}_i= \mathbb{C}^d$ and $\mathcal{O}_i=\mathbb{C}^D$ for $i\in \{1, \cdots, k\}$ and the joint Hilbert spaces $\mathcal{I}$ and $\mathcal{O}$ defined as $\mathcal{I}\coloneqq \bigotimes_{i=1}^k \mathcal{I}_i$ and 
$\mathcal{O}\coloneqq \bigotimes_{i=1}^k \mathcal{O}_i$, respectively. Suppose $\{\doublewidetilde{\mathcal{S}^{}}_{\mathrm{opt}}, \doublewidetilde{\mathcal{F}^{}}_{\mathrm{opt}}\}:[\mathcal{L}(\mathcal{I})\to \mathcal{L}(\mathcal{O})]\to[\mathcal{L}(\mathcal{P})\to\mathcal{L}(\mathcal{F})]$ be a $k$-input parallel superinstrument implementing isometry inversion of $V\in \mathbb{V}_{\mathrm{iso}}(d,D)$ with the optimal success probability $p_{\mathrm{opt}}  = p_\mathrm{opt}(d,D,k)$, i.e.,
\begin{align}
    \doublewidetilde{\mathcal{S}^{}}_{\mathrm{opt}}(\widetilde{\mathcal{V}}^{\otimes k}) (\rho_\mathrm{in})=p_{\mathrm{opt}}\widetilde{\mathcal{V}}_{\mathrm{inv}} (\rho_\mathrm{in})\;\;\;(\forall V\in  \mathbb{V}_\mathrm{iso} (d,D),  \forall \rho_\mathrm{in} \in \mathcal{L}(\mathrm{Im} V)).\label{eq:optimal_parallel_isometry_inversion}
\end{align}
The parallel superinstrument $\{\doublewidetilde{\mathcal{S}^{}}_{\mathrm{opt}}, \doublewidetilde{\mathcal{F}^{}}_{\mathrm{opt}}\}$ can be written as
\begin{align}
    \doublewidetilde{\mathcal{S}^{}}_{\mathrm{opt}}(\widetilde{\Lambda}_{\mathrm{in}})
    &=\widetilde{\mathcal{D}}_{S}\circ 
    \left(
    \widetilde{\Lambda}_{\mathrm{in}} \otimes \widetilde{\1}_{\mathcal{A}}
    \right)\circ \widetilde{\mathcal{E}},\label{eq:parallel_isometry_inversion_superinstrument}\\
    \doublewidetilde{\mathcal{F}^{}}_{\mathrm{opt}}(\widetilde{\Lambda}_{\mathrm{in}})
    &=\widetilde{\mathcal{D}}_{F}\circ 
    \left(
    \widetilde{\Lambda}_{\mathrm{in}} \otimes \widetilde{\1}_{\mathcal{A}}
    \right)\circ \widetilde{\mathcal{E}},
\end{align}
where $\mathcal{A}$ {is} an auxiliary Hilbert space, $\widetilde{\mathcal{E}}:\mathcal{L}(\mathcal{P})\to \mathcal{L}(\mathcal{I}\otimes \mathcal{A})$ is a CPTP map, and $\{\widetilde{\mathcal{D}}_{S}, \widetilde{\mathcal{D}}_{F}\}$ is a quantum instrument (see the left panel of Figure~\ref{fig:untiary_inversion_to_isometry_inversion}~ (b)). 

We construct a parallel protocol for unitary inversion of $U\in \mathbb{U}(d)$ with the same success probability $p_{\mathrm{opt}}$ as the parallel protocol for isometry inversion of $V_{d, D}\in \mathbb{V}_{\mathrm{iso}}(d, D)$ as follows. We consider Hilbert spaces $\mathcal{P}'=\mathbb{C}^d$ and $\mathcal{O}'_i=\mathbb{C}^d$ for $i\in \{1, \cdots, k\}$ and the joint Hilbert space $\mathcal{O}'\coloneqq \bigotimes_{i=1}^k \mathcal{O}'_i$. We define the embedding isometry operator $V^{\mathrm{embed}}:\mathbb{C}^d\to \mathbb{C}^D$ as
\begin{align}
    V^{\mathrm{embed}}\coloneqq \sum_{i=0}^{d-1}\ket{i}_{\mathbb{C}^D}\bra{i}_{\mathbb{C}^d},\label{eq:def_embedding}
\end{align}
where $\{\ket{i}\}_{i=0}^{d-1}$ and $\{\ket{i}\}_{i=0}^{D-1}$ are the computational bases of $\mathbb{C}^d$ and $\mathbb{C}^D$, respectively. By defining $\{\doublewidetilde{\mathcal{S}'}, \doublewidetilde{\mathcal{F}'}\}:[\mathcal{L}(\mathcal{I})\to \mathcal{L}(\mathcal{O}')]\to[\mathcal{L}(\mathcal{P}')\to\mathcal{L}(\mathcal{F})]$ as
\begin{align}
    \doublewidetilde{\mathcal{S}'}(\widetilde{\Lambda}_{\mathrm{in}})
    \coloneqq&
    \widetilde{\mathcal{D}}_{S}\circ 
    \left[
    \left(\bigotimes_{i=1}^{k} \widetilde{\mathcal{V}}^{\mathrm{embed}}_{\mathcal{O}'_i\to \mathcal{O}_i}\circ \widetilde{\Lambda}_{\mathrm{in}}\right) \otimes \widetilde{\1}_{\mathcal{A}}
    \right]
    \circ \widetilde{\mathcal{E}}\circ \widetilde{\mathcal{V}}^{\mathrm{embed}}_{\mathcal{P}'\to \mathcal{P}},\\
    \doublewidetilde{\mathcal{F}'}(\widetilde{\Lambda}_{\mathrm{in}})
    \coloneqq&
    \widetilde{\mathcal{D}}_{F}\circ 
    \left[
    \left(\bigotimes_{i=1}^{k} \widetilde{\mathcal{V}}^{\mathrm{embed}}_{\mathcal{O}'_i\to \mathcal{O}_i}\circ \widetilde{\Lambda}_{\mathrm{in}}\right) \otimes \widetilde{\1}_{\mathcal{A}}
    \right]
    \circ \widetilde{\mathcal{E}}\circ \widetilde{\mathcal{V}}^{\mathrm{embed}}_{\mathcal{P}'\to \mathcal{P}},
\end{align}
$\{\doublewidetilde{\mathcal{S}'}, \doublewidetilde{\mathcal{F}'}\}$ is a parallel superinstrument as shown in Figure~\ref{fig:untiary_inversion_to_isometry_inversion}~(b). Since $V^{\mathrm{embed}}U$ is an isometry operator for all $U\in \mathbb{U}(d)$, we obtain
\begin{align}
    \doublewidetilde{\mathcal{S}'}(\widetilde{\mathcal{U}}^{\otimes k})=p_{\mathrm{opt}}(\widetilde{\mathcal{V}}^{\mathrm{embed}}\circ \widetilde{\mathcal{U}})_{\mathrm{inv}} \circ \widetilde{\mathcal{V}}^{\mathrm{embed}}
\end{align}
for all $U\in \mathbb{U}(d)$ from Eq.~(\ref{eq:optimal_parallel_isometry_inversion}). Since the inverse operation of an isometry operation is defined as Eq.~(\ref{eq:inverse}),
\begin{align}
    \doublewidetilde{\mathcal{S}'}(\widetilde{\mathcal{U}}^{\otimes k})\circ \widetilde{\mathcal{U}}
    &=p_{\mathrm{opt}}(\widetilde{\mathcal{V}}^{\mathrm{embed}}\circ \widetilde{\mathcal{U}})_{\mathrm{inv}} \circ \widetilde{\mathcal{V}}^{\mathrm{embed}}\circ \widetilde{\mathcal{U}}\\
    &=p_{\mathrm{opt}} \widetilde{\1}_d,
\end{align}
i.e.,
\begin{align}
    \doublewidetilde{\mathcal{S}'}(\widetilde{\mathcal{U}}^{\otimes k})=p_{\mathrm{opt}} \widetilde{\mathcal{U}}^{\dagger}
\end{align}
holds for all $U\in \mathbb{U}(d)$. Therefore, the parallel superinstrument $\{\doublewidetilde{\mathcal{S}'}, \doublewidetilde{\mathcal{F}'}\}$ implements unitary inversion of $U\in \mathbb{U}(d)$ with the success probability $p_{\mathrm{opt}}$. This completes the proof of the first part of Theorem~\ref{thm:optimal_isometry_inversion_and_PBT}.
 As stated in the first paragraph of the proof, the second part (achievability of the success probability $p_\mathrm{succ} = \lfloor k/(d-1) \rfloor/[d^2+\lfloor k/(d-1)\rfloor -1]$) of Theorem~\ref{thm:optimal_isometry_inversion_and_PBT} follows from the combination of the first part and a parallel delayed input-state protocol for unitary inversion shown in Appendix~\ref{sec:appendix_unitary_inversion} and Ref.~\cite{Quintino2019Reversing}.
\end{proof}

\begin{figure}[tb]
    \centering
    \includegraphics[width=\linewidth]{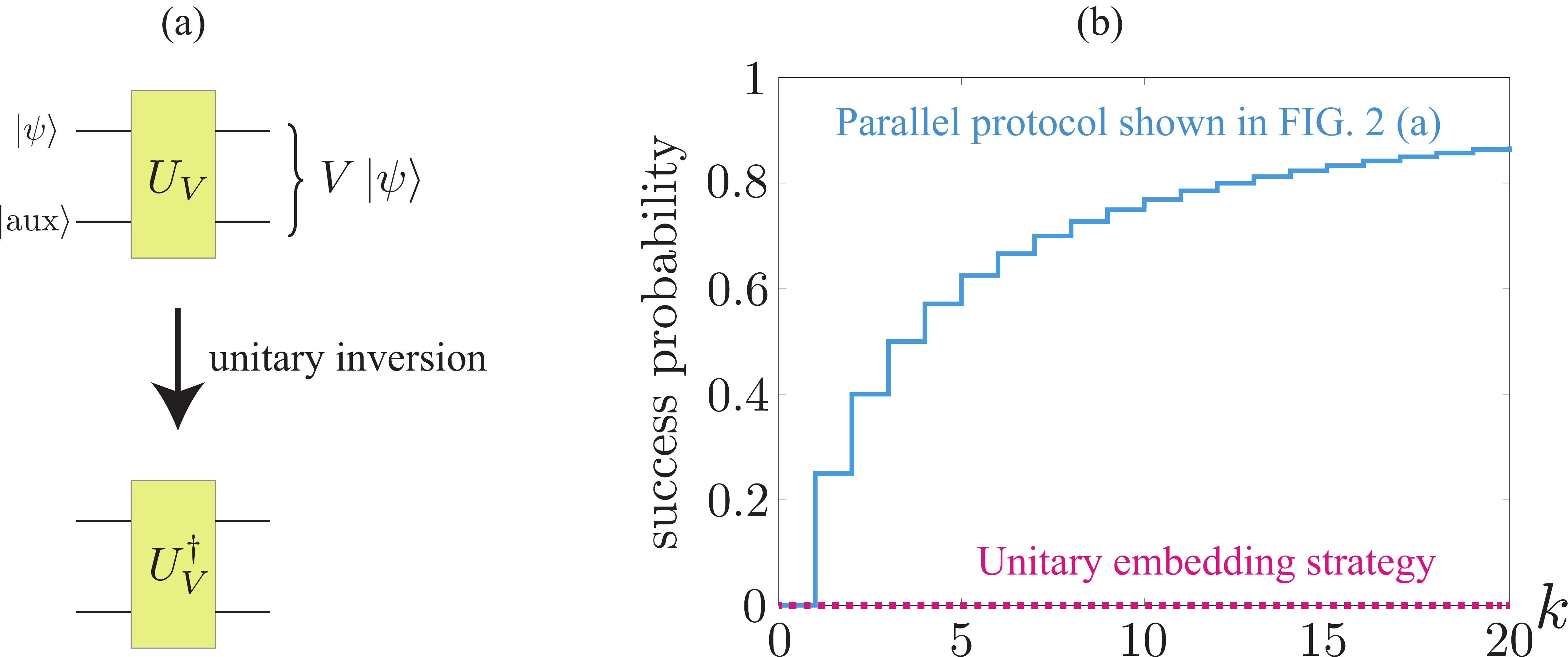}
    \caption{{(a) Unitary embedding strategy for isometry inversion. By inverting a unitary operation $U_V$ satisfying $V\ket{\psi}=U_V(\ket{\psi}\otimes \ket{\mathrm{aux}})$, we obtain the inverse map of an isometry operation {$\widetilde{\mathcal{V}}$}.} (b) Comparison of the success probability of the two protocols for isometry inversion using $k$ calls of $V\in \mathbb{V}_{\mathrm{iso}}(2, 2^5)$ or $U_V\in \mathbb{U}(2^5)$: the parallel protocol shown in Figure~\ref{fig:isometry_inversion_protocols}~(a) (blue line) and the strategy based on the embedding of the input isometry operation $V$ in a unitary operation $U_V$ (red dotted line).}
    \label{fig:untiary_embedding_strategy_and_graph}
\end{figure}

\subsection{Analysis of other related settings}
\label{subsec:comparison}

\subsubsection{Unitary embedding strategy}
for $D=dd'$, any isometry operator $V\in \mathbb{V}_\mathrm{iso} (d,D)$ can be  embedded in an appropriate unitary operation $ U_V  \in \mathbb{U}(D)$  by adding a fixed auxiliary system $\ket{\mathrm{aux}}\in\mathbb{C}^{d'}$  and seeking $U_V$ satisfying
\begin{align}
    V\ket{\psi}= U_V (\ket{\psi}\otimes \ket{\mathrm{aux}})
\end{align}
for all $\ket{\psi} \in \mathbb{C}^d$   (see Figure~\ref{fig:untiary_embedding_strategy_and_graph}~(a)).  Assuming that the black box implementing $ U_V$ and the fixed auxiliary state  are given, we can implement an inverse map of $\widetilde{\mathcal{V}}$ by applying the unitary inversion protocol. Any $D$-dimensional unitary inversion protocol needs at least $D-1$ calls of an input unitary operation  to be successful with a non-zero probability \cite{Quintino2019Reversing}. In contrast, our direct isometry inversion protocol using $k$ parallel calls of an input isometry operation $\widetilde{\mathcal{V}}$ corresponding to $V\in \mathbb{V}_\mathrm{iso} (d,D)$  implements the inverse map $\widetilde{\mathcal{V}}_{\mathrm{inv}}$ with a success probability $p_{\mathrm{succ}}=\lfloor k/(d-1) \rfloor/[d^2+\lfloor k/(d-1) \rfloor -1]$.
We compare these success probabilities for $d=2, D=2^5$ and $1\leq k\leq 20$ (see Figure~\ref{fig:untiary_embedding_strategy_and_graph}~(b)). The isometry inversion protocol achieves a success probability $p_{\mathrm{succ}}\approx 87\%$ at $k=20$, while probabilistic unitary inversion is impossible for any $k \leq 30$.

\subsubsection{Process tomography strategy}
Approximate isometry inversion can be achieved by approximating the classical description of $V\in \mathbb{V}_\mathrm{iso} (d,D)$ by quantum process tomography first \cite{Nielsen2010Quantum, Mohseni2008Quantum}, and then  applying the inverse map calculated from the classical description.  Quantum process tomography consists of three steps: preparing a quantum state called a probe state, applying a black box quantum operation to a probe state, and performing a measurement on the output state. By  repeating this procedure, we obtain a probability distribution of measurement outcomes to estimate a black box operation. An isometry operator $V \in \mathbb{V}_{\mathrm{iso}}(d, D)$ is uniquely specified by $2Dd-d^2+d-1  = \Omega(Dd)$ real parameters\footnote{We say $f(x) = \Omega(g(x))$ if $\limsup_{x\to \infty}{|f(x)/g(x)|}>0$.}, and $\Omega(\epsilon^{-2}\ln (1/(1-p)))$ measurements are necessary to estimate each parameter with  at most  error $\epsilon$ and  at least probability $p$ \cite{Mohseni2008Quantum}. Therefore, we need $ \Omega(Dd\epsilon^{-2}\ln (1/(1-p)))$ calls of an isometry operation $\widetilde{\mathcal{V}}$ for the process tomography of $V$.  On the other hand, since we can achieve the success probability $p_{\mathrm{succ}}=\lfloor k/(d-1) \rfloor/[d^2+\lfloor k/(d-1)\rfloor -1]$ by $k$ parallel calls of an isometry operation $\widetilde{\mathcal{V}}$, the success probability becomes $ p_{\text{succ}}$ for $k= (d-1)(d^2-1)/(1- p_{\text{succ}})=\mathcal{O}(d^3/(1- p_{\text{succ}}))$.\footnote{We say $f(x) = \mathcal{O}(g(x))$ if $\limsup_{x\to \infty}{|f(x)/g(x)|}<\infty$.}
 The isometry inversion protocol outperforms the strategy based on process tomography at two points. First, the isometry inversion protocol can implement the inverse operation  without error ($\epsilon=0$) using finite calls of $\widetilde{\mathcal{V}}$, while the process tomography needs $\Omega(\epsilon^{-2})$ calls to achieve the  error below $\epsilon$. Second, the isometry inversion protocol can be done with  a $D$-independent number of calls of $\widetilde{\mathcal{V}}$, while the process tomography requires a $D$-dependent number of calls.
 Though the scaling of the number of calls in our protocol with respect to $p_\mathrm{succ}$ is polynomial, we also find a protocol with $k = \mathcal{O}(\ln (1/(1-p_\mathrm{succ})))$ using a ``success-or-draw'' protocol for $d=2$ and $D=3$, whose existence is shown numerically (see Section \ref{sec:SDP}).

\section{Difference between isometry inversion protocols and unitary inversion protocols}
\label{sec:difference}

The parallel unitary inversion protocol presented in Ref. \cite{Quintino2019Reversing} consists of a concatenation of unitary complex conjugation and unitary transposition. {The implementation of unitary complex conjugation shown in Ref.~\cite{Miyazaki2019Complex} relies on the fact that the complex conjugate representation of $\mathbb{U}(d)$ is unitarily equivalent to the antisymmetric subspace in the tensor representation of $\mathbb{U}(d)$ on $(\mathbb{C}^d)^{\otimes d-1}$. The unitary transposition protocol presented in Ref.~\cite{Quintino2019Probabilistic} utilizes a variant of gate teleportation \cite{Gottesman1999Demonstrating} or the probabilistic port-based teleportation \cite{Ishizaka2008Asymptotica,Studzinski2017Portbased}.} In contrast, we show that isometry inversion protocols cannot be decomposed into isometry complex conjugation and isometry transposition.  We can consider another similar strategy by considering ``isometry pseudo complex conjugation'', but this strategy is shown to be inefficient. We first investigate protocols for isometry complex conjugation and isometry  pseudo complex conjugation, and then analyze the difference between isometry inversion and unitary inversion in this section.  Isometry transposition is investigated in Appendix \ref{sec:appendix_isometry_tranpose_PBT}.

\subsection{The no-go theorem for isometry complex conjugation}
\label{sec:isometry_conjugation_adjointation}

We prove that any isometry inversion protocol cannot be decomposed into isometry complex conjugation and isometry transposition  by showing a no-go theorem for isometry complex conjugation.

 To state the no-go theorem properly, we introduce the notion of  {\it general} superinstrument  including the ones with an indefinite causal order \cite{Oreshkov2012Quantum, Chiribella2013Quantum, Araujo2015Witnessing, Wechs2019Definition, Bisio2019Theoretical, Yokojima2021Consequences, Vanrietvelde2021Routed}, which  describes the most general higher-order probabilistic transformation  where the order of the use of the input maps is not pre-determined.  Using the same notations of the Hilbert spaces and the input maps introduced for a superinstrument in Section \ref{sec:supermap_superinstrument_higher-order}, a general superinstrument is a set of $k$-input supermaps $\{\doublewidetilde{\mathcal{C}^{}}_a\}: [\mathcal{L}(\mathcal{I})\to \mathcal{L}(\mathcal{O})] \to [\mathcal{L}(\mathcal{P})\to \mathcal{L}(\mathcal{F})]$ such that a set of output maps $\{\doublewidetilde{\mathcal{C}^{}}_a\otimes \doublewidetilde{\1}(\widetilde{\Lambda}_{\mathrm{in}})\}_ a$ is a quantum instrument for any set of input  CPTP maps $\{\widetilde{\Lambda}_{\mathrm{in}}^{(i)}\}_{i=1}^{k}$, where $\doublewidetilde{\1}$ is the identity supermap defined by $\doublewidetilde{\1}(\widetilde{\Lambda})=\widetilde{\Lambda}$.
 This definition of a general superinstrument $\{\supermap{C}_a\}$ states that a set of output maps is a quantum instrument even if the superinstrument acts on a subsystem of input CPTP maps.  Note that the condition that
$\supermap{C}_a \otimes \doublewidetilde{\1}(\widetilde{\Lambda}_\mathrm{in})$ is a quantum instrument for any input CPTP maps $\{\widetilde{\Lambda}_{\mathrm{in}}^{(i)}\}_{i=1}^{k}$,
is strictly stronger than the relaxed condition that
(b) $\supermap{C}_a(\widetilde{\Lambda}_\mathrm{in})$ is a quantum instrument for any input CPTP maps $\{\widetilde{\Lambda}_{\mathrm{in}}^{(i)}\}_{i=1}^{k}$,
which is similar to the distinction between positive maps and completely positive maps.

A general superinstrument $\{\doublewidetilde{\mathcal{C}^{}}_a\}$  is not necessarily represented by a quantum circuit,  which requires to fix a causal ordering of the use of the input maps.   To illustrate a protocol using a general supermap in a similar manner of quantum circuits, we use a notation of a general supermap represented by a box with windows for plugging input maps (see Figure~\ref{fig:general_superinstrument_and_complex_conjugation_proof}~(a)).  As the windows of the box are not causally ordered, the box represents a general supermap, a transformation of input maps with unspecified order to the output map\footnote{This notation is adopted from Ref.~\cite{Yokojima2021Consequences}.}.

 Using the notion of general superinstrument, we can state the following no-go theorem for isometry complex conjugation.

\begin{Thm}
\label{thm:isometry_complex_conjugation_nogo}
If $D\geq 2d$, it is impossible to transform finite calls of an isometry operation $\widetilde{\mathcal{V}}: \mathcal{L}(\mathbb{C}^d)\to\mathcal{L}(\mathbb{C}^D)$ into its complex conjugate map $\widetilde{\mathcal{V}}^*$ with a non-zero success probability  using any general superinstrument.
\end{Thm}

Note that  this no-go theorem states the impossibility of isometry complex conjugation from finite calls even if we can implement a general superinstrument, which is beyond the quantum circuit model.  To prove this no-go theorem, we first show the impossibility for $D$ is a multiple of $d$.

\begin{figure}[tb]
    \centering
    \includegraphics[width=\linewidth]{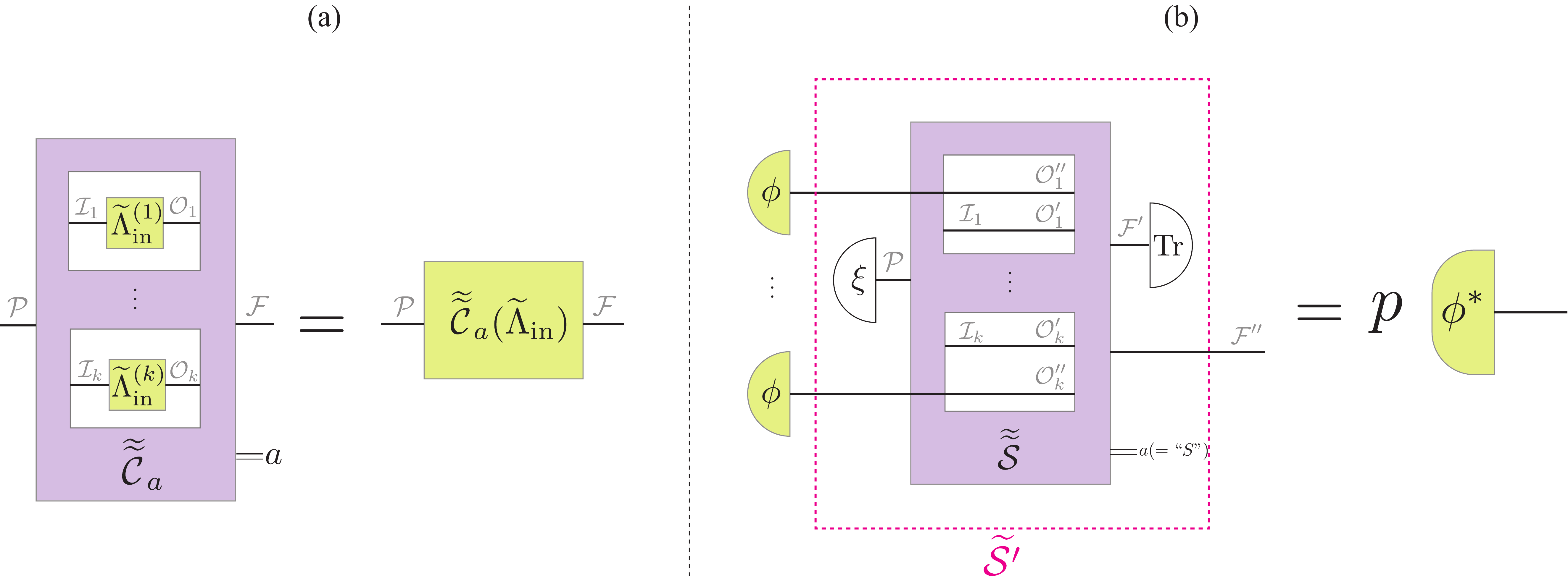}
    \caption{(a) The left hand side shows an illustration of a general superinstrument $\{\doublewidetilde{\mathcal{C}^{}}_a\}$ as a purple box with windows where input maps $\widetilde{\Lambda}_{\mathrm{in}}^{(i)}\;(i\in\{1, \cdots, k\})$ are plugged. The right hand side is the output map of a general superinstrument $\{\doublewidetilde{\mathcal{C}^{}}_a\}$ for input maps $\widetilde{\Lambda}_{\mathrm{in}}^{(i)}$. (b) A hypothetical protocol using a general superinstrument $\{\doublewidetilde{\mathcal{S}^{}}, \doublewidetilde{\mathcal{F}^{}}\}$ for isometry complex conjugation to prove Lemma~\ref{lem:isometry_conjugate_nogo1} by contradiction. The wire between $\mathcal{I}_i$ and $\mathcal{O}_i$ ($i\in \{1, \cdots, k\}$) represent the identity map and $\xi$ is an arbitrary quantum state.}
    \label{fig:general_superinstrument_and_complex_conjugation_proof}
\end{figure}

\begin{Lem}
\label{lem:isometry_conjugate_nogo1}
Let $d'$ be a natural number greater than 1.
If $D=dd'$, it is impossible to transform finite calls of an isometry operation $\widetilde{\mathcal{V}}: \mathcal{L}(\mathbb{C}^d)\to\mathcal{L}(\mathbb{C}^D)$ into its complex conjugate map $\widetilde{\mathcal{V}}^*$ with a non-zero success probability  using any general superinstrument.
\end{Lem}

\begin{proof}
We prove this Lemma~by contradiction.  Probabilistic exact state complex conjugation of $\ket{\phi}$ is shown to be possible by applying probabilistic isometry complex conjugation to an isometry operator $V_{\ket{\phi}}$, which joints an auxiliary quantum state $\ket{\phi}$ to an input state. If probabilistic isometry complex conjugation can be implemented with a non-zero success probability, probabilistic exact state complex conjugation is also possible with a non-zero success probability, which contradicts the no-go theorem of state complex conjugation\footnote{The impossibility of probabilistic exact state complex conjugation $\ket{\phi}^{\otimes k} \mapsto \ket{\phi^*}$ can be easily understood for $k=1$. Since $\ketbra{\phi^*}{\phi^*}=\ketbra{\phi}{\phi}^T$ holds due to the hermicity of $\ketbra{\phi}{\phi}$, a quantum instrument $\{\widetilde{\mathcal{S}}', \widetilde{\mathcal{F}}'\}$ satisfies $\widetilde{\mathcal{S}}'(\ketbra{\phi}{\phi})=p_{\text{succ}}\ketbra{\phi^*}{\phi^*}$ if and only if $\widetilde{\mathcal{S}}'(\rho) = p_{\text{succ}} \rho^T$ for any $\rho\in \mathcal{L}(\mathcal{H})$. However, since transposition of a state is not completely positive, any quantum instrument $\{\widetilde{\mathcal{S}}', \widetilde{\mathcal{F}}'\}$ does not achieve non-zero success probability.} \cite{Yang2014Certifying}. See Figure~\ref{fig:general_superinstrument_and_complex_conjugation_proof}~(b) for  a hypothetical protocol using a general superinstrument to prove this Lemma~by contradiction.

We consider a $k$-input  general superinstrument $\{\doublewidetilde{\mathcal{S}^{}}, \doublewidetilde{\mathcal{F}^{}}\}: [\mathcal{L}(\mathcal{I})\to\mathcal{L}(\mathcal{O})]\to[\mathcal{L}(\mathcal{P})\to\mathcal{L}(\mathcal{F})]$, where the Hilbert spaces are given by $\mathcal{P}= \mathbb{C}^d$, $\mathcal{I}_i= \mathbb{C}^d$, $\mathcal{F}= \mathbb{C}^D=\mathbb{C}^{dd'}$, $\mathcal{O}_i= \mathbb{C}^D=\mathbb{C}^{dd'}$ for $i\in\{1, \cdots, k\}$, and $\mathcal{I}$ and $\mathcal{O}$ are the joint Hilbert spaces defined by $\mathcal{I}\coloneqq \bigotimes_{i=1}^{k}\mathcal{I}_i$ and $\mathcal{O}\coloneqq \bigotimes_{i=1}^{k}\mathcal{O}_i$, respectively.
We assume that the  general superinstrument $\{\doublewidetilde{\mathcal{S}^{}}, \doublewidetilde{\mathcal{F}^{}}\}$ implements an isometry complex conjugation protocol with a non-zero success probability $p_{\mathrm{succ}}>0$, i.e.,
\begin{align}
    \doublewidetilde{\mathcal{S}^{}}(\widetilde{\mathcal{V}}^{\otimes k})=p_{\mathrm{succ}}\widetilde{\mathcal{V}}^*\;\;\;(\forall V\in \mathbb{V}_{\mathrm{iso}}(d,dd')).
\end{align}

To construct the hypothetical protocol shown in Figure~\ref{fig:general_superinstrument_and_complex_conjugation_proof}~(b), we decompose the Hilbert space $\mathcal{O}_i$ as $\mathcal{O}_i=\mathcal{O}'_i\otimes \mathcal{O}_i''$, where $\mathcal{O}_i'= \mathbb{C}^d$ and $\mathcal{O}''_i= \mathbb{C}^{d'}$. We take the computational basis of $\mathcal{O}_i$ as $\{\ket{j'}_{\mathcal{O}'_i}\otimes \ket{j''}_{\mathcal{O}''_i}\}$, where $\{\ket{j'}_{\mathcal{O}'_i}\}$ and $\{\ket{j''}_{\mathcal{O}''_i}\}$ are the computational basis of $\mathcal{O}'_i$ and $\mathcal{O}''_i$, respectively. Similarly, we decompose the Hilbert space $\mathcal{F}$ as $\mathcal{F}=\mathcal{F}'\otimes \mathcal{F}''$,where $\mathcal{F}'= \mathbb{C}^d$ and $\mathcal{F}''= \mathbb{C}^{d'}$ and choose the computational basis of $\mathcal{F}$ as $\{\ket{j'}_{\mathcal{F}'}\otimes \ket{j''}_{\mathcal{F}''}\}$, where $\{\ket{j'}_{\mathcal{F}'}\}$ and $\{\ket{j''}_{\mathcal{F}''}\}$ are the computational basis of $\mathcal{F}'$ and $\mathcal{F}''$, respectively.

We define an isometry operator $V_{\ket{\phi}}\in \mathbb{V}_{\mathrm{iso}}(d, dd')$ by $V_{\ket{\phi}}\ket{\psi}\coloneqq \ket{\psi}\otimes \ket{\phi}$, where $\ket{\phi}\in \mathbb{C}^{d'}$ and $\ket{\psi}\in \mathbb{C}^d$.
Inserting $k$ calls of the isometry operation $\widetilde{\mathcal{V}}_{\ket{\phi}}$ corresponding to $V_{\ket{\phi}}$ into $\doublewidetilde{\mathcal{S}^{}}$, we obtain \begin{align}
    \doublewidetilde{\mathcal{S}^{}}(\widetilde{\mathcal{V}}_{\ket{\phi}}^{\otimes k})=p_{\mathrm{succ}}\widetilde{\mathcal{V}}_{\ket{\phi}^*}.
\end{align}
Next, inserting an arbitrary quantum state $\ket{\xi}\in\mathcal{P}$ into $\doublewidetilde{\mathcal{S}^{}}(\widetilde{\mathcal{V}}_{\ket{\phi}}^{\otimes k})$ and discarding $\mathcal{F}'$, we obtain
\begin{align}
    \Tr_{\mathcal{F}'}
    \left[
    \doublewidetilde{\mathcal{S}^{}}
    \left(
    \widetilde{\mathcal{V}}_{\ket{\phi}}^{\otimes k}
    \right)
    (\ketbra{\xi}{\xi})
    \right]
    &=p_{\mathrm{succ}}\Tr_{\mathcal{F}'}
    \left[
    \widetilde{\mathcal{V}}_{\ket{\phi}^*}(\ketbra{\xi}{\xi})
    \right]\\
    &=p_{\mathrm{succ}}\ketbra{\phi^*}{\phi^*}.\label{eq:state_conjugation}
\end{align}
We define $\widetilde{\mathcal{S}}':\bigotimes_{i=1}^{k}\mathcal{L}(\mathcal{O}''_i)\to\mathcal{L}(\mathcal{F}')$ by
\begin{align}
    \widetilde{\mathcal{S}}'\left(\bigotimes_{i=1}^{k}\rho^{(i)}\right)\coloneqq \Tr_{\mathcal{F}'}
    \left[
    \doublewidetilde{\mathcal{S}^{}}\left(\bigotimes_{i=1}^{k}\widetilde{\Lambda}_{{\rho^{(i)}}}(\ketbra{\xi}{\xi})\right)
    \right],
\end{align}
where $\widetilde{\Lambda}_{\rho}:\mathcal{L}(\mathbb{C}^d)\to\mathcal{L}(\mathbb{C}^{dd'})$ is a CPTP map defined by $\widetilde{\Lambda}_{\rho}(\sigma)=\sigma\otimes \rho$. By defining $\widetilde{\mathcal{F}}':\bigotimes_{i=1}^{k}\mathcal{L}(\mathcal{O}''_i)\to\mathcal{L}(\mathcal{F}')$ similarly,
$\{\widetilde{\mathcal{S}}', \widetilde{\mathcal{F}}'\}$ is a  quantum instrument. From Eq.~(\ref{eq:state_conjugation}), we obtain
\begin{align}
    \widetilde{\mathcal{S}}'(\ketbra{\phi}{\phi}^{\otimes k})=p_{\mathrm{succ}}\ketbra{\phi^{*}}{\phi^{*}}\;\;\;(\forall\ket{\phi}\in\mathbb{C}^{d'}).
\end{align}
However, this contradicts the fact that probabilistic exact quantum state complex conjugation is impossible using finite copies of quantum states \cite{Yang2014Certifying}.
\end{proof}

\begin{proof}[Proof of Theorem~\ref{thm:isometry_complex_conjugation_nogo}]
 We show this no-go theorem by contradiction. We first show that, for any $D\geq D'\geq d$, if probabilistic exact isometry complex conjugation of $V\in \mathbb{V}_{\mathrm{iso}}(d, D)$ is possible, we can implement probabilistic exact isometry complex conjugation of $V\in \mathbb{V}_{\mathrm{iso}}(d, D')$ with the same success probability by embedding $\mathbb{V}_{\mathrm{iso}}(d, D')$ to $\mathbb{V}_{\mathrm{iso}}(d, D)$. Then, if probabilistic exact isometry complex conjugation is possible with a non-zero success probability for $D\geq 2d$, it is possible for $D=2d$, which contradicts Lemma~\ref{lem:isometry_conjugate_nogo1}.

We consider Hilbert spaces $\mathcal{P}=\mathbb{C}^d$, $\mathcal{F}=\mathbb{C}^D$, $\mathcal{I}_i= \mathbb{C}^d$ and $\mathcal{O}_i=\mathbb{C}^D$ for $i\in \{1, \cdots, k\}$ and the joint Hilbert spaces $\mathcal{I}$ and $\mathcal{O}$ defined as $\mathcal{I}\coloneqq \bigotimes_{i=1}^k \mathcal{I}_i$ and 
$\mathcal{O}\coloneqq \bigotimes_{i=1}^k \mathcal{O}_i$, respectively. Suppose $\{\doublewidetilde{\mathcal{S}^{}}, \doublewidetilde{\mathcal{F}^{}}\}:[\mathcal{L}(\mathcal{I})\to \mathcal{L}(\mathcal{O})]\to[\mathcal{L}(\mathcal{P})\to\mathcal{L}(\mathcal{F})]$ be a $k$-input general superinstrument implementing isometry inversion of $V\in \mathbb{V}_{\mathrm{iso}}(d,D)$ with a non-zero success probability $p_{\mathrm{succ}}>0$, i.e.,
\begin{align}
    \doublewidetilde{\mathcal{S}^{}}(\widetilde{\mathcal{V}}^{\otimes k})=p_{\mathrm{succ}}\widetilde{\mathcal{V}}^*\;\;\;(\forall V\in  \mathbb{V}_\mathrm{iso}).\label{eq:isometry_cc}
\end{align}
To construct a parallel protocol of isometry complex conjugation of $V\in \mathbb{V}_{\mathrm{iso}}(d, D')$, we first introduce Hilbert spaces $\mathcal{F}'=\mathbb{C}^{D'}$ and $\mathcal{O}'_i=\mathbb{C}^{D'}$ for $i\in \{1, \cdots, k\}$ and the joint Hilbert space $\mathcal{O}'\coloneqq \bigotimes_{i=1}^{k}\mathcal{O}'_i$. We define the embedding isometry operator $V^{\mathrm{embed}}: \mathbb{C}^{D'}\to \mathbb{C}^D$ similarly to Eq.~(\ref{eq:def_embedding}) and the CPTP map $\widetilde{\Xi}: \mathcal{L}(\mathcal{F})\to \mathcal{L}(\mathcal{F}')$ as
\begin{align}
    \widetilde{\Xi}(\rho)\coloneqq \sum_{i,j=0}^{D'-1}\ketbra{i}{j}_{\mathcal{F}'}\bra{i}\rho \ket{j}+\sum_{i=D'}^{D-1}\frac{\1_{D'}}{D'}\Tr(\ketbra{i}{i} \rho),
\end{align}
where $\{\ket{i}\}_{i=0}^{D'-1}$ and $\{\ket{i}\}_{i=0}^{D-1}$ are the computational bases of $\mathcal{F}'$ and $\mathcal{F}$, respectively. The CPTP map $\widetilde{\Xi}$ is an inverse map of $V^{\mathrm{embed}}$, i.e.,
\begin{align}
    \widetilde{\Xi}\circ \widetilde{\mathcal{V}}^{\mathrm{embed}}=\widetilde{\mathcal{I}}_d. \label{eq:extract_embedding}
\end{align}
By defining $\{\doublewidetilde{\mathcal{S}'}, \doublewidetilde{\mathcal{F}'}\}:[\mathcal{L}(\mathcal{I})\to \mathcal{L}(\mathcal{O}')]\to[\mathcal{L}(\mathcal{P})\to\mathcal{L}(\mathcal{F}')]$ as
\begin{align}
    \doublewidetilde{\mathcal{S}'}(\widetilde{\Lambda}_{\mathrm{in}})&\coloneqq \widetilde{\Xi}_{\mathcal{F}'\to \mathcal{F}}\circ \doublewidetilde{\mathcal{S}^{}}\left(\bigotimes_{i=1}^{k}\widetilde{\mathcal{V}}^{\mathrm{embed}}_{\mathcal{O}'_i\to \mathcal{O}_i}\circ \widetilde{\Lambda}_{\mathrm{in}}\right),\\
    \doublewidetilde{\mathcal{F}'}(\widetilde{\Lambda}_{\mathrm{in}})&\coloneqq \widetilde{\Xi}_{\mathcal{F}'\to \mathcal{F}}\circ \doublewidetilde{\mathcal{F}^{}}\left(\bigotimes_{i=1}^{k}\widetilde{\mathcal{V}}^{\mathrm{embed}}_{\mathcal{O}'_i\to \mathcal{O}_i}\circ \widetilde{\Lambda}_{\mathrm{in}}\right),
\end{align}
$\{\doublewidetilde{\mathcal{S}'}, \doublewidetilde{\mathcal{F}'}\}$ is a $k$-input general superinstrument. Since $V^{\mathrm{embed}}V_{d, D'}$ is an isometry operator for all $V_{d, D'}\in \mathbb{V}_{\mathrm{iso}}(d, D')$, we obtain
\begin{align}
    \doublewidetilde{\mathcal{S}'}(\widetilde{\mathcal{V}}_{d, D'})=p_{\mathrm{succ}} \widetilde{\Xi}\circ (\widetilde{\mathcal{V}}^{\mathrm{embed}}\circ \widetilde{\mathcal{V}}_{d, D'})^*
\end{align}
for all $V_{d, D'}\in \mathbb{V}_{\mathrm{iso}}(d, D')$ from Eq.~(\ref{eq:isometry_cc}). Since $V^{\mathrm{embed} *}=V^{\mathrm{embed}}$ holds (see Eq.~(\ref{eq:def_embedding})) and the CPTP map $\widetilde{\Xi}$ is an inverse map of $V^{\mathrm{embed}}$ (see Eq.~(\ref{eq:extract_embedding})), we obtain
\begin{align}
    \doublewidetilde{\mathcal{S}'}(\widetilde{\mathcal{V}}_{d, D'})
    &=p_{\mathrm{succ}} \widetilde{\Xi}\circ \widetilde{\mathcal{V}}^{\mathrm{embed}} \circ \widetilde{\mathcal{V}}_{d, D'}^*\\
    &=p_{\mathrm{succ}} \widetilde{\mathcal{V}}_{d, D'}^*
\end{align}
for all $V_{d, D'}\in \mathbb{V}_{\mathrm{iso}}(d, D')$. Thus, the general superinstrument $\{\doublewidetilde{\mathcal{S}'}, \doublewidetilde{\mathcal{F}'}\}$ implements isometry complex conjugation with the success probability $p_{\mathrm{succ}}$. However, Lemma~\ref{lem:isometry_conjugate_nogo1} immediately shows that probabilistic exact isometry complex conjugation with a non-zero success probability is impossible for $D=2d$. Therefore, probabilistic exact isometry complex conjugation with a non-zero success probability is impossible for $D\geq2d$.
\end{proof}

\subsection{Any isometry inversion protocol  concatenating isometry pseudo complex conjugation and the gate teleportation is inefficient}
\label{subsec:isometry_pseudo_cc}

Due to the no-go theorem for isometry complex conjugation, isometry inversion cannot be decomposed into isometry complex conjugation and isometry transposition. However, we can consider another similar strategy to construct isometry inversion  since the inverse map $\widetilde{\mathcal{V}}_{\mathrm{inv}}$ of an isometry operation $\widetilde{\mathcal{V}}$ is not necessarily the adjoint map $\widetilde{\mathcal{V}}^{\dagger}$. We consider ``isometry pseudo complex conjugation'', which is a task to implement the transposed map of $\widetilde{\mathcal{V}}_{\mathrm{inv}}$, which we call the pseudo complex conjugate map.  In other words, we say that a CP map $\widetilde{\mathcal{V}}_{\text{pcc}}: \mathcal{L}(\mathbb{C}^d) \to \mathcal{L}(\mathbb{C}^D)$ is a pseudo complex conjugate map of an isometry operation $\widetilde{\mathcal{V}}: \mathcal{L}(\mathbb{C}^d) \to \mathcal{L}(\mathbb{C}^D)$ if and only if $\widetilde{\mathcal{V}}_{\text{pcc}}^T \circ \widetilde{\mathcal{V}}(\rho) =\rho$ for all $\rho\in\mathcal{L}(\mathbb{C}^d)$. If  isometry pseudo complex conjugation is implementable, isometry inversion can be obtained by transposing the pseudo complex conjugate map using the gate teleportation \cite{Gottesman1999Demonstrating}.  We show that such a protocol is possible but not  as efficient as the protocol proposed in Section~\ref{sec:isometry_inversion}.

For $V\in \mathbb{V}_\mathrm{iso} (d,D)$, we define a CP map $\widetilde{\mathcal{V}}'': \mathcal{L}(\mathbb{C}^d)\to \mathcal{L}(\mathbb{C}^D)$ as
\begin{align}
    \widetilde{\mathcal{V}}''(\rho_{\mathrm{in}})=V^{*}\rho_{\mathrm{in}} (V^{*})^{\dagger} +\Pi^{*}_{(\Im V)^{\perp}}\Tr(\rho_{\mathrm{in}}).\label{eq:f(V)}
\end{align}
Then,  $\widetilde{\mathcal{V}}''$ is a pseudo complex conjugate map of $\widetilde{\mathcal{V}}$ since $\widetilde{\mathcal{V}}'\coloneqq \widetilde{\mathcal{V}}''^{T}$ is an inverse operation of $\widetilde{\mathcal{V}}$ given by
\begin{align}
    \widetilde{\mathcal{V}}'(\rho_{\mathrm{in}})=V^{\dagger} \rho_{\mathrm{in}} V+\1_d\Tr
    \left[
    \Pi_{(\Im V)^{\perp}}\rho_\mathrm{in}
    \right]
    .
\end{align}
The pseudo complex conjugate map $\widetilde{\mathcal{V}}''$ is implementable with a certain number of calls of $\widetilde{\mathcal{V}}$  using the protocol shown by Figure~\ref{fig:isometry_conjugate_modoki}.
 
In the protocol shown by Figure~\ref{fig:isometry_conjugate_modoki}, the Hilbert spaces are given by  $\mathcal{P}=\mathbb{C}^d$, $\mathcal{I}_i=\mathbb{C}^d$, $\mathcal{F}=\mathbb{C}^D$, and $\mathcal{O}_i=\mathbb{C}^D$ for $i\in \{1, \cdots, d-1\}$.  We also define the joint Hilbert spaces of $(d-1)$-input spaces and $(d-1)$-output spaces by $\mathcal{I}\coloneqq \bigotimes_{i=1}^{d-1} \mathcal{I}_i$ and $\mathcal{O}\coloneqq \bigotimes_{i=1}^{d-1} \mathcal{O}_i$, respectively.   The isometry operator $V^{\mathrm{a.s.}}$  encodes quantum information of $\rho_\mathrm{in}$ on $\mathcal{P}=\mathbb{C}^d$ into the $d$-dimensional totally antisymmetric subspace of $\mathcal{I}$ and is given by
\begin{align}
    V^{\mathrm{a.s.}}\coloneqq \sum_{\vec{k}\in\{0, \cdots, d-1\}^{d}}\frac{\epsilon_{\vec{k}}}{\sqrt{(d-1)!}}\ketbra{k_1\cdots k_{d-1}}{k_d},\label{eq:def_V^as}
\end{align}
where $\epsilon_{\vec{k}}$ is the antisymmetric tensor with rank $d$ and $\{\ket{k}\}_{k=0}^{d-1}$ is the computational basis of $\mathcal{P}$ and $\mathcal{I}_i$.    The CPTP map $\widetilde{\Lambda}:\mathcal{L}(\mathcal{O})\to \mathcal{L}(\mathcal{F})$  decodes quantum information encoded in the $d$-dimensional subspace of $\mathcal{O}$ into $\mathcal{F}=\mathbb{C}^D$ and given by
\begin{align}
    \widetilde{\Lambda}(\rho)\coloneqq \frac{1}{D-d+1}\sum_{0\leq j_1<\cdots < j_d\leq D-1}A_{\vec{j}}\rho A_{\vec{j}}^{\dagger}  + \frac{\1_D}{D} \Tr[(\1_{\mathcal{O}}-\Pi_{\mathcal{O}}^{\text{a.s.}})\rho],\label{eq:def_Lambda}
\end{align}
where $A_{\vec{j}}: \mathcal{O}\to\mathcal{F}$ is defined by
\begin{align}
    A_{\vec{j}}\coloneqq \sum_{\vec{k}\in \{1,\cdots d\}^d}\frac{\epsilon_{\vec{k}}}{\sqrt{(d-1)!}}\ketbra{j_{k_d}}{j_{k_1}\cdots j_{k_{d-1}}}
\end{align}
and $\{\ket{j}\}_{j=0}^{D-1}$ is the computational basis of $\mathcal{F}$ and $\mathcal{O}_i$.  The operator $\Pi_{\mathcal{O}}^{\text{a.s.}}$ is the orthogonal projector onto the antisymmetric subspace of $\mathcal{O}=\bigotimes_{i=1}^{d-1} \mathcal{O}_i$.

 The protocol presented above implements isometry pseudo complex conjugation as shown in the following Theorem.

\begin{Thm}
\label{thm:conjugate_modoki}
A parallel protocol shown in Figure~\ref{fig:isometry_conjugate_modoki} transforms $d-1$ calls of an isometry operation $\widetilde{\mathcal{V}}$ corresponding to $V \in \mathbb{V}_\mathrm{iso} (d, D)$ into its pseudo complex conjugate map $\widetilde{\mathcal{V}}''$ with a success probability $p_{\mathrm{succ}}=1/(D-d+1)$.
\end{Thm}

\begin{figure}[tbh]
    \centering
    \includegraphics[width=0.7\linewidth]{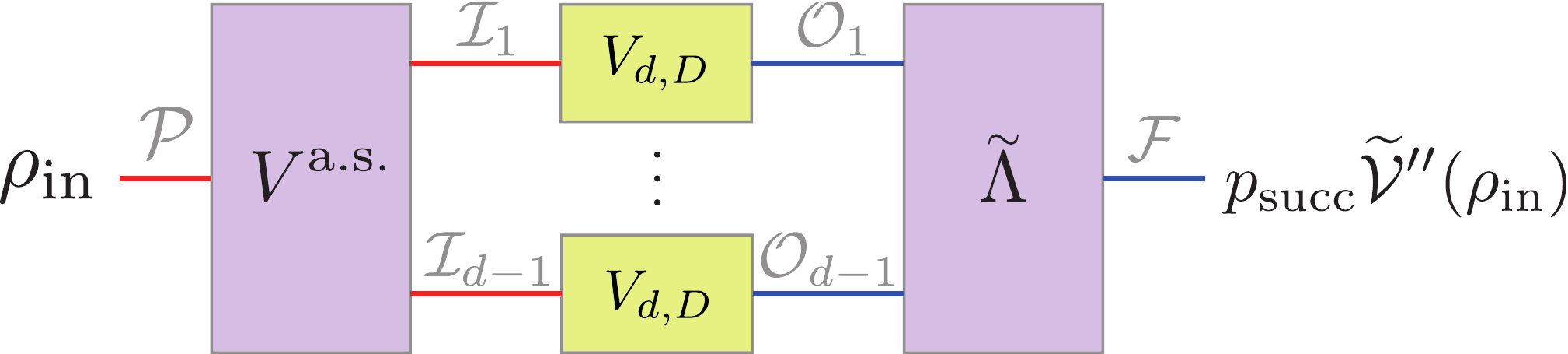}
    \caption{A quantum circuit representation of a  probabilistic parallel protocol to implement the pseudo complex conjugate map $\widetilde{\mathcal{V}}''$. The isometry operator $V^{\mathrm{a.s.}}$ and the CPTP map $\widetilde{\Lambda}$ are defined in Eqs.~(\ref{eq:def_V^as}) and (\ref{eq:def_Lambda}), respectively.}
    \label{fig:isometry_conjugate_modoki}
\end{figure}

Note that the success probability $p_{\mathrm{succ}}$ is less than 1 for $D>d$ while the protocol is composed of deterministic operations, because $\widetilde{\mathcal{V}}''$ is a trace increasing map for $D>d$ (see Eq.~(\ref{eq:def_success_probability}) for the definition of $p_{\mathrm{succ}}$). See Appendix \ref{sec:appendix_isometry_conjugate_modoki} for the proof.

Next, we consider a variant of gate teleportation \cite{Gottesman1999Demonstrating} given by Figure~\ref{fig:gate_teleportation}, which is able to probabilistically implement the transposed map in terms of the computational basis of {\it any} CPTP map $\widetilde{\Lambda}_{\mathrm{in}}$, denoted by $\widetilde{\Lambda}^T_{\mathrm{in}}$. The transposed map $\widetilde{\Lambda}^T_{\mathrm{in}}$ for a CPTP map $\widetilde{\Lambda}_{\mathrm{in}}$ given by its action as $\widetilde{\Lambda}_{\mathrm{in}}(\rho)=\sum_k K_k \rho K_k^\dagger$ in terms of the Kraus operators $\{K_k\}$ is defined as
\begin{align}
    \widetilde{\Lambda}^T_{\mathrm{in}}(\rho)\coloneqq \sum_{k} K_k^T \rho (K_k^T)^\dagger.
\end{align}
In this variant of the gate teleportation, the Hilbert spaces are given by $\mathcal{P}=\mathbb{C}^D$, $\mathcal{O}_1=\mathbb{C}^D$, $\mathcal{F}=\mathbb{C}^d$, $\mathcal{I}_1=\mathbb{C}^d$, and $\ket{\phi'_{\mathrm{PBT}}}$ in the protocol for general $k$ is reduced to a maximally entangled state in $\mathbb{C}^d \otimes \mathbb{C}^d$ defined by
\begin{align}
    \ket{\Phi^+_d}\coloneqq \frac{1}{\sqrt{d}}\sum_{i=0}^{d-1}\ket{i}\otimes \ket{i}\in \mathcal{I}_1\otimes \mathcal{F}.\label{eq:def_max_ent_state_d}
\end{align}
The POVM $\mathcal{M}$ is reduced to the Bell measurement defined as the projective measurement on the basis $\{(X_{D}^{-i}Z_{D}^{-j}\otimes \1_{\mathcal{O}_1})\ket{\Phi^{+}_{D}}_{\mathcal{P}\mathcal{O}_1}\}_{i=0, j=0}^{D-1, D-1}$, where $X_{D}\coloneqq \sum_{j=0}^{D-1}\ketbra{j\oplus 1}{j}$
is the shift operator and $Z_{D}\coloneqq \sum_{j=0}^{D-1}e^{2\pi j\sqrt{-1}/D}\ketbra{j}{j}$ is the clock operator.   The protocol succeeds only when the outcome is given by $(i, j)=(0, 0)$.  The difference from the standard gate teleportation protocol is the Hilbert space where the input operation is applied, namely on $\mathcal{I}_1$, instead of $\mathcal{F}$ in standard gate teleportation \cite{Gottesman1999Demonstrating}.  The property of the maximally entangled state given by $(\1_d \otimes A) \ket{\Phi^+_d} = (A^T \otimes \1_d) \ket{\Phi^+_d}$ for any $A \in \mathcal{L}(\mathbb{C}^d)$ allows teleporting the state transformed by the transposed CPTP map, instead of the CPTP map. 

\begin{figure}[tbh]
    \centering
    \includegraphics[width=0.6\linewidth]{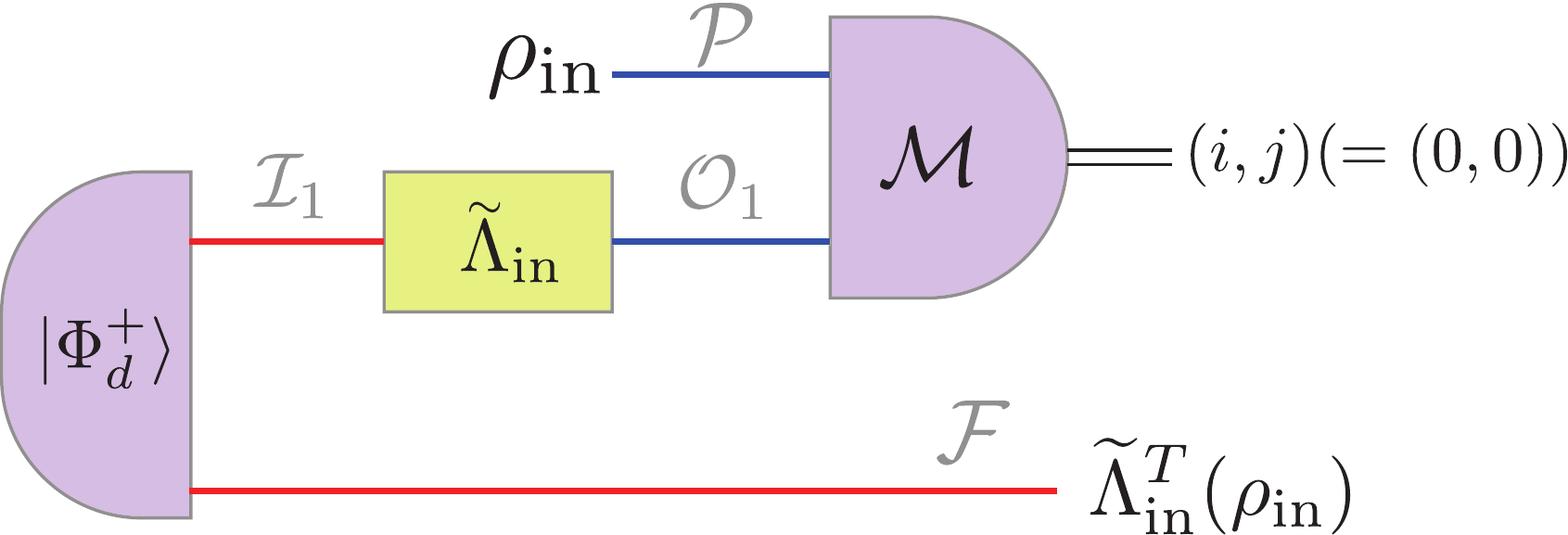}
    \caption{ The gate teleportation circuit implements the transposed map $\widetilde{\Lambda}_{\mathrm{in}}^{T}$ for any CPTP map $\widetilde{\Lambda}_{\mathrm{in}}$ with a success probability $p_{\mathrm{succ}}=1/(Dd)$. The quantum state $\ket{\Phi_d^+}$ is a maximally entangled state $\ket{\Phi_d^+}\coloneqq d^{-1/2}\sum_{i=0}^{d-1}\ket{i}\otimes \ket{i}$ and the POVM $\mathcal{M}$ is the Bell measurement.  This protocol succeeds when the measurement outcome of $\mathcal{M}$ is $(i,j)=(0,0)$.}
    \label{fig:gate_teleportation}
\end{figure}

Concatenating the protocol for  the pseudo complex conjugate map $\widetilde{\mathcal{V}}''$ and the  variant of gate teleportation  shown in Figure~\ref{fig:gate_teleportation}, we obtain  an isometry inversion protocol, whose success probability is $p_{\mathrm{succ}}=1/[Dd(D-d+1)]$. This success probability is less than that of a protocol shown in Figure~\ref{fig:isometry_inversion_protocols}~(c). In general, the success probability of any isometry inversion protocol that is a concatenation of an isometry pseudo complex conjugation protocol and the variant of the gate teleportation  shown in Figure~\ref{fig:gate_teleportation} is bounded by $p_{\mathrm{succ}}\leq 1/(Dd)$ since the success probability of the  variant of the gate teleportation is $1/(Dd)$.

\section{Numerical optimization of success probabilities by semidefinite programming}
\label{sec:SDP}

In the preceding sections, we analytically investigated the success probability of parallel protocols for higher-order quantum transformations of isometry operations.  Though parallel protocol is efficient in terms of the circuit depth compared to other protocols such as sequential ones, more general strategies than parallel protocols can  possibly improve the success probability  by utilizing the temporal resource. To analyze this possibility, we perform a numerical optimization of the success probability of isometry inversion, isometry  (pseudo) complex conjugation, and isometry transposition by semidefinite programming (SDP) in the same way as presented in Ref.~\cite{Quintino2019Probabilistic}  (see Appendix~\ref{appendix:sdp} for the detail). In addition to parallel protocols, we also consider sequential protocols shown in Figure~\ref{fig:adaptive_comb_and_isometry_inversion_sod}~(a) and the most general protocols, which  include the cases with an indefinite causal order  \cite{Oreshkov2012Quantum, Chiribella2013Quantum, Araujo2015Witnessing, Wechs2019Definition, Bisio2019Theoretical, Yokojima2021Consequences, Vanrietvelde2021Routed}.

To construct efficient sequential protocols for isometry inversion, we also consider the ``success-or-draw'' version \cite{Dong2021SuccessorDraw} of an isometry inversion protocol shown in Figure~\ref{fig:adaptive_comb_and_isometry_inversion_sod}~(b). As we define in Section \ref{sec:supermap_superinstrument_higher-order}, we say that a superinstrument $\{\doublewidetilde{\mathcal{S}^{}}, \doublewidetilde{\mathcal{F}^{}}\}$ implements a probabilistic isometry inversion protocol if the element for the success case  $\doublewidetilde{\mathcal{S}^{}}$ satisfies
\begin{align}
    \doublewidetilde{\mathcal{S}^{}}(\widetilde{\mathcal{V}}^{\otimes k})=p\widetilde{\mathcal{V}}_{\mathrm{inv}},
\end{align}
or equivalently,
\begin{align}
    \doublewidetilde{\mathcal{S}^{}}(\widetilde{\mathcal{V}}^{\otimes k})\circ \widetilde{\mathcal{V}}=p\widetilde{\1}_d
\end{align}
for all $V\in \mathbb{V}_{\mathrm{iso}}(d, D)$. No extra condition is required for the element for the fail case $\doublewidetilde{\mathcal{F}^{}}$ as long as $\{\doublewidetilde{\mathcal{S}^{}}, \doublewidetilde{\mathcal{F}^{}}\}$ forms a superinstrument. In general, if the protocol fails, namely, obtaining the outcome corresponding to the fail case $\doublewidetilde{\mathcal{F}^{}}$, the input state $\rho_\mathrm{in}$ is destroyed.  Therefore this type of ``success-or-fail'' probabilistic protocols cannot be repeated or sequentially combined with another protocol.

\begin{figure}[tbh]
    \centering
    \includegraphics[width=0.8\linewidth]{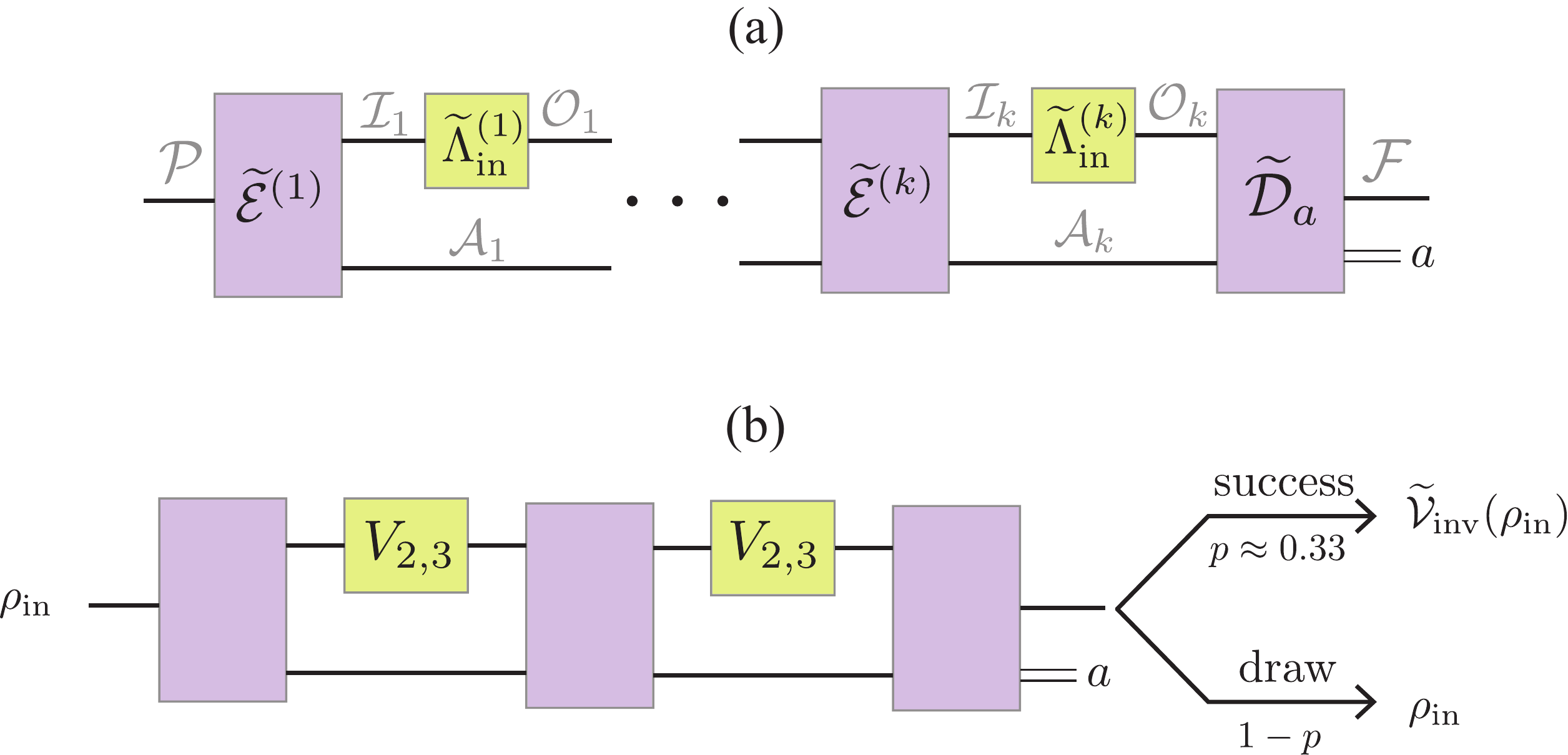}
    \caption{(a) A quantum circuit representation of a sequential superinstrument, where $\widetilde{\Lambda}_{\mathrm{in}}^{(i)}\;(i\in \{1, \cdots, k\})$ are input maps, $\widetilde{\mathcal{E}}^{(i)}\;(i\in\{1, \cdots, k\})$ are CPTP maps, and $\{\widetilde{\mathcal{D}}_a\}$ is a quantum instrument. (b) A quantum circuit representation of a sequential  ``success-or-draw'' protocol for isometry inversion of $V\in \mathbb{V}_{\mathrm{iso}}(2,3)$. For the input quantum state $\rho_{\mathrm{in}}\in \mathcal{L}(\mathrm{Im} V)$, we obtain the output quantum state given by $\widetilde{\mathcal{V}}_{\mathrm{inv}}(\rho_{\mathrm{in}})$ after success or the output quantum state given by $\rho_{\mathrm{in}}$ after a draw. The optimal success probability $p\approx 0.33$ is obtained by the SDP for $d=2$, $D=3$ and $k=2$.}
    \label{fig:adaptive_comb_and_isometry_inversion_sod}
\end{figure}

For the ``success-or-draw'' version of an isometry inversion protocol, we impose an additional condition on $\doublewidetilde{\mathcal{F}^{}}$ given by
\begin{align}
    \doublewidetilde{\mathcal{F}^{}}(\widetilde{\mathcal{V}}^{\otimes k})\circ \widetilde{\mathcal{V}}=(1-p)\widetilde{\mathcal{V}}
\end{align}
for all $V\in \mathbb{V}_{\mathrm{iso}}(d, D)$. This condition means that even after the protocol fails, we recover the original input state as the output of the protocol if the input state is in the image $\mathrm{Im} V$ of the input isometry operator $V\in \mathbb{V}_{\mathrm{iso}}(d, D)$. Such a type of failure is called a \emph{draw} because the protocol can be repeated after failure. By repeating the ``success-or-draw'' protocol until success at most $m$ times, we can achieve the probability given by $p_{\mathrm{succ}}=1-(1-p)^m$. In other words, the failure probability $1-p_{\mathrm{succ}}$ decreases exponentially with the number $k'=mk$ of calls of an input isometry operation.   Therefore, the number of calls of an input isometry operation scales logarithmically with respect to the success probability, i.e., $k' = \mathcal{O}(\ln (1/(1-p_\mathrm{succ})))$. We perform the SDP to obtain such a ``success-or-draw'' protocol for isometry inversion in a similar way as presented in Ref.~\cite{Dong2021SuccessorDraw}  (see Appendix~\ref{appendix:sdp} for the detail).

\begin{table}[tbh]
    \centering
    \begin{tabular}{wc{0.23\linewidth}wc{0.23\linewidth}wc{0.23\linewidth}wc{0.23\linewidth}}\hline\hline
        \multicolumn{4}{c}{isometry inversion}\\\hline
        $d=2, D=3$  & parallel & sequential & general \\\hline
        $k=1$ & \bm{$0.25$} & \bm{$0.25$} & \bm{$0.25$}\\
        $k=2$ & \bm{$0.4$} & $0.4286\approx 3/7$ & $0.4286\approx 3/7$\\\hline\hline
        \multicolumn{4}{c}{isometry complex conjugation}\\\hline
        $d=2, D=3$  & parallel & sequential & general \\\hline
         $k\leq 3$ & 0 & 0 & 0\\\hline\hline
        \multicolumn{4}{c}{isometry pseudo complex conjugation}\\\hline
         $d=2, D=3$  &  parallel &  sequential &  general \\\hline
         $k=1$ &  $0.5000 \approx 1/2$ &  $0.5000 \approx 1/2$ &  $0.5000 \approx 1/2$\\
         $k=2$ &  0.5453 &  0.5453 &  0.5453\\\hline
         $d=2, D=4$  &  parallel &  sequential &  general \\\hline
         $k=1$ &  $0.3333 \approx 1/3$ &  $0.3333 \approx 1/3$ &  $0.3333 \approx 1/3$\\\hline\hline
        \multicolumn{4}{c}{isometry transposition}\\\hline
        $d=2, D=3$  & parallel & sequential & general \\\hline
        $k=1$ & \bm{$1/6$} & \bm{$1/6$} & \bm{$1/6$}\\
        $k=2$ & $0.2857\approx 2/7$ & $0.3077$ & $0.3333\approx 1/3$\\\hline
        $d=2, D=4$  & parallel & sequential & general \\\hline
        $k=1$ & \bm{$0.125$} & \bm{$0.125$} & \bm{$0.125$}\\
        $k=2$ & $0.22\approx 2/9$ & - & -\\
        \hline\hline
    \end{tabular}
    \caption{The optimal success probabilities of isometry inversion,  (pseudo) complex conjugation,  and transposition using $k$ calls of an input isometry operation $\widetilde{\mathcal{V}}$ corresponding to $V\in \mathbb{V}_\mathrm{iso} (d, D)$. The bold values are obtained analytically.}
    \label{tab:SDP}
\end{table}

\begin{table}[tbh]
    \centering
    \begin{tabular}{wc{0.23\linewidth}wc{0.23\linewidth}wc{0.23\linewidth}wc{0.23\linewidth}}\hline\hline
        \multicolumn{4}{c}{unitary inversion}\\\hline
        $d=2$  & parallel & sequential & general \\\hline
        $k=1$ & \bm{$0.25$} & \bm{$0.25$} & \bm{$0.25$}\\
        $k=2$ & \bm{$0.4$} & $0.4286\approx 3/7$ & $0.4444\approx 4/9$\\
        \hline\hline
    \end{tabular}
    \caption{The optimal success probability of unitary inversion using $k$ calls of an input unitary operation $\widetilde{\mathcal{U}}$ corresponding to $U\in \mathbb{U}(d)$  derived in Ref. \cite{Quintino2019Reversing}.}
    \label{tab:prob_unitary_inversion_SDP}
\end{table}

Table~\ref{tab:SDP} and Figure~\ref{fig:adaptive_comb_and_isometry_inversion_sod}~(b) show the results of the SDP optimization. Table \ref{tab:prob_unitary_inversion_SDP} shows the optimal success probability of unitary inversion, which was  already obtained in Ref.~\cite{Quintino2019Reversing}, for the comparison with that of isometry inversion.
We  implement the SDP code in MATLAB \cite{matlab} by modifying the SDP code \cite{unitary_inversion_code} accompanied Refs.~\cite{Quintino2019Reversing, Quintino2019Probabilistic}, originally designed to obtain the optimal success probability of unitary inversion and unitary transposition. The interpreters  CVX \cite{cvx, gb08} and YALMIP  \cite{yalmip, Lofberg2004}  are used with the solvers SDPT3 \cite{sdpt3,toh1999sdpt3,tutuncu2003solving}, SeDuMi \cite{sedumi}, MOSEK \cite{mosek} and  SCS \cite{scs} to perform the SDP.
 The code also uses functions in QETLAB \cite{qetlab}.
 All codes to obtain the results shown in Table~\ref{tab:SDP} and Figure~\ref{fig:adaptive_comb_and_isometry_inversion_sod}~(b) are available at Ref.~\cite{isometry_inversion_code} under the MIT license \cite{mit_license}.

First, we consider isometry inversion. Theorem~\ref{thm:optimal_isometry_inversion_and_PBT} shows that the optimal success probability of probabilistic parallel isometry inversion protocols using $k$ calls of an input isometry operation corresponding to $V\in \mathbb{V}_\mathrm{iso} (d,D)$ is equal to that of  probabilistic parallel unitary inversion protocols using $k$ calls of an input unitary operation corresponding to $U\in \mathbb{U}(d)$.  
This statement can be also checked numerically for $d=2$, $D=3$ and $k\in\{1, 2\}$ by comparing Table~\ref{tab:SDP} with Table \ref{tab:prob_unitary_inversion_SDP}. 
Comparison of Table \ref{tab:SDP} with Table \ref{tab:prob_unitary_inversion_SDP} also shows that the optimal success probability of sequential isometry inversion is the same as that of sequential unitary inversion for $d=2$, $D=3$ and $k=2$. This numerical result implies the possibility that the optimal success probability of $k$-input sequential protocols for isometry inversion of $V\in \mathbb{V}_\mathrm{iso} (d,D)$ does not depend on $D$. 

We also numerically confirmed the existence of a sequential ``success-or-draw'' protocol for isometry inversion for $d=2$, $D=3$ and $k=2$ by the SDP (see Figure~\ref{fig:adaptive_comb_and_isometry_inversion_sod}~(b)). By repeating this ``success-or-draw'' protocol, we can implement a sequential protocol for isometry inversion with the success probability scaling as $p_{\mathrm{succ}}=1-\exp[-\mathcal{O}(k)]$. In contrast, we can show the upper bound of the success probability of any parallel protocol for isometry inversion given by $p_{\mathrm{succ}}\leq 1-\mathcal{O}(k^{-1})$ from the same upper bound for unitary inversion presented in Ref. \cite{Quintino2019Reversing}. Thus an exponential improvement of the success probability of a sequential protocol for isometry inversion compared to a parallel protocol for $d=2$ and $D=3$ is exhibited. However, a further general protocol does not improve the success probability of isometry inversion compared to a sequential protocol for $d=2$, $D=3$ and $k=2$ as shown numerically in Table~\ref{tab:SDP}, whereas an improvement with a general protocol was observed for the case of $d=2$ and $k=2$ unitary inversion as shown in Table \ref{tab:prob_unitary_inversion_SDP}.

Next, we consider isometry complex conjugation. Theorem~\ref{thm:isometry_complex_conjugation_nogo} shows that it is impossible to transform finite calls of an isometry operation $\widetilde{\mathcal{V}}: \mathcal{L}(\mathbb{C}^d)\to\mathcal{L}(\mathbb{C}^D)$ into its complex conjugate map $\widetilde{\mathcal{V}}^*$ with a non-zero success probability for $D\geq 2d$. In addition to this theorem, the numerical result in Table~\ref{tab:SDP} indicates that it is also impossible for $d=2, D=3$ and $k \in \{1, 2, 3 \}$, although $D\geq 2d$ is not satisfied. This result implies the possibility that probabilistic isometry complex conjugation is impossible even for $d<D<2d$.

 Then, we consider isometry pseudo complex conjugation. The numerical result shown in Table~\ref{tab:SDP} indicates that the optimal success probability of isometry pseudo complex conjugation is equal to the success probability $p_{\mathrm{succ}}=1/(D-d+1)$ of a parallel protocol shown in Figure~\ref{fig:isometry_conjugate_modoki} for $d=2$, $D\in\{3,4\}$ and $k=d-1=1$. The numerical result also shows that a parallel protocol achieves the optimal success probability of isometry pseudo complex conjugation among general protocols including indefinite causal order for $d=2$, $D=3$ and $k=2$.

Finally, we consider isometry transposition  (see also Appendix \ref{sec:appendix_isometry_tranpose_PBT}). The numerical result shown in Table~\ref{tab:SDP} indicates that the optimal success probability of parallel isometry transposition is equal to the success probability $p_{\mathrm{succ}}=k/(Dd+k-1)$ of a protocol shown in Figure~\ref{isometry_transpose_PBT} for $d=2$, $D\in \{3, 4\}$ and $k=2$. This result implies the possibility that a protocol shown in Figure~\ref{isometry_transpose_PBT} achieves the optimal success probability for any choice of $d$, $D$ and $k$. The numerical result also shows that a sequential protocol and a general protocol can improve the success probability of isometry transposition in contrast to the case of isometry inversion.

We present a conjecture and open problems obtained from the discussion for the SDP optimization results.

{\it Conjecture}: 
The optimal success probability of probabilistic parallel protocols that transform $k$ calls of an isometry operation $\widetilde{\mathcal{V}}$ corresponding to $V\in \mathbb{V}_\mathrm{iso} (d,D)$ into its transposed map $\widetilde{\mathcal{V}}^T$ is $p_{\mathrm{succ}}=k/(Dd+k-1)$, which is achieved by the protocol shown in Figure~\ref{isometry_transpose_PBT}.

{\it Open problem 1}: Does the optimal success probability of sequential isometry inversion depend on $D$?

{\it Open problem 2}: Does indefinite causal order improve the success probability of isometry inversion?

{\it Open problem 3}: Is it possible to transform finite calls of an isometry operation $\widetilde{\mathcal{V}}$ corresponding to $V\in \mathbb{V}_\mathrm{iso} (d,D)$ into its complex conjugate map $\widetilde{\mathcal{V}}^*$ with a non-zero success probability for $d<D<2d$?

\section{Conclusion}
\label{sec:conclusion}
We presented a probabilistic exact parallel protocol for isometry inversion that constructs a decoder from multiple calls of a black box encoder implementing an unknown isometry operation  transforming a $d$-dimensional system to a $D$-dimensional system  for $D>d$. {The success probability of this protocol is independent of $D$.} Thus, this protocol  significantly outperforms other isometry inversion protocols that use $D$-dimensional unitary inversion protocols or quantum process tomography of isometry operations  for $D \gg d$. This shows a potential of our protocol for applications in quantum information processing  involving encoding and decoding with black boxes. 
In particular, we consider a typical example of an encoding black box represented by an isometry operation to spread quantum information of a qudit (a $d$-dimensional system) into a $n$-qudit system (a $d^n$-dimensional system). It may seem that  inverting the function from a black box function is difficult due to the exponential dimensionality of $d^n$ of isometry operations, but our result shows that such an implementation is easy whenever $d$ is small enough.

We developed a new technique to construct isometry inversion protocols, since the strategy for unitary inversion used in the previous work \cite{Quintino2019Reversing} is not applicable to isometry inversion.  This is because the probabilistic exact isometry complex conjugation is impossible for $D \geq 2d$, and  the concatenation of isometry pseudo complex conjugation and a variant of the gate teleportation is inefficient. Then, we showed the decomposition of the tensor product of an isometry and invent a CPTP map $\widetilde{\Psi}$ that can be implemented by the quantum Schur transform \cite{harrow2005applications, bacon2006efficient, Krovi2019Efficient} and the ``measure-and-prepare'' strategy \cite{horodecki2003entanglement, Bisio2010Optimal}. This CPTP map transforms $k+1$ parallel calls of an isometry operation into the Haar integral of $k+1$ tensor product of unitary operations as shown in Eq.~(\ref{eq:Psi_U}) of Lemma~\ref{lem:Lambda}, which contributes to keeping the dimension dependence of the success probability to $d$ and independent of $D$. This technique provides a new application of the quantum Schur transform, which is known to have various applications to quantum information processing \cite{harrow2005applications}.

We also performed the SDP to investigate the improvement of the success probability of a sequential protocol or a protocol with indefinite causal order for isometry inversion, isometry complex conjugation and isometry transposition compared to a parallel protocol. From the numerical calculation, we found a ``success-or-draw'' isometry inversion protocol for $d=2$ and $D=3$. By repeating this protocol, we can obtain a sequential protocol for isometry inversion with a failure probability decreasing exponentially with the number of calls of the input isometry operation. This result exhibits an exponential improvement of the success probability of a sequential protocol for isometry inversion for $d=2$ and $D=3$ compared to a parallel protocol.

\begin{acknowledgments}
The authors thank M.~T.~Quintino for inspiring discussions, especially for pointing out the relationship between the uniqueness of isometry transposition and the optimality of the success probability.
 We also thank an anonymous reviewer of Quantum for providing a resource evaluation of process tomography of isometry operations.
This work was supported by MEXT Quantum Leap Flagship Program (MEXT QLEAP) Grant No. JPMXS0118069605, JPMXS0120351339, Japan Society for the Promotion
of Science (JSPS) KAKENHI Grant No. 18K13467, 21H03394 and the Forefront Physics and Mathematics Program to Drive Transformation (FoPM) of the University of Tokyo.
\end{acknowledgments}

\appendix

\section{Higher-order quantum transformations}
\label{sec:appendix_supermap}
Choi-Jamiolkowski (CJ) isomorphism is  a useful tool for treating linear maps as operators called Choi operators by  ``lowering'' the order  of transformations.  Similarly, the CJ isomorphism is used for describing higher-order quantum transformations by Choi operators of higher-order quantum transformation.  We summarize the formulation of linear maps and higher-order quantum transformations using the Choi operators.

Any linear map $\widetilde{\Lambda}: \mathcal{L}(\mathcal{I})\to \mathcal{L}(\mathcal{O})$ can be represented uniquely by the Choi operator defined by
\begin{align}
    J_{\widetilde{\Lambda}}\coloneqq \sum_{i, j} \ketbra{i}{j}_{\mathcal{I}} \otimes \widetilde{\Lambda}(\ketbra{i}{j})_{\mathcal{O}} \in \mathcal{L}(\mathcal{I}\otimes \mathcal{O}),
\end{align}
where $\{\ket{i}\}$ is an orthonormal basis in $\mathcal{I}$.
The map $\widetilde{\Lambda}$ can be obtained from $ J_{\widetilde{\Lambda}}$ as
\begin{align}
    \widetilde{\Lambda}(\rho_{\mathcal{I}})=\Tr_{\mathcal{I}}[J_{\widetilde{\Lambda}}(\rho_{\mathcal{I}}^{T} \otimes \1_{\mathcal{O}})],
\end{align}
where $\rho_{\mathcal{I}}^{T}$ is the transposition of an input quantum state $\rho_{\mathcal{I}}$ in the computational basis.

A CPTP map $\widetilde{\Lambda}$ can be characterized in terms of its  Choi operator.
First, we consider a CP map $\widetilde{\Lambda}$ and let $\{K_k \}: \mathcal{I}\to \mathcal{O}$ be its Kraus operators, i.e., $\widetilde{\Lambda}(\rho)=\sum_{{k}}K_{k} \rho K_{k}^{\dagger}$. Then, its  Choi operator is given by
\begin{align}
    J_{\widetilde{\Lambda}}=\sum_{k}\dketbra{K_{k}},
    \label{eq:CP_choi}
\end{align}
where $\dket{K_k}$ is a vector representing a rank-1 Choi vector defined as  $\dket{K_{k}} \coloneqq \sum_{i} \ket{i}_{\mathcal{I}} \otimes (K_k \ket{i})_{\mathcal{O}}\in \mathcal{I}\otimes \mathcal{O}$, also referred to as a dual ket vector. Since $J_{\widetilde{\Lambda}}$ is represented by a sum of rank-1 positive operators, it is a positive  operator. Conversely, if $J_{\widetilde{\Lambda}}$ is positive, $J_{\widetilde{\Lambda}}$ can be written in the form of Eq.~(\ref{eq:CP_choi}). Then, $\widetilde{\Lambda}(\rho)=\sum_{{k}}K_{k}\rho K_{k}^{\dagger}$, which means that $\widetilde{\Lambda}$ is a CP map. Therefore, $\widetilde{\Lambda}$ is CP if and only if its  Choi operator $J_{\widetilde{\Lambda}}$ is positive. 
Next, a map $\widetilde{\Lambda}$ is TP if and only if $\Tr_{\mathcal{O}}\widetilde{\Lambda}(\ketbra{i}{j})=\delta_{i,j}$ for all $i, j$. Therefore, $\widetilde{\Lambda}$ is TP if and only if $\Tr_{\mathcal{O}} J_{\widetilde{\Lambda}}=I_{\mathcal{I}}$. Similarly, a quantum instrument is represented by a set of CP maps $\{\widetilde{\Lambda}_a \}$ and it can be also characterized by the corresponding set of Choi operators $\{J_{\widetilde{\Lambda}_a} \}$.

The composition of two maps can be represented by a link product denoted by $\star$. Let $X\in \mathcal{L}(\mathcal{X}\otimes \mathcal{Y})$ and $Y\in \mathcal{L}(\mathcal{Y}\otimes \mathcal{Z})$. The link product of $X$ and $Y$ is defined as
\begin{align}
    Y\star X\coloneqq \Tr_{\mathcal{H}}[(\1_{\mathcal{X}}\otimes Y)(X^{T_{\mathcal{Y}}}\otimes \1_{\mathcal{Z}})],
\end{align}
where $X^{T_{\mathcal{Y}}}$ is the partial transpose of $X$ on $\mathcal{Y}$. We consider two maps $\widetilde{\Lambda}_1:\mathcal{L}(\mathcal{I}_1)\to \mathcal{L}(\mathcal{I}_2)$ and $\widetilde{\Lambda}_2:\mathcal{L}(\mathcal{I}_2)\to \mathcal{L}(\mathcal{O}_1)$. The  Choi operator of $\widetilde{\Lambda}_2\circ \widetilde{\Lambda}_1$ is obtained by
\begin{align}
    J_{\widetilde{\Lambda}_2\circ \widetilde{\Lambda}_1}=J_{\widetilde{\Lambda}_2}\star J_{\widetilde{\Lambda}_1}.
\end{align}

Similarly to CPTP maps and quantum instruments, higher order-quantum transformations can be represented by Choi operators.   In particular, we use the Choi operator representation of a single-input superinstrument for the analysis in this paper.  A single-input superinstrument $\{\doublewidetilde{\mathcal{C}^{}}_{a}\}$ can be characterized by the corresponding set of Choi operators $\{C_a\}$ as
\begin{align}
    C_a&\geq 0,\label{eq:parallel_comb_cond1}\\
    \Tr_{\mathcal{F}}C&=\Tr_{\mathcal{FO}}C\otimes \frac{\1_{\mathcal{O}}}{d_{\mathcal{O}}},\label{eq:parallel_comb_cond2}\\
    \Tr_{\mathcal{FOI}}C&=\Tr C\otimes \frac{\1_{\mathcal{P}}}{d_{\mathcal{P}}},\label{eq:parallel_comb_cond3}\\
    \Tr C&=d_{\mathcal{P}}d_{\mathcal{O}},\label{eq:parallel_comb_cond4}
\end{align}
where $C$ is defined as $C\coloneqq \sum_{{a}} C_a$ \cite{Chiribella2008Transforming}. The Choi operators $\{C_a\}$ are related to the corresponding single-input superinstrument $\{\doublewidetilde{\mathcal{C}^{}}_{a}\}$ as
\begin{align}
    C_a\star J_{\widetilde{\Lambda}_\mathrm{in}}=J_{\doublewidetilde{\mathcal{C}^{}}_{a}(\widetilde{\Lambda}_\mathrm{in})}.
\end{align}
See Ref.~\cite{Araujo2017Purification} for the characterization of $k$-input superinstruments.

\section{{Extension of the Schur-Weyl duality to the decomposition of the tensor product of {isometry operators}}}
\label{sec:schur-weyl_duality}
{We first review the Schur-Weyl duality.}
We consider Hilbert spaces $\mathcal{X}_i=\mathbb{C}^d$ for $i\in \{1, \cdots, k\}$ and define the joint Hilbert space by $\mathcal{X}\coloneqq \bigotimes _{i=1}^k \mathcal{X}_i$. We consider representations of the unitary group $\mathbb{U}(d)$ and the symmetric group $\mathfrak{S}_k$ defined as
\begin{align}
    \mathbb{U}(d)\to &\mathcal{L}(\mathcal{X}) ;\;\;\; U\mapsto U^{\otimes k},\\
    \mathfrak{S}_k\to &\mathcal{L}(\mathcal{X}) ;\;\;\; \sigma\mapsto P_{\sigma}\label{eq:def_W},
\end{align}
where $P_{\sigma}$ is the permutation operator defined by 
\begin{align}
    P_{\sigma}\left(\bigotimes _{i=1}^k\ket{\psi_{i}}\right)=\bigotimes_{i=1}^k \ket{\psi_{\sigma^{-1}(i)}}.
\end{align}
These representations can be decomposed as
\begin{align}
    \mathcal{X}&= \bigoplus_{\mu\vdash k} \mathcal{U}_{\mu, \mathcal{X}}^{(d)}\otimes \mathcal{S}_{\mu, \mathcal{X}}^{(k)},\label{eq:schur_weyl}\\
    U^{\otimes k}&=\bigoplus_{\mu\vdash k} U_{\mu}\otimes \1_{\mathcal{S}_{\mu, \mathcal{X}}^{(k)}},\\
    P_{\sigma}&=\bigoplus_{\mu\vdash k} \1_{\mathcal{U}_{\mu, \mathcal{X}}^{(d)}}\otimes P_{\sigma, \mu},
\end{align}
where the summands are indexed by the Young diagrams $\mu$ with $k$ boxes,  $\mathbb{U}(d) \ni U \mapsto U_{\mu} \in \mathcal{L}(\mathcal{U}_{\mu, \mathcal{X}}^{(d)})$ are irreducible representations of $\mathbb{U}(D)$,  $\mathfrak{S}_k \ni \sigma \mapsto P_{\sigma, \mu} \in \mathcal{L}(\mathcal{S}_{\mu, \mathcal{X}}^{(k)})$ are irreducible representations of $\mathfrak{S}_k$ and $\1_{\mathcal{S}_{\mu, \mathcal{X}}^{(k)}}$ and $\1_{\mathcal{U}_{\mu, \mathcal{X}}^{(d)}}$ are the identity operators on $\mathcal{S}_{\mu, \mathcal{X}}^{(k)}$ and $\mathcal{U}_{\mu, \mathcal{X}}^{(d)}$, respectively \cite{Iwahori1978Representation}.
The dimension of $\mathcal{U}_{\mu, \mathcal{X}}^{(d)}$ is non-zero if and only if $\mu$ has at most $k$ rows, and {$\mathcal{U}_{\mu, \mathcal{X}}^{(d)}$ and} $\mathcal{S}_{\mu, \mathcal{X}}^{(k)}$ {are} spanned by {bases} called {the Gel'fand-Zetlin basis and the} Young orthonormal basis \cite{Krovi2019Efficient, sagan2001symmetric}{, respectively}. Each element in {the Gel'fand-Zetlin basis and} the Young orthonormal basis is labeled by {a semi-standard tableau $u_\mu$ and} a standard tableau {$s_\mu$} whose frame is $\mu${, respectively. In total, the Hilbert space $\mathcal{X}$ is spanned by the set of vectors $\{\ket{\mu, u_\mu, s_\mu}\}$, which is called the Schur basis.}

We extend the Schur-Weyl duality to show the decomposition of the tensor product of isometry {operators}. We consider Hilbert spaces $\mathcal{Y}_i= \mathbb{C}^D$ for $i\in \{1, \cdots, k\}$ and define the joint Hilbert space $\mathcal{Y}$ {by $\mathcal{Y} \coloneqq \bigotimes_{i=1}^{k}\mathcal{Y}_i$}. We consider the tensor product $V^{\otimes k}: \mathcal{X}\to \mathcal{Y}$ of an isometry {operator} $V\in \mathbb{V}_{\mathrm{iso}}(d,D)$. We decompose $V^{\otimes k}$ in the Schur basis as
\begin{align}
    V^{\otimes k}=\bigoplus_{\substack{\mu, \mu' \vdash k\\l(\mu)\leq d, l(\mu')\leq D}}  \sum_{\alpha} A^\alpha_{\mu, \mu'}\otimes B^\alpha_{\mu, \mu'},
\end{align}
where $l(\mu)$ is the number of rows of a Young diagram $\mu$, $\{A^{\alpha}_{\mu, \mu'}\}$ is a basis of the set of linear operators $\mathcal{L}(\mathcal{U}^{(d)}_{\mu, \mathcal{X}}\to \mathcal{U}^{(D)}_{\mu', \mathcal{Y}})$ and $B^{\alpha}_{\mu, \mu'}$ is a linear operator $B^{\alpha}_{\mu, \mu'}: \mathcal{S}^{(k)}_{\mu, \mathcal{X}}\to \mathcal{S}^{(k)}_{\mu', \mathcal{Y}}$.
Since $V^{\otimes k}$ is  invariant under the action of the symmetric group $\mathfrak{S}_k$, i.e.,
\begin{align}
    P_{\sigma, \mathcal{Y}}^{\dagger} V^{\otimes k} P_{\sigma, \mathcal{X}}=V^{\otimes k}
\end{align}
holds for all $\sigma\in \mathfrak{S}_k$,
\begin{align}
    P_{\sigma, \mu'}^{\dagger} B^{\alpha}_{\mu, \mu'} P_{\sigma, \mu}=B^\alpha_{\mu, \mu'}
\end{align}
holds for all $\mu, \mu'\vdash k$ and $\sigma\in \mathfrak{S}_k$. From Schur's lemma, if $\mu=\mu'$ (i.e., the irreducible representations $\mathcal{S}^{(k)}_{\mu, \mathcal{X}}$ and $\mathcal{S}^{(k)}_{\mu', \mathcal{Y}}$ are unitarily equivalent), $B^{\alpha}_{\mu, \mu'}$ is the isomorphism between the irreducible representations $\mathcal{S}^{(k)}_{\mu, \mathcal{X}}$ and $\mathcal{S}^{(k)}_{\mu', \mathcal{Y}}$, and if $\mu\neq \mu'$, $B^{\alpha}_{\mu, \mu'}=0$. Note that the isomorphism between the irreducible representations of the symmetric group is unique up to a constant multiplication.
Thus, we obtain the decomposition of $V^{\otimes k}$ given by
\begin{align}
    V^{\otimes k}=\bigoplus_{\substack{\mu\vdash k\\l(\mu)\leq d}} V_{\mu}\otimes I_{\mathcal{S}_{\mu, \mathcal{X}}^{(k)}\to \mathcal{S}_{\mu, \mathcal{Y}}^{(k)}},
    \label{eq:isometry_schur_weyl}
\end{align}
where $V_{\mu}\in \mathcal{L}(\mathcal{U}_{\mu, \mathcal{X}}^{(d)}\to \mathcal{U}_{\mu, \mathcal{Y}}^{(D)})$ is a linear operator and $I_{\mathcal{S}_{\mu, \mathcal{X}}^{(k)}\to \mathcal{S}_{\mu, \mathcal{Y}}^{(k)}}$ is the isomorphism between irreducible representations $\mathcal{S}_{\mu, \mathcal{X}}^{(k)}$ and $\mathcal{S}_{\mu, \mathcal{Y}}^{(k)}$, which transforms a basis vector of $\mathcal{S}_{\mu, \mathcal{X}}^{(k)}$ corresponding to a standard tableau into a basis vector of $\mathcal{S}_{\mu, \mathcal{Y}}^{(k)}$ corresponding to the same standard tableau. Since $V^{\otimes k}$ and $I_{\mathcal{S}_{\mu, \mathcal{X}}^{(k)}\to \mathcal{S}_{\mu, \mathcal{Y}}^{(k)}}$ are isometry operators, $V_{\mu}$ is also an isometry operator.

Note that the decomposition of the tensor product of isometric extension $V_{\widetilde{\Lambda}}: \mathcal{A}\to \mathcal{B}\otimes \mathcal{E}$ of a quantum channel $\widetilde{\Lambda}: \mathcal{L}(\mathcal{A})\to \mathcal{L}(\mathcal{B})$ is discussed in Ref.~\cite{harrow2005applications}. {Reference} \cite{harrow2005applications} uses the basis of the input space $\mathcal{A}^{\otimes k}$ as the Schur basis, and the basis of the output space $(\mathcal{B}\otimes \mathcal{E})^{\otimes k}$ as the tensor product of the Schur bases of $\mathcal{B}^{\otimes k}$ and $\mathcal{E}^{\otimes k}$ to represent the tensor product $V_{\widetilde{\Lambda}}^{\otimes k}$. However, the expression of $V_{\widetilde{\Lambda}}^{\otimes k}$ is not block diagonal in that basis. In contrast, we show the block diagonal decomposition of the tensor product $V^{\otimes k}$ of an isometry operator $V$ as shown in Eq.~(\ref{eq:isometry_schur_weyl}).

\section{Haar measure on the unitary group}
The Haar measure $\dd U$ is the uniform measure defined on the set of unitary operators $\mathbb{U}(d)$. More precisely, it is uniquely determined by the following properties \cite{Kobayashi2005Lie}:
\begin{align}
    \int \dd U&=1,\\
    \dd (U'UU'')&=\dd U\;\;\;(\forall U', U''\in \mathbb{U}(D)).\label{eq:uniformity}
\end{align}

We consider the action of the Haar random unitary operations on a quantum state.
Suppose $\mathcal{X}_i= \mathbb{C}^d$ for $i\in\{1, \cdots,k\}$ and define the joint Hilbert space by $\mathcal{X}=\bigotimes_{i=1}^{k} \mathcal{X}_i$. For a quantum state $\rho\in \mathcal{L}(\mathcal{X})$, we define a quantum state $\rho'\in\mathcal{L}(\mathcal{X})$ by
\begin{align}
    \rho'\coloneqq \int \dd U U^{\otimes k}\rho U^{\dagger \otimes k},
\end{align}
where $\dd U$ is the Haar measure on $\mathbb{U}(d)$. From Eq.~(\ref{eq:uniformity}), we obtain
\begin{align}
    [U^{\otimes k}, \rho']=0\;\;\;(\forall U\in \mathbb{U}(D)),
\end{align}
where $[A,B]\coloneqq AB-BA$ is a commutator. By Schur's lemma, the operator $\rho'$ satisfies
\begin{align}
    \rho'=\bigoplus_{\mu\vdash k} \frac{I_{\mathcal{U}_{\mu, \mathcal{X}}^{(d)}}}{d_{\mathcal{U}_{\mu}^{(d)}}}\otimes X_{\mathcal{S}^{(k)}_{\mu, \mathcal{X}}},\label{eq:haar_unitary_1}
\end{align}
where $X_{\mathcal{S}^{(k)}_{\mu, \mathcal{X}}}$ is a positive  operator on $\mathcal{S}^{(k)}_{\mu, \mathcal{X}}$. The operator $X_{\mathcal{S}^{(k)}_{\mu, \mathcal{X}}}$ is calculated as
\begin{align}
    X_{\mathcal{S}^{(k)}_{\mu, \mathcal{X}}}
    &=\Tr_{\mathcal{U}^{(d)}_{\mu, \mathcal{X}}}(\Pi_{\mu, \mathcal{X}}\rho')\\
    &=\int \dd U \Tr_{\mathcal{U}^{(d)}_{\mu, \mathcal{X}}}
    \left\{
    \Pi_{\mu, \mathcal{X}} \left[
    \widetilde{\mathcal{U}}_{\mu}\otimes \widetilde{\1}_{\mathcal{S}^{(k)}_{\mu, \mathcal{X}}}(\rho)
    \right]
    \right\}\\
    &=\int \dd U \Tr_{\mathcal{U}^{(d)}_{\mu, \mathcal{X}}}
    \left\{
    \left[
    \widetilde{\mathcal{U}}^{\dagger}_{\mu}\otimes \widetilde{\1}_{\mathcal{S}^{(k)}_{\mu, \mathcal{X}}}(\Pi_{\mu, \mathcal{X}})
    \right] 
    \rho
    \right\}\\
    &=\int \dd U \Tr_{\mathcal{U}^{(d)}_{\mu, \mathcal{X}}}(\Pi_{\mu, \mathcal{X}} \rho)\\
    &=\Tr_{\mathcal{U}^{(d)}_{\mu, \mathcal{X}}}(\Pi_{\mu, \mathcal{X}} \rho),\label{eq:haar_unitary_2}
\end{align}
where $\Pi_{\mu, \mathcal{X}}$ is a projector from the Hilbert space $\mathcal{X}$ to its subspace $\mathcal{U}_{\mu, \mathcal{X}}^{(d)}\otimes \mathcal{S}_{\mu, \mathcal{X}}^{(k)}$. 

\section{The parallel unitary inversion protocol}
\label{sec:appendix_unitary_inversion}

\begin{figure}[tb]
    \centering
    \includegraphics[width=0.7\linewidth]{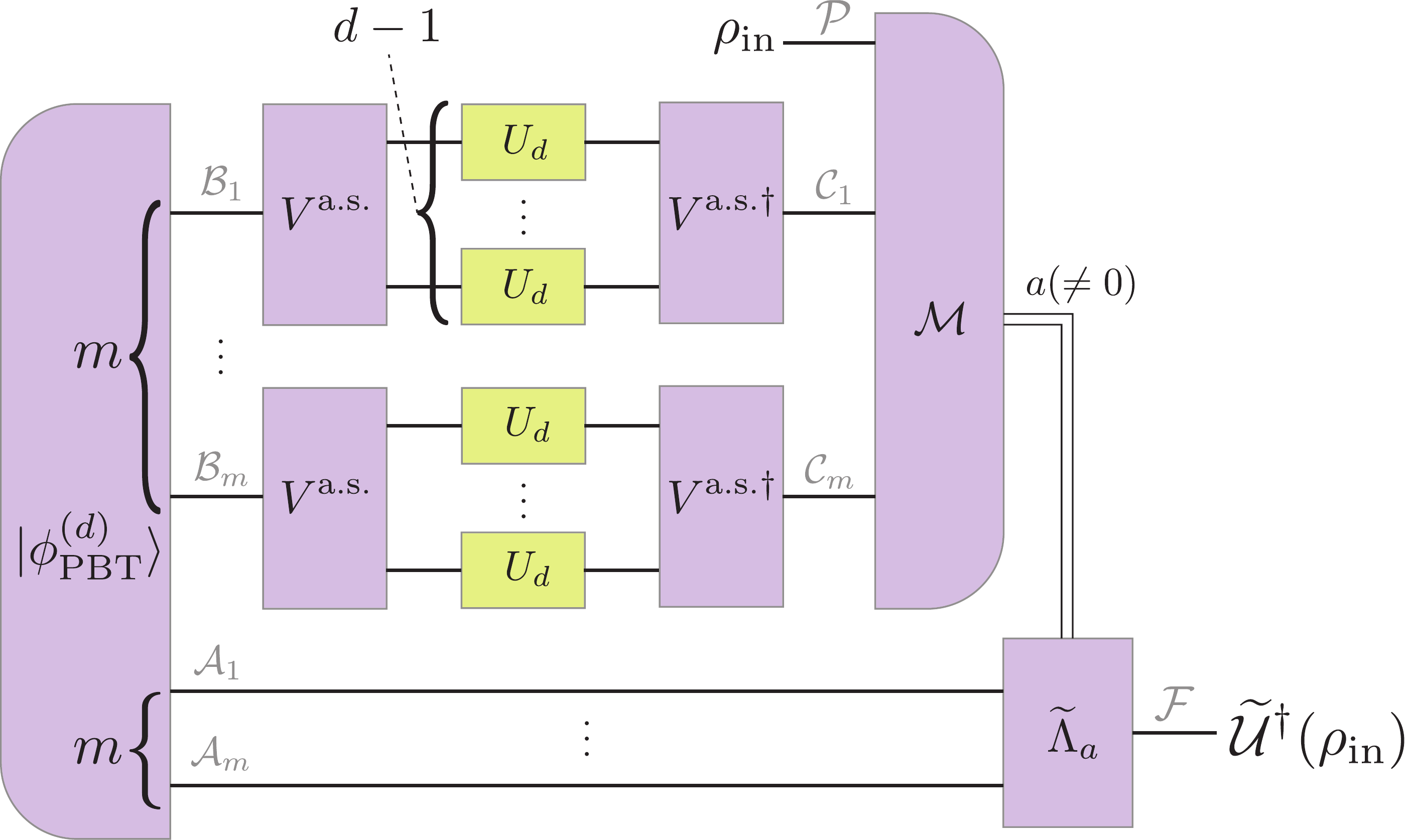}
    \caption{A quantum circuit representation of a parallel protocol for unitary inversion from $k=m(d-1)$ calls of $U_d\in \mathbb{U}(d)$ \cite{Quintino2019Reversing}. Each wire corresponds to a $d$-dimensional system. The quantum state $ \ket{\phi^{(d)}_{\mathrm{PBT}}}$ and the POVM $\mathcal{M}=\{\Gamma_a^{(d)}\}_{a=0}^{k}$ are the optimal resource state and the POVM for the probabilistic port-based teleportation \cite{Ishizaka2008Asymptotica,Studzinski2017Portbased}, which are defined in Eqs.~(\ref{eq:def_phi_PBT}) and (\ref{eq:def_gamma_a^d}), respectively. The isometry operator $V^{\mathrm{a.s.}}$ represents an encoding of quantum information on a totally antisymmetric state defined in Eq.~(\ref{eq:def_V^as}). The conditional CPTP map $\widetilde{\Lambda}_a$ is the operation to select the quantum state in $\mathcal{A}_a$ corresponding to the measurement outcome $a$ of $\mathcal{M}$ as the output state for $a\neq 0$, which is defined in Eq.~(\ref{eq:def_V_a}). This protocol succeeds when the measurement outcome $a$ is $a\neq 0$.}
    \label{fig:unitary_inverse_parallel}
\end{figure}

We show a quantum circuit representation of a $k$-input unitary inversion protocol presented in Ref.~\cite{Quintino2019Reversing} (See Figure~\ref{fig:unitary_inverse_parallel}). This protocol achieves a success probability $p_{\mathrm{succ}}=\lfloor k/(d-1) \rfloor/[d^2+\lfloor k/(d-1) \rfloor -1]$. In Figure~\ref{fig:unitary_inverse_parallel}, each wire corresponds to a $d$-dimensional system. We define the joint Hilbert spaces by $\mathcal{A}\coloneqq \bigotimes_{i=1}^{m} \mathcal{A}_i$, $\mathcal{B}\coloneqq \bigotimes_{i=1}^{m} \mathcal{B}_i$, $\mathcal{C}\coloneqq \bigotimes_{i=1}^{m} \mathcal{C}_i$ and $\overline{\mathcal{C}}_a\coloneqq \bigotimes_{i\neq a} \mathcal{C}_i$ for $a\in \{1, \cdots, m\}$. The quantum state $ \ket{\phi^{(d)}_{\mathrm{PBT}}}\in \mathcal{A}\otimes \mathcal{B}$ and the POVM $\mathcal{M}=\{\Gamma_a^{(d)}\}_{a=0}^{k}$ are the optimal resource state and the POVM for the probabilistic port-based teleportation \cite{Ishizaka2008Asymptotica,Studzinski2017Portbased}, respectively, which are defined by
\begin{align}
    &\ket{\phi^{(d)}_{\mathrm{PBT}}}
    \coloneqq
    (X'^{\frac{1}{2}}_{\mathcal{B}}\otimes \1_{\mathcal{A}})\ket{\Phi^+_{d^k}}_{\mathcal{BA}},\label{eq:def_phi_PBT}\\
    &\Gamma^{(d)}_a \coloneqq (\1_{\mathcal{P}}\otimes X'^{-\frac{1}{2}}_{\mathcal{C}})\left(\ketbra{\Phi^+_d}{\Phi^+_d}_{\mathcal{P}\mathcal{C}_a}\otimes \Theta'_{\overline{\mathcal{C}}_{a}}\right)(\1_{\mathcal{P}}\otimes X'^{-\frac{1}{2}}_{\mathcal{C}}),\label{eq:def_gamma_a^d}
\end{align}
where $\Theta'_{\overline{\mathcal{C}}_{a}}$ and $X'_{\mathcal{B}}$ are defined similarly to Eqs.~(\ref{eq:theta}) and (\ref{eq:X}) as
\begin{align}
    \Theta'_{\overline{\mathcal{C}}_{a}} &\coloneqq \sum_{\alpha\vdash  m-1} \frac{d^{m+1}g_d(m)d_{\mathcal{U}^{(d)}_{\alpha}}}{m d_{\mathcal{S}^{(m-1)}_{\alpha}}}\Pi_{\alpha, \overline{\mathcal{C}}_{a}},\\
    X'_{\mathcal{B}}&\coloneqq \sum_{\mu\vdash m}\frac{d^m g_d(m)d_{\mathcal{U}^{(d)}_{\mu}}}{d_{\mathcal{S}^{(m)}_{\mu}}}\Pi_{\mu, \mathcal{B}},
\end{align}
and $\ket{\Phi^+_{d^k}}$ and $\ket{\Phi^+_{d}}$ are the maximally entangled states defined in Eqs.~(\ref{eq:def_max_ent_state_d^k}) and (\ref{eq:def_max_ent_state_d}).
The conditional CPTP map $\widetilde{\Lambda}_a$ and the isometry operator $V^{\mathrm{a.s.}}$ are defined in Eqs.~(\ref{eq:def_V_a}) and (\ref{eq:def_V^as}), respectively.

\section{Parallel isometry transposition}
\label{sec:appendix_isometry_tranpose_PBT}

 \subsection{Construction of parallel isometry transposition protocol}
We can construct a parallel protocol for isometry transposition similarly  to unitary transposition \cite{Quintino2019Probabilistic} based on the port-based teleportation \cite{Ishizaka2008Asymptotica, Studzinski2017Portbased} and obtain the success probability of the protocol as stated in the following Theorem.

\begin{Thm}
\label{theorem:isometry_transpose_PBT}
A parallel protocol shown in Figure~\ref{isometry_transpose_PBT} transforms $k$ calls of an isometry operation $\widetilde{\mathcal{V}}$ corresponding to $V \in \mathbb{V}_\mathrm{iso} (d, D)$ into its transposed map $\widetilde{\mathcal{V}}^T$ with a success probability $p_{\mathrm{succ}}=k/(Dd+k-1)$.
\end{Thm}

\begin{figure}[t]
    \centering
    \includegraphics[width=0.7\linewidth]{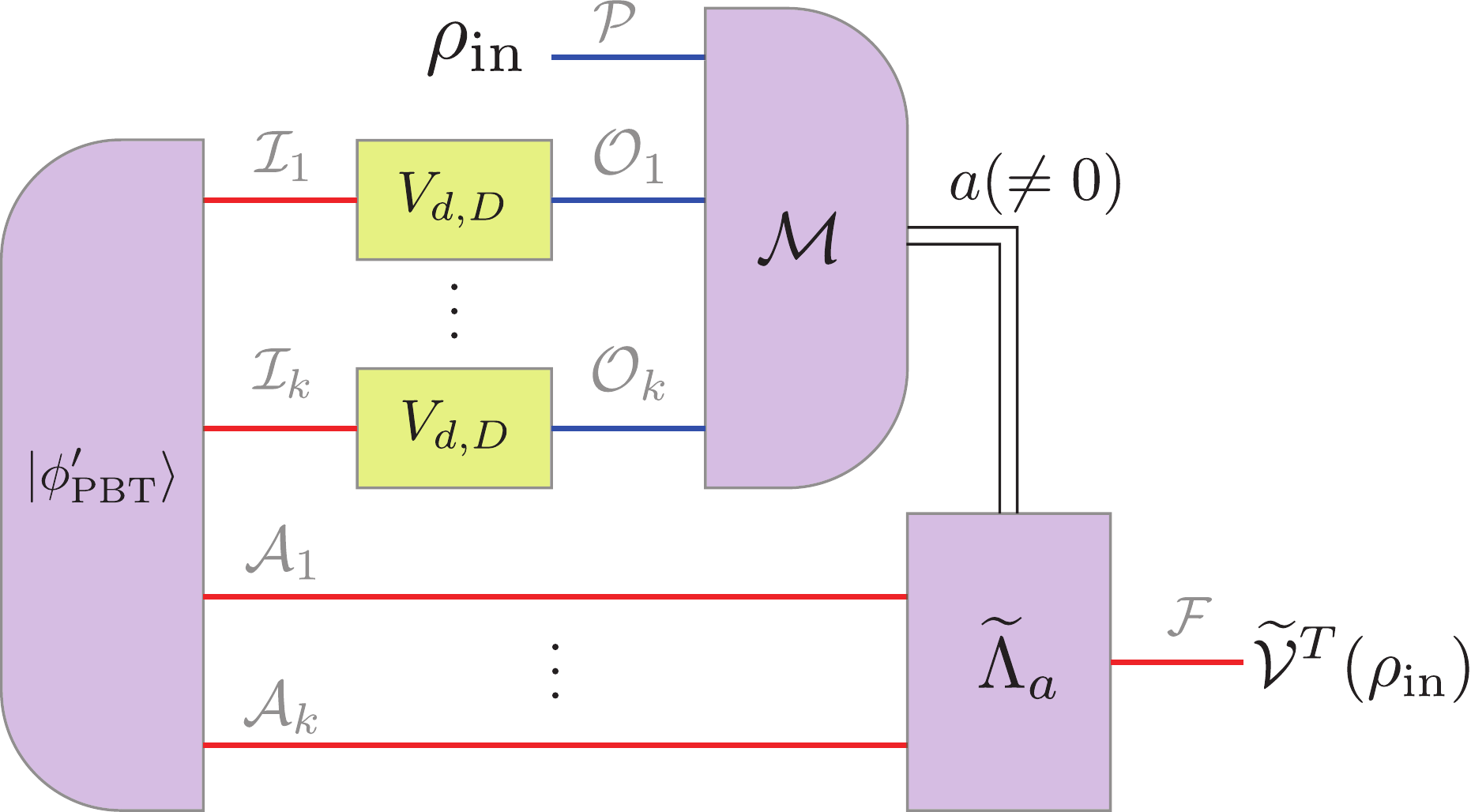}
    \caption{A quantum circuit representation of a  probabilistic parallel protocol for isometry transposition.  The quantum state $\ket{\phi_{\mathrm{PBT}}'}$ is a modified state of the optimal resource state for the probabilistic port-based teleportation \cite{Ishizaka2008Asymptotica, Studzinski2017Portbased} defined in Eq.~(\ref{eq:def_phi_PBT'}). The POVM $\mathcal{M}=\{\Gamma^{(D)}_a\}_{a=0}^{k}$ is also the optimal POVM for the probabilistic port-based teleportation, which is defined in Eq.~(\ref{eq:def_gamma_a}). The conditional CPTP map $\widetilde{\Lambda}_a$ is the operation to select the quantum state in $\mathcal{A}_a$ corresponding to the measurement outcome $a$ of $\mathcal{M}$ as the output state for $a\neq 0$, which is defined in Eq.~(\ref{eq:def_V_a}). This protocol succeeds when the measurement outcome $a$ is $a\neq 0$.}
    \label{isometry_transpose_PBT}
\end{figure}

In Figure~\ref{isometry_transpose_PBT}, the Hilbert spaces are given by $\mathcal{F}=\mathbb{C}^d$, $\mathcal{I}_i=\mathbb{C}^d$, $\mathcal{A}_i=\mathbb{C}^d$, $\mathcal{P}=\mathbb{C}^D$, and $\mathcal{O}_i=\mathbb{C}^D$ for $i\in \{1, \cdots, k\}$. We define the joint Hilbert spaces by $\mathcal{I}\coloneqq \bigotimes_{i=1}^k \mathcal{I}_i$, $\mathcal{O}\coloneqq \bigotimes_{i=1}^k \mathcal{O}_i$, $\overline{\mathcal{O}}_{a}\coloneqq \bigotimes_{i\neq a}\mathcal{O}_i$ and $\mathcal{A}\coloneqq \bigotimes_{i=1}^k \mathcal{A}_i$. The quantum state $\ket{\phi_{\mathrm{PBT}}'}$, the POVM $\mathcal{M}$, and the conditional CPTP map $\widetilde{\Lambda}_a$ are defined as follows.

To define $\ket{\phi'_{\mathrm{PBT}}}\in \mathcal{I}\otimes \mathcal{A}$, we firstly define an operator $Y_{\mathcal{I}}$ by
\begin{align}
    Y_{\mathcal{I}}\coloneqq \frac{1}{\sum_{\mu\vdash k} d_{\mathcal{U}^{(d)}_{\mu}}d_{\mathcal{U}^{(D)}_{\mu}}}\sum_{\mu\vdash k}\frac{d^k d_{\mathcal{U}^{(D)}_{\mu}}}{d_{\mathcal{S}_{\mu}^{(k)}}}\Pi_{\mu, \mathcal{I}},
\end{align}
where $d_{\mathcal{U}^{(d')}_{\mu}}$ and $d_{\mathcal{S}_{\mu}^{(k)}}$ are the dimensions of
$\mathcal{U}^{(D)}_{\mu, \mathbb{C}^{d'\otimes k}}$ and $\mathcal{S}^{(k)}_{\mu, \mathbb{C}^{D\otimes {k}}}$ for $d' \in \{d, D\}$, respectively, 
and $\Pi_{\mu, \mathcal{I}}$ is the orthogonal projector on $\mathcal{I}$ onto its subspace $\mathcal{U}^{({d'})}_{\mu, \mathcal{I}}\otimes \mathcal{S}^{(k)}_{\mu, \mathcal{I}}$. Then, we define $\ket{\phi'_{\mathrm{PBT}}}\in \mathcal{I}\otimes \mathcal{A}$ by
\begin{align}
    \ket{\phi'_{\mathrm{PBT}}}\coloneqq &(Y^{\frac{1}{2}}_{\mathcal{I}}\otimes I_{\mathcal{A}})\ket{\Phi^+_{d^k}}_{\mathcal{IA}},\label{eq:def_phi_PBT'}
\end{align}
where $\ket{\Phi^+_{d^k}}_{\mathcal{IA}}$ is the maximally entangled state given by
\begin{align}
    \ket{\Phi^+_{d^k}}_{\mathcal{IA}}\coloneqq \frac{1}{\sqrt{d^k}}\sum_{i=0}^{d^k-1}\ket{i}\otimes \ket{i}\in\mathcal{I}\otimes \mathcal{A}\label{eq:def_max_ent_state_d^k}
\end{align}
in terms of the computational basis $\{\ket{i}\}$ of $\mathcal{I}$ and $\mathcal{A}$.
The state $\ket{\phi'_{\mathrm{PBT}}}$ is a modified state of the optimal resource state $\ket{\phi^{(d)}_{\mathrm{PBT}}}$ for the probabilistic port-based teleportation of a $d$-dimensional quantum state \cite{Ishizaka2008Asymptotica, Studzinski2017Portbased} (see Appendix \ref{sec:appendix_unitary_inversion} for the definition of $\ket{\phi^{(d)}_{\mathrm{PBT}}}$).

The POVM $\mathcal{M}  = \{\Gamma^{(D)}_a\}_{a=0}^{k}$ is defined as follows. 
We define operators $\Theta_{\overline{\mathcal{O}}_{a}}$ and  $X_{\mathcal{O}}$ by
\begin{align}
    \Theta_{\overline{\mathcal{O}}_{a}} &\coloneqq \sum_{\alpha\vdash k-1}\frac{D^{k+1}g_D(k)d_{\mathcal{U}^{(D)}_{\alpha}}}{kd_{\mathcal{S}^{(k-1)}_{\alpha}}}\Pi_{\alpha, \overline{\mathcal{O}}_{a}},\label{eq:theta}\\
    X_{\mathcal{O}}&\coloneqq \sum_{\mu\vdash k}\frac{D^k g_D(k)d_{\mathcal{U}^{(D)}_{\mu}}}{d_{\mathcal{S}^{(k)}_{\mu}}}\Pi_{\mu, \mathcal{O}},\label{eq:X}
\end{align}
where $d_{\mathcal{U}^{(D)}_{\alpha}}$ and $d_{S^{(k-1)}_{\alpha}}$ are defined similarly to $d_{\mathcal{U}^{(D)}_{\mu}}$ and $d_{\mathcal{S}_{\mu}^{(k)}}$, respectively, $g_D(k)$ is defined by 
\begin{align}
    g_D(k)\coloneqq \left[\sum_{\mu\vdash k} \left(d_{\mathcal{U}^{(D)}_{\mu}}\right)^2\right]^{-1}
\end{align}
and $\Pi_{\alpha, \overline{\mathcal{O}}_{a}}$ and $\Pi_{\mu, \mathcal{O}}$ are orthogonal projectors defined similarly to $\Pi_{\mu, \mathcal{I}}$. Using $\Theta_{\overline{\mathcal{O}}_{a}}$ and $X_{\mathcal{O}}$, we define an operator  $\Gamma^{(D)}_a\;(a\in\{1, \cdots, k\})$ on $\mathcal{P}\otimes \mathcal{O}$ by
\begin{align}
    \Gamma^{(D)}_a &\coloneqq (\1_{\mathcal{P}}\otimes X_{\mathcal{O}}^{-\frac{1}{2}})\left(\ketbra{\Phi^+_D}{\Phi^+_D}_{\mathcal{P}\mathcal{O}_a}\otimes \Theta_{\overline{\mathcal{O}}_{a}}\right)(\1_{\mathcal{P}}\otimes X_{\mathcal{O}}^{-\frac{1}{2}}).\label{eq:def_gamma_a}
\end{align}
The set of operators $\{\Gamma^{(D)}_a \}_{a=1}^{k}$ satisfies $\Gamma^{(D)}_a\geq 0$ for $a\in\{1, \cdots, k\}$ and $\sum_{a=1}^{k} \Gamma^{(D)}_a\leq  \1_{\mathcal{P}\mathcal{O}}$ \cite{Studzinski2017Portbased}.  Thus, by defining $\Gamma^{(D)}_0\coloneqq  \1_{\mathcal{PO}}- \sum_{a=1}^{k} \Gamma^{(D)}_a$, a set of operators {$\{\Gamma^{(D)}_a\}_{a=0}^{k}$ forms a POVM.} The POVM $\{\Gamma^{(D)}_a\}_{a=0}^{k}$ is the optimal POVM for the probabilistic port-based teleportation of a $D$-dimensional quantum state \cite{Ishizaka2008Asymptotica, Studzinski2017Portbased}.

The conditional CPTP map $\widetilde{\Lambda}_a: \mathcal{L}(\mathcal{A})\to \mathcal{L}(\mathcal{F})$  for the outcome $a\in\{1, \cdots, k\}$ of $\mathcal{M}$ is defined by
\begin{align}
    \widetilde{\Lambda}_a(\rho_{\mathcal{A}})=\sum_{j, j'}\ketbra{j}{j'}_{\mathcal{F}} \bra{j}\Tr_{\overline{\mathcal{A}}_a}\rho_{\mathcal{A}}\ket{j'},\label{eq:def_V_a}
\end{align}
where $\{\ket{j}\}$ is the computational basis of Hilbert spaces $\mathcal{A}_a$ and $\mathcal{F}$. This conditional CPTP map represents the operation to select the quantum state in $\mathcal{A}_a$ corresponding to the measurement outcome $a$ as the output state.  The conditional CPTP map corresponding to the measurement outcome $a=0$ is not needed since our protocol only succeeds for $a\neq 0$ and the output state is aborted if $a=0$, the failure case.

 The protocol presented above is shown to implement isometry transposition with the success probability $p_{\text{succ}}=k/(Dd+k-1)$.

\begin{proof}[Proof of Theorem~\ref{theorem:isometry_transpose_PBT}]
First, we show the equality
\begin{align}
    (V^{\otimes k}_{\mathcal{I}\to \mathcal{O}}\otimes \1_{\mathcal{A}})\ket{\phi'_{\mathrm{PBT}}}_{\mathcal{I}\mathcal{A}}
    =\sqrt{\frac{\sum_{\mu\vdash k}(d_{\mathcal{U}^{(D)}_{\mu}})^2}{\sum_{\mu\vdash k}d_{\mathcal{U}^{(d)}_{\mu}} d_{\mathcal{U}^{(D)}_{\mu}}}}(\1_{\mathcal{O}}\otimes V^{T\otimes k}_{\mathcal{B}\to \mathcal{A}})\ket{\phi^{(D)}_{\mathrm{PBT}}}_{\mathcal{O}\mathcal{B}},\label{eq:h1}
\end{align}
where the joint Hilbert space $\mathcal{B}$ is defined by $\mathcal{B}\coloneqq \bigotimes_{i=1}^{k}\mathcal{B}_i$  for $\mathcal{B}_i = \mathbb{C}^D$, and the quantum state
$\ket{\phi^{(D)}_{\mathrm{PBT}}}\in \mathcal{O}\otimes \mathcal{B}$ is the optimal resource state for the probabilistic port-based teleportation of a $D$-dimensional quantum state \cite{Ishizaka2008Asymptotica, Studzinski2017Portbased} defined by
\begin{align}
    \ket{\phi^{(D)}_{\mathrm{PBT}}}\coloneqq (X_{\mathcal{O}}^{\frac{1}{2}}\otimes \1_{\mathcal{B}})\ket{\Phi^+_{D^k}}_{\mathcal{OB}}
\end{align}
using $X_{\mathcal{O}}$ given in Eq.~(\ref{eq:X}) and the maximally entangled state $\ket{\Phi^+_{D^k}}$. To show this equality, we define the maximally entangled state for each $\mu$ denoted as $\ket{\phi^+_{d, \mu}}\in (\mathcal{U}^{(d)}_{\mu, \mathcal{I}}\otimes \mathcal{S}^{(k)}_{\mu, \mathcal{I}})\otimes (\mathcal{U}^{(d)}_{\mu, \mathcal{A}}\otimes \mathcal{S}^{(k)}_{\mu, \mathcal{A}})$ given by
\begin{align}
    \ket{\phi^+_{d, \mu}}_{\mathcal{IA}}\coloneqq \frac{1}{\sqrt{d_{\mathcal{U}^{(d)}_{\mu}}d_{\mathcal{S}^{(k)}_{\mu}}}}\sum_{u_\mu, s_\mu} \ket{\mu, u_\mu, s_\mu}_{\mathcal{I}}\otimes \ket{\mu, u_\mu, s_\mu}_{\mathcal{A}},
\end{align}
where $\{\ket{\mu, u_\mu, s_\mu}\}$ is {the Schur basis.} Similarly, we define the maximally entangled state for each $\mu$ denoted as $\ket{\phi^+_{D, \mu}}\in (\mathcal{U}^{(D)}_{\mu, \mathcal{O}}\otimes \mathcal{S}^{(k)}_{\mu, \mathcal{O}})\otimes (\mathcal{U}^{(D)}_{\mu, \mathcal{B}}\otimes \mathcal{S}^{(k)}_{\mu, \mathcal{B}})$. Then, the quantum state $\ket{\phi'_{\mathrm{PBT}}}$ defined in Eq.~(\ref{eq:def_phi_PBT'}) can be written as

\begin{align}
    &\ket{\phi'_{\mathrm{PBT}}}\nonumber\\
    &= \sum_{\mu\vdash k} \sqrt{\frac{d^k d_{\mathcal{U}^{(D)}_{\mu}}}{d_{\mathcal{S}^{(k)}_{\mu}}\sum_{\nu\vdash k} d_{\mathcal{U}^{(d)}_{\nu}}d_{\mathcal{U}^{(D)}_{\nu}}}} (\Pi_{\mu, \mathcal{I}}\otimes \1_{\mathcal{A}})\ket{\Phi^+_{d^k}}_{\mathcal{IA}}\\
    &= \sum_{\mu\vdash k} \sqrt{\frac{d^k d_{\mathcal{U}^{(D)}_{\mu}}}{d_{\mathcal{S}^{(k)}_{\mu}}\sum_{\nu\vdash k} d_{\mathcal{U}^{(d)}_{\nu}}d_{\mathcal{U}^{(D)}_{\nu}}}} (\Pi_{\mu, \mathcal{I}}\otimes \1_{\mathcal{A}}) \sum_{\nu\vdash k} \sqrt{\frac{d_{\mathcal{U}_\nu^{(d)}}d_{\mathcal{S}_\nu^{(k)}}}{d^k}} (U^{\mathrm{Sch}\dagger}_{\mathcal{I}} \otimes U^{\mathrm{Sch}\dagger}_{\mathcal{A}}) \ket{\phi^+_{d, \nu}}_{\mathcal{IA}}\\
    &=\sum_{\mu \vdash k}p_{\mu} [\1_{\mathcal{I}}\otimes {U^{\mathrm{Sch}\dagger}_{\mathcal{A}}(U^{\mathrm{Sch}\dagger}_{\mathcal{A}})^t}]\ket{\phi^+_{d, \mu}}_{\mathcal{IA}},
\end{align}
where $U^{\mathrm{Sch}}$ is the quantum Schur transforms defined above Lemma~\ref{lem:Lambda}, $p_{\mu}$ is a positive value given by

\begin{align}
    p_{\mu}\coloneqq \sqrt{\frac{d_{\mathcal{U}_\mu^{(d)}}d_{\mathcal{U}_\mu^{(D)}}}{\sum_{\nu\vdash k}d_{\mathcal{U}_\nu^{(d)}}d_{\mathcal{U}_\nu^{(D)}}}},
\end{align}
and $X^{t}$ denotes the transpose of $X$ in the Schur basis. 
Similarly, $\ket{\phi^{(D)}_{\mathrm{PBT}}}$ is calculated as
\begin{align}
    \ket{\phi^{(D)}_{\mathrm{PBT}}}_{\mathcal{OB}}=
    \sum_{\mu \vdash k}q_{\mu} [\1_{\mathcal{I}}\otimes {U^{\mathrm{Sch}\dagger}_{\mathcal{B}}(U^{\mathrm{Sch}\dagger}_{\mathcal{B}})^t}]\ket{\phi^+_{D, \mu}}_{\mathcal{OB}},\label{eq:h12}
\end{align}
where $q_{\mu}$ is a positive value given by
\begin{align}
    q_{\mu}\coloneqq \frac{d_{\mathcal{U}^{(D)}_{\mu}}}{\sqrt{\sum_{\nu\vdash k} (d_{\mathcal{U}^{(D)}_{\nu}})^2}}.\label{eq:h13}
\end{align}
Since a tensor product $V^{\otimes k}$ of an isometry operator $V\in \mathbb{V}_\mathrm{iso} (d,D)$ can be decomposed in the irreducible representation form as Eq.~(\ref{eq:isometry_schur_weyl}), we obtain
\begin{align}
    (V^{\otimes k}_{\mathcal{I}\to \mathcal{O}}\otimes \1_{\mathcal{A}})\ket{\phi'_{\mathrm{PBT}}}_{\mathcal{IA}}
    &=\sum_{\mu} p_{\mu}(V^{\otimes k}_{\mathcal{I}\to \mathcal{O}}\otimes \1_{\mathcal{A}})
    [\1_{\mathcal{I}}\otimes {U^{\mathrm{Sch}\dagger}_{\mathcal{A}}(U^{\mathrm{Sch}\dagger}_{\mathcal{A}})^t}]
    \ket{\phi^+_{d, \mu}}_{\mathcal{IA}}\\
    &=\sum_{\mu} p_{\mu}
    \left[
    V_{\mu}\otimes I_{\mathcal{S}^{(k)}_{\mu, \mathcal{I}}\to \mathcal{S}^{(k)}_{\mu, \mathcal{O}}}\otimes {U^{\mathrm{Sch}\dagger}_{\mathcal{A}}(U^{\mathrm{Sch}\dagger}_{\mathcal{A}})^t}
    \right]
    \ket{\phi^+_{d, \mu}}_{\mathcal{IA}}\\
    &=\sum_{\mu}r_{\mu}
    [\1_{\mathcal{O}}\otimes {U^{\mathrm{Sch}\dagger}_{\mathcal{A}}(U^{\mathrm{Sch}\dagger}_{\mathcal{A}})^t}]
    \left(
    V_{\mu}\otimes I_{\mathcal{S}^{(k)}_{\mu, \mathcal{A}}\to \mathcal{S}^{(k)}_{\mu, \mathcal{B}}}
    \right)^{t}
    \ket{\phi^+_{D, \mu}}_{\mathcal{OB}},\label{eq:h16'}
\end{align}
where $r_\mu$ is a positive value given by
\begin{align}
    r_\mu
    \coloneqq  p_{\mu} \sqrt{\frac{d_{\mathcal{U}^{(D)}_{\mu}}}{d_{\mathcal{U}^{(d)}_{\mu}}}}
    =\frac{d_{\mathcal{U}^{(D)}_{\mu}}}{\sqrt{\sum_{\nu\vdash k} d_{\mathcal{U}^{(d)}_{\nu}}d_{\mathcal{U}^{(D)}_{\nu}}}}.\label{eq:h17}
\end{align}
For $X: \mathcal{A}\to \mathcal{B}$, the transpose $X^{t}$ in the Schur basis can be converted to the transpose $X^{T}$ in the computational basis as
\begin{align}
    X^T
    &={U^{\mathrm{Sch}\dagger}_{\mathcal{A}}(U^{\mathrm{Sch}}_{\mathcal{B}}}X_{\mathcal{A}\to \mathcal{B}}{U^{\mathrm{Sch}\dagger}_{\mathcal{A}})^{t}U^{\mathrm{Sch}}_{\mathcal{B}}}\\
    &={U^{\mathrm{Sch}\dagger}_{\mathcal{A}}( U^{\mathrm{Sch}\dagger}_{\mathcal{A}})^t}
    X^{t}_{\mathcal{B}\to \mathcal{A}}
    {(U^{\mathrm{Sch}}_{\mathcal{B}})^{t}
    U^{\mathrm{Sch}}_{\mathcal{B}}}.
\end{align}
Using this relation, we proceed the calculation in Eq.~(\ref{eq:h16'}) as
\begin{align}
    (V^{\otimes k}_{\mathcal{I}\to \mathcal{O}}\otimes \1_{\mathcal{A}})\ket{\phi'_{\mathrm{PBT}}}_{\mathcal{IA}}
    &=\sum_{\mu}r_{\mu}
    \left\{
    \1_{\mathcal{O}}\otimes  
    \left[
    \left(
    V_{\mu}\otimes I_{\mathcal{S}^{(k)}_{\mu, \mathcal{A}}\to \mathcal{S}^{(k)}_{\mu, \mathcal{B}}}
    \right)^{T}
    {U^{\mathrm{Sch} \dagger}_{\mathcal{B}}
    (U^{\mathrm{Sch} \dagger}_{\mathcal{B}})^t}
    \right]
    \right\}
    \ket{\phi^+_{D, \mu}}_{\mathcal{OB}}.\label{eq:h20}
\end{align}
From Eqs.~(\ref{eq:h12}), (\ref{eq:h13}), (\ref{eq:h17}) and (\ref{eq:h20}), we obtain Eq.~(\ref{eq:h1}).

Let $\rho'_{\mathrm{out},a}\in\mathcal{L}(\mathcal{A}_a)$ be the output state after obtaining the outcome $a$ of the POVM $\mathcal{M}=\{\Gamma^{(D)}_a\}$ given by Eq.~(\ref{eq:def_gamma_a}), but before applying the correction $\widetilde{\Lambda}_a$ given by Eq.~(\ref{eq:def_V_a}). For $a\neq 0$, the output state $\rho'_{\mathrm{out},a}$ multiplied by the probability to obtain the measurement outcome $a$, denoted by $\rho_{\mathrm{out}, a}$, is calculated as
\begin{align}
    &\rho_{\mathrm{out},a}\nonumber\\
    &=\Tr_{\mathcal{P}\mathcal{O}\overline{\mathcal{A}}_a}
    \left(
    (\Gamma_{a, \mathcal{PO}}\otimes \1_{\mathcal{A}})
    \left\{
    \rho_{\mathrm{in}, \mathcal{P}}\otimes
    \left[
    \widetilde{\mathcal{V}}^{\otimes k}_{\mathcal{I}\to \mathcal{O}}\otimes \widetilde{\1}_{\mathcal{A}}(\ketbra{\phi'_{\mathrm{PBT}}}{\phi'_{\mathrm{PBT}}}_{\mathcal{IA}})
    \right]
    \right\}
    \right)\\
    &=\frac{\sum_{\mu\vdash k}(d_{\mathcal{U}^{(D)}_{\mu}})^2}{\sum_{\mu\vdash k}d_{\mathcal{U}^{(d)}_{\mu}} d_{\mathcal{U}^{(D)}_{\mu}}}
    \Tr_{\mathcal{P}\mathcal{O}\overline{\mathcal{A}}_a}
    \left(
    (\Gamma_{a, \mathcal{PO}}\otimes \1_{\mathcal{A}})
    \left\{
    \rho_{\mathrm{in}, \mathcal{P}}\otimes
    \left[
    \widetilde{\1}_{\mathcal{O}}\otimes \widetilde{\mathcal{V}}^{T\otimes k}_{\mathcal{B}\to\mathcal{A}}(\ketbra{\phi_{\mathrm{PBT}}}{\phi_{\mathrm{PBT}}}_{\mathcal{OB}})
    \right]
    \right\}
    \right)\\
    &=\frac{\sum_{\mu\vdash k}(d_{\mathcal{U}^{(D)}_{\mu}})^2}{\sum_{\mu\vdash k}d_{\mathcal{U}^{(d)}_{\mu}} d_{\mathcal{U}^{(D)}_{\mu}}}
    \Tr_{\mathcal{P}\mathcal{O}\overline{\mathcal{A}}_a}
    \left(
    (\Gamma_{a, \mathcal{PO}}\otimes \1_{\mathcal{A}})
    \left\{
    \rho_{\mathrm{in}, \mathcal{P}}\otimes
    \left[
    \widetilde{\mathcal{X}}^{\frac{1}{2}}_{\mathcal{O}}\otimes \widetilde{\mathcal{V}}^{T\otimes k}_{\mathcal{B}\to\mathcal{A}}(\ketbra{\Phi^+_{D^{k}}}{\Phi^+_{D^{k}}}_{\mathcal{OB}})
    \right]
    \right\}
    \right)\\
    &=\frac{\sum_{\mu\vdash k}(d_{\mathcal{U}^{(D)}_{\mu}})^2}{\sum_{\mu\vdash k}d_{\mathcal{U}^{(d)}_{\mu}} d_{\mathcal{U}^{(D)}_{\mu}}}\nonumber\\
    &\hspace{12pt}\times
    \Tr_{\mathcal{P}\mathcal{O}\overline{\mathcal{A}}_a}
    \left(
    \left[
    \widetilde{\1}_{\mathcal{P}}\otimes \widetilde{\mathcal{X}}^{\dagger\frac{1}{2}}_{\mathcal{O}}(\Gamma_{a, \mathcal{PO}})\otimes \1_{\mathcal{A}}
    \right]
    \left\{
    \rho_{\mathrm{in}, \mathcal{P}}\otimes
    \left[
    \widetilde{\1}_{\mathcal{O}}\otimes \widetilde{\mathcal{V}}^{T\otimes k}_{\mathcal{B}\to\mathcal{A}}(\ketbra{\Phi^+_{D^{k}}}{\Phi^+_{D^{k}}}_{\mathcal{OB}})
    \right]
    \right\}
    \right).
\end{align}
Since $X_{\mathcal{O}}=X_{\mathcal{O}}^{\dagger}$ holds and the definition of $\Gamma_a$ is given by Eq.~(\ref{eq:def_gamma_a}), we obtain
\begin{align}
    \left(
    \widetilde{\1}_{\mathcal{P}}\otimes \widetilde{\mathcal{X}}^{\dagger\frac{1}{2}}_{\mathcal{O}}
    \right)
    (\Gamma_{a, \mathcal{PO}})=\ketbra{\Phi^+_D}{\Phi^+_D}_{\mathcal{P}\mathcal{O}_a}\otimes \Theta_{\overline{\mathcal{O}}_a}.
\end{align}
Therefore, $\rho_{\mathrm{out}, a}$ is further calculated as
\begin{align}
    &\rho_{\mathrm{out},a}\nonumber\\
    &=\frac{\sum_{\mu\vdash k}(d_{\mathcal{U}^{(D)}_{\mu}})^2}{\sum_{\mu\vdash k}d_{\mathcal{U}^{(d)}_{\mu}} d_{\mathcal{U}^{(D)}_{\mu}}}\nonumber\\
    &\hspace{12pt}\times
    \Tr_{\mathcal{P}\mathcal{O}\overline{\mathcal{A}}_a}
    \left(
    (\ketbra{\Phi^+_D}{\Phi^+_D}_{\mathcal{P}\mathcal{O}_a}\otimes \Theta_{\overline{\mathcal{O}}_a}\otimes \1_{\mathcal{A}})
    \left\{
    \rho_{\mathrm{in}, \mathcal{P}}\otimes
    \left[
    \widetilde{\1}_{\mathcal{O}}\otimes \widetilde{\mathcal{V}}^{T\otimes k}_{\mathcal{B}\to\mathcal{A}}(\ketbra{\Phi^{+}_{D^{k}}}{\Phi^{+}_{D^{k}}}_{\mathcal{OB}})
    \right]
    \right\}
    \right)\\
    &=\frac{\sum_{\mu\vdash k}(d_{\mathcal{U}^{(D)}_{\mu}})^2}{\sum_{\mu\vdash k}d_{\mathcal{U}^{(d)}_{\mu}} d_{\mathcal{U}^{(D)}_{\mu}}}
    \Tr
    \left\{
    (\Theta_{\overline{\mathcal{O}}_a}\otimes \1_{\overline{\mathcal{A}}_a})
    \left[
    \widetilde{\1}_{\overline{\mathcal{O}}_a}\otimes \widetilde{\mathcal{V}}^{T\otimes k-1}_{\overline{\mathcal{B}}_a\to\overline{\mathcal{A}}_a} (\ketbra{\Phi^+_{D^{k-1}}}{\Phi^+_{D^{k-1}}}_{\overline{\mathcal{O}}_a\overline{\mathcal{B}}_a})
    \right]
    \right\}\nonumber\\
    &\hspace{12pt}\times \Tr_{\mathcal{P}\mathcal{O}_a}
    \left\{
    (\ketbra{\Phi^+_D}{\Phi^+_D}_{\mathcal{P}\mathcal{O}_a}\otimes \1_{\mathcal{B}_a\to\mathcal{A}_a})
    \left[
    \widetilde{\1}_{\mathcal{P}\mathcal{O}_a}\otimes \widetilde{\mathcal{V}}^{T}_{\mathcal{B}_a}
    (\rho_{\mathrm{in}, \mathcal{P}}\otimes \ketbra{\Phi^+_D}{\Phi^+_D}_{\mathcal{O}_a\mathcal{B}_a})
    \right]
    \right\}\\
    &=\frac{\sum_{\mu\vdash k}(d_{\mathcal{U}^{(D)}_{\mu}})^2}{\sum_{\mu\vdash k}d_{\mathcal{U}^{(d)}_{\mu}} d_{\mathcal{U}^{(D)}_{\mu}}}
    \frac{d^{k-1}}{D^{k+1}}
    \Tr
    \left\{
    (\Theta_{\overline{\mathcal{O}}_a}\otimes \1_{\overline{\mathcal{A}}_a})
    \left[
    \widetilde{\mathcal{V}}^{\otimes k-1}_{\overline{\mathcal{I}}_a\to\overline{\mathcal{O}}_a}\otimes \widetilde{\1}_{\overline{\mathcal{A}}_a} (\ketbra{\Phi^+_{d^{k-1}}}{\Phi^+_{d^{k-1}}}_{\overline{\mathcal{I}}_a\overline{\mathcal{A}}_a})
    \right]
    \right\}\nonumber\\
    &\hspace{12pt}\times
    \widetilde{\mathcal{V}}^T_{\mathcal{B}_a}\circ \widetilde{\1}_{\mathcal{P}\to \mathcal{B}_a}
    (\rho_{\mathrm{in}, \mathcal{P}}).
\end{align}
The final successful output state $\rho'_{\mathrm{out}}$ is obtained by applying the conditional CPTP map $\widetilde{\Lambda}_a$ on $\rho'_{\mathrm{out},a}$ for $a\in\{1, \cdots, k\}$, which corrects the index of the Hilbert space from $\mathcal{A}_a$ to $\mathcal{F}$. Let $\rho_{\mathrm{out}}$ be the final successful output state $\rho'_{\mathrm{out}}$ multiplied by the probability that the measurement outcome $a$ satisfies $a\neq 0$. Namely, from the definition of the successful
output state, $\rho_{\mathrm{out}}$ is written by
\begin{align}
    \rho_{\mathrm{out}}=\sum_{a=1}^{k}\rho_{\mathrm{out},a}=p_{\mathrm{succ}}\widetilde{\mathcal{V}}^T(\rho_{\mathrm{in}}),
\end{align}
where the success probability $p_{\mathrm{succ}}$ is given by
\begin{align}
    p_{\mathrm{succ}}
    &=\sum_{a=1}^{k}\frac{\sum_{\mu\vdash k}(d_{\mathcal{U}^{(D)}_{\mu}})^2}{\sum_{\mu\vdash k}d_{\mathcal{U}^{(d)}_{\mu}} d_{\mathcal{U}^{(D)}_{\mu}}}
    \frac{d^{k-1}}{D^{k+1}}
    \Tr
    \left\{
    (\Theta_{\overline{\mathcal{O}}_a}\otimes \1_{\overline{\mathcal{A}}_a})
    \left[
    \widetilde{\mathcal{V}}^{\otimes k-1}_{\overline{\mathcal{I}}_a\to\overline{\mathcal{O}}_a}\otimes \widetilde{\1}_{\overline{\mathcal{A}}_a} (\ketbra{\Phi^+_{d^{k-1}}}{\Phi^+_{d^{k-1}}}_{\overline{\mathcal{I}}_a\overline{\mathcal{A}}_a})
    \right]
    \right\}.
\end{align}
The success probability $p_{\mathrm{succ}}$ is calculated as
\begin{align}
    p_{\mathrm{succ}}
    &=\sum_{a=1}^{k}\frac{\sum_{\mu\vdash k}(d_{\mathcal{U}^{(D)}_{\mu}})^2}{\sum_{\mu\vdash k}d_{\mathcal{U}^{(d)}_{\mu}} d_{\mathcal{U}^{(D)}_{\mu}}}
    \frac{1}{D^{k+1}}
    \Tr(\Theta_{\overline{\mathcal{O}}_a} V^{\otimes k-1}V^{\dagger \otimes k-1})\\
    &=\sum_{a=1}^{k}\frac{\sum_{\mu\vdash k}(d_{\mathcal{U}^{(D)}_{\mu}})^2}{\sum_{\mu\vdash k}d_{\mathcal{U}^{(d)}_{\mu}} d_{\mathcal{U}^{(D)}_{\mu}}}
    \frac{1}{D^{k+1}}
    \Tr(\Theta_{\overline{\mathcal{O}}_a} \Pi_{\Im V^{\otimes k-1}})\\
    &=\sum_{a=1}^{k}\frac{\sum_{\mu\vdash k}(d_{\mathcal{U}^{(D)}_{\mu}})^2}{\sum_{\mu\vdash k}d_{\mathcal{U}^{(d)}_{\mu}} d_{\mathcal{U}^{(D)}_{\mu}}}
    \sum_{\alpha\vdash k-1}\frac{g_D(k)d_{\mathcal{U}^{(D)}_{\alpha}}}{kd_{\mathcal{S}^{(k-1)}_{\alpha}}}
    \Tr(\Pi_{\alpha, \overline{\mathcal{O}}_a} \Pi_{\Im V^{\otimes k-1}})\\
    &=\sum_{a=1}^{k}\frac{\sum_{\mu\vdash k}(d_{\mathcal{U}^{(D)}_{\mu}})^2}{\sum_{\mu\vdash k}d_{\mathcal{U}^{(d)}_{\mu}} d_{\mathcal{U}^{(D)}_{\mu}}}
    \sum_{\alpha\vdash k-1}\frac{g_D(k)d_{\mathcal{U}^{(D)}_{\alpha}}d_{\mathcal{U}^{(d)}_{\alpha}}}{k}\label{eq:h16}\\
    &=\frac{\sum_{\alpha\vdash k-1}d_{\mathcal{U}^{(d)}_{\alpha}} d_{\mathcal{U}^{(D)}_{\alpha}}}{\sum_{\mu\vdash k}d_{\mathcal{U}^{(d)}_{\mu}} d_{\mathcal{U}^{(D)}_{\mu}}},
\end{align}
where Eq.~(\ref{eq:h16}) follows from the relation $\Tr(\Pi_{\alpha, \overline{\mathcal{O}}_a} \Pi_{\Im V^{\otimes k-1}})=d_{\mathcal{U}^{(d)}_\alpha}d_{\mathcal{S}^{(k-1)}_{\alpha}}$, which is obtained by the isomorphism $\Im V^{\otimes k-1} = \bigoplus_{\alpha\vdash k-1}\mathcal{U}^{(d)}_{\alpha}\otimes \mathcal{S}^{(k-1)}_{\alpha}$.

We further calculate the success probability $p_{\mathrm{succ}}$ using techniques similar to those presented in Ref.~\cite{Studzinski2017Portbased}.
Let $\chi_{d',  k}$ be the character of the representation $P_{\sigma}$ of $\mathfrak{S}_{k}$ on $(\mathbb{C}^{d'})^{\otimes  k}$ and $\chi_{\mu}$ be the irreducible character of the representation $ \mu$ of $\mathfrak{S}_{ k}$.
Then, we have
\begin{align}
    \chi_{d', k}=\sum_{\mu\vdash  k}d_{\mathcal{U}^{(d')}_ \mu}\chi_\mu.
\end{align}
We define an inner product of two characters $\chi, \chi'$ by 
\begin{align}
    \langle \chi, \chi'\rangle\coloneqq \frac{1}{|\mathfrak{S}_{k}|}\sum_{\sigma\in \mathfrak{S}_{k}}\chi(\sigma) \chi'^*(\sigma).
\end{align}
Then, $\langle \chi_\mu, \chi_\nu\rangle=\delta_{ \mu,  \nu}$ for irreducible representations $\mu, \nu\vdash  k$. Therefore,
\begin{align}
    \langle \chi_{d, k}, \chi_{D, k} \rangle=\sum_{\mu\vdash  k}d_{\mathcal{U}^{(d)}_\mu}d_{\mathcal{U}^{(D)}_\mu}.
\end{align}
holds.
By definition of the inner product, we obtain
\begin{align}
    \sum_{\mu\vdash  k}d_{\mathcal{U}^{(d)}_{\mu'}}d_{\mathcal{U}^{(D)}_\mu}=\frac{1}{|\mathfrak{S}_{k}|}\sum_{\sigma\in \mathfrak{S}_{k}}\chi_{d, k}(\sigma)\chi^*_{D, k}(\sigma).
\end{align}
Let $l(\sigma)$ be the minimum number $n$ such that $\sigma$ is written as a product of $n$ permutations $\sigma=\tau_1\cdots \tau_n$. 
Then, we have $\chi_{d', k}(\sigma)=d'^{l(\sigma)}$. Therefore, we obtain
\begin{align}
    \sum_{ \mu\vdash  k}d_{\mathcal{U}^{( d)}_{\mu'}}d_{\mathcal{U}^{( D)}_\mu}=\frac{1}{ k!}\sum_{\sigma\in\mathfrak{S}_{k}}(dD)^{l(\sigma)}.
\end{align}
A permutation $\sigma\in \mathfrak{S}_{k}$ can be written uniquely as $\sigma=(a  k)\tau$ using $a\in \{1, \cdots,  k\}$ and $\tau\in\mathfrak{S}_{ k-1})$, and $l(\sigma)$ can be calculated inductively by the relation given by
\begin{align}
    l(\sigma)=
    \begin{cases}
    l(\tau)+1 & (a=  k)\\
    l(\tau) & (a\neq  k)
    \end{cases}.
\end{align}
Therefore, we finally obtain
\begin{align}
    \sum_{ \mu\vdash  k}d_{\mathcal{U}^{(d)}_\mu}d_{\mathcal{U}^{(D)}_\mu}
    &=\frac{1}{ k!}\sum_{\sigma\in\mathfrak{S}_{ k}}(dD)^{l(\sigma)}\\
    &=\frac{dD+ k-1}{ k!}\sum_{\tau\in \mathfrak{S}_{k-1}}(dD)^{l(\tau)}\\
    &=\frac{dD+ k-1}{ k}\sum_{\alpha\vdash  k-1}d_{\mathcal{U}^{(d)}_\alpha} d_{\mathcal{U}^{(D)}_\alpha},
\end{align}
which leads to
\begin{align}
    p_{\mathrm{succ}}
    &=\frac{\sum_{\alpha\vdash k-1}d_{\mathcal{U}^{(d)}_{\alpha}} d_{\mathcal{U}^{(D)}_{\alpha}}}{\sum_{\mu\vdash k}d_{\mathcal{U}^{(d)}_{\mu}} d_{\mathcal{U}^{(D)}_{\mu}}}\\
    &=\frac{k}{Dd+k-1}.
\end{align}
\end{proof}

 \subsection{Optimality of the  parallel isometry transposition for $k=1$}

We show that this isometry transposition protocol is optimal for $k=1$ from the uniqueness of isometry transposition similarly to  the case of unitary transposition where the uniqueness of transposition implies the optimality of the parallel protocol \cite{Dong2021Quantum}.

\begin{Thm}
\label{thm:isometry_transposition_optimality}
The optimal success probability of probabilistic protocols that transform a single call of an isometry operation $\widetilde{\mathcal{V}}: \mathcal{L}(\mathbb{C}^d)\to \mathcal{L}(\mathbb{C}^D)$ into its transposed map $\widetilde{\mathcal{V}}^T$ is $p_{\mathrm{opt}}=1/(Dd)$.
\end{Thm}

First, we show the following Lemma~on the uniqueness of isometry transposition.
\begin{Lem}\label{lem:uniqueness_transpose}
If a single-input superinstrument $\{\doublewidetilde{\mathcal{S}^{}}, \doublewidetilde{\mathcal{F}^{}}\}$ implements a probabilistic exact isometry transposition protocol, i.e.,
\begin{align}
\doublewidetilde{\mathcal{S}^{}}(\widetilde{\mathcal{V}})=p_{\mathrm{succ}}\widetilde{\mathcal{V}}^T\;\;\;(\forall V\in \mathbb{V}_\mathrm{iso} (d,D)),\label{eq:single-input_isometry_transposition}
\end{align}
the  Choi operator $S$ of $\doublewidetilde{\mathcal{S}^{}}$ is uniquely given by
\begin{align}
    S=p_{\mathrm{succ}}Dd\ketbra{\Phi^+_D}{\Phi^+_D}_{\mathcal{P}\mathcal{O}_1}\otimes \ketbra{\Phi^+_d}{\Phi^+_d}_{\mathcal{I}_1\mathcal{F}}.
    \label{eq:isometry_transpose_uniqueness}
\end{align}
Moreover, $\doublewidetilde{\mathcal{S}^{}}(\widetilde{\Lambda})=p_{\mathrm{succ}}\widetilde{\Lambda}^T$ holds for any map $\widetilde{\Lambda}: \mathcal{L}(\mathbb{C}^d)\to \mathcal{L}(\mathbb{C}^D)$.
\end{Lem}
\begin{proof}
We choose a set $\sigma=\{\sigma_1, \cdots, \sigma_d\}\subset \{1, \cdots, D\}$ such that $\sigma_1<\cdots<\sigma_d$.
Let the dimensions of the Hilbert spaces as $\mathcal{P}', \mathcal{O}'_1=\mathbb{C}^d$. 
We also define $W_{\sigma}\coloneqq \sum_{i}\ket{i}\bra{\sigma_i}$ and $\Pi_{\sigma}\coloneqq W_{\sigma}^{\dagger}W_{\sigma}$.
We define operators $S'\in \mathcal{L}(\mathcal{P}'\otimes \mathcal{I}_1\otimes \mathcal{O}'_{1}\otimes \mathcal{F})$ and $S''\in \mathcal{L}(\mathcal{P}\otimes \mathcal{I}_1\otimes \mathcal{O}'_{1}\otimes \mathcal{F})$ by
\begin{align}
    S'&\coloneqq 
    \left(
    \widetilde{W}_{\sigma, \mathcal{P}\to \mathcal{P}'}\otimes \widetilde{W}_{\sigma, \mathcal{O}_1\to \mathcal{O}'_1}\otimes \widetilde{\1}_{\mathcal{I}_1\mathcal{F}}
    \right)
    (S),\\
    S''&\coloneqq  
    \left(
    \widetilde{W}_{\sigma, \mathcal{O}_1\to \mathcal{O}'_1}\otimes \widetilde{\1}_{\mathcal{P}\mathcal{I}_1\mathcal{F}}
    \right)(S).
\end{align}
From Eq.~(\ref{eq:single-input_isometry_transposition}), we obtain
\begin{align}
    S\star \dketbra{V}=p_{\mathrm{succ}}\dketbra{V^T}\;\;\;(\forall V\in \mathbb{V}_\mathrm{iso} (d,D)).
\end{align}
Then, for $U\in \mathbb{U}(D)$, we have
\begin{align}
    S''\star \dketbra{U}_{\mathcal{I}_1\mathcal{O}'_1}
    &=S \star 
    \left[
    \left(
    \widetilde{\1}_{\mathcal{I}_1}\otimes \widetilde{\mathcal{W}}^{T}_{\sigma, \mathcal{O}'_1\to\mathcal{O}_1}
    \right)
    (\dketbra{U}_{\mathcal{I}_1\mathcal{O}'_1})
    \right]\\
    &=S \star \dketbra{V}_{\mathcal{I}_1\mathcal{O}_1}\\
    &=p_{\mathrm{succ}}\dketbra{V^T}_{\mathcal{P}\mathcal{F}},
\end{align}
where $V\coloneqq W^T_{\sigma}U$ is an isometry operator. Then, we obtain
\begin{align}
    S'\star \dketbra{U}_{\mathcal{I}_1\mathcal{O}'_1}=p_{\mathrm{succ}}\dketbra{U^T}_{\mathcal{P}'\mathcal{F}},\label{eq:a6}\\
    [S''-\widetilde{\mathcal{W}}^{\dagger}_{\sigma, \mathcal{P}'\to\mathcal{P}}\otimes \widetilde{\1}_{\mathcal{I}_1\mathcal{O}'_1\mathcal{F}}(S')]\star \dketbra{U}_{\mathcal{I}_1\mathcal{O}'_1}=0.\label{eq:a7}
\end{align}

From Eq.~(\ref{eq:a6}) and the uniqueness of unitary transposition \cite{Dong2021Quantum}, we obtain
\begin{align}
    S'=pd^2\ketbra{\Phi^+_d}{\Phi^+_d}_{\mathcal{P}'\mathcal{O}'_1}\otimes \ketbra{\Phi^+_d}{\Phi^+_d}_{\mathcal{I}_1\mathcal{F}}.
\end{align}
From Eq.~(\ref{eq:a7}) and the fact that $\mathrm{span}\{\dketbra{U}_{\mathcal{I}_1\mathcal{O}'_1}\}=\mathcal{I}_1\otimes \mathcal{O}'_1$ holds, we obtain
\begin{align}
    S''=
    \left(
    \widetilde{\mathcal{W}}^{\dagger}_{\sigma, \mathcal{P}'\to\mathcal{P}}\otimes \widetilde{\1}_{\mathcal{I}_1\mathcal{O}'_1\mathcal{F}}
    \right)
    (S').
\end{align}
We define $S_{\sigma}\coloneqq \widetilde{\Pi}_{\mathcal{O}_1, \sigma}\otimes \widetilde{\1}_{\mathcal{P}\mathcal{I}_1\mathcal{F}}(S)$. Then, $S_{\sigma}$ is calculated as
\begin{align}
    S_{\sigma}
    &=
    \left(
    \widetilde{\mathcal{W}}_{\sigma, \mathcal{O}'_1\to \mathcal{O}_1}^{\dagger}\otimes \widetilde{\1}_{\mathcal{P}\mathcal{I}_1\mathcal{F}}
    \right)
    (S'')\\
    &=
    \left(
    \widetilde{\mathcal{W}}_{\sigma, \mathcal{P}'\to\mathcal{P}}^{\dagger}\otimes \widetilde{\mathcal{W}}_{\sigma, \mathcal{O}'_1\to \mathcal{O}_1}^{\dagger}\otimes \widetilde{\1}_{\mathcal{I}_1\mathcal{F}}
    \right)
    (S'')\\
    &=p_{\mathrm{succ}}d^2\ketbra{\Phi'^+_{\sigma}}{\Phi'^+_{\sigma}}_{\mathcal{P}\mathcal{O}_1}\otimes \ketbra{\Phi^+_d}{\Phi^+_d}_{\mathcal{I}_1\mathcal{F}},
\end{align}
where $\ket{\Phi'^+_{\sigma}}_{\mathcal{P}\mathcal{O}_1}$ is defined as 
\begin{align}
    \ket{\Phi'^+_{\sigma}}_{\mathcal{P}\mathcal{O}_1}\coloneqq \frac{1}{\sqrt{d}}\sum_{i=1}^{d}\ket{\sigma_{i}\sigma_{i}}_{\mathcal{P}\mathcal{O}_1}.
\end{align}
In other words, for $i_m\in \{1, \cdots, D\}$ and $j_m, k_m, l_m\in \{1, \cdots, d\}$ for $m\in \{1, 2\}$, the matrix elements of $S$ are given by 
\begin{align}
    \bra{i_1 j_1\sigma_{k_1}l_1}S\ket{i_2 j_2\sigma_{k_2}l_2}
    &=\bra{i_1 j_1\sigma_{k_1}l_1}S_{\sigma}\ket{i_2j_2\sigma_{k_2}l_2}\\
    &=p_{\mathrm{succ}}\delta_{i_1,\sigma_{k_1}}\delta_{j_1,l_1}\delta_{i_2,\sigma_{k_2}}\delta_{j_2,l_2}.
\end{align}
Since this holds for any $\sigma$, we obtain
\begin{align}
    \bra{i_1j_1k_1l_1}S\ket{i_2j_2k_2l_2}
    =p_{\mathrm{succ}}\delta_{i_1,k_1}\delta_{j_1,l_1}\delta_{i_2,k_2}\delta_{j_2,l_2}.
\end{align}
for $i_m, k_m\in \{1, \cdots, D\}$ and $j_m, l_m\in\{1, \cdots, d\}$ for $m\in \{1, 2\}$. Thus, $S$ is uniquely determined as
\begin{align}
    S=p_{\mathrm{succ}}Dd\ketbra{\Phi^+_D}{\Phi^+_D}_{\mathcal{P}\mathcal{O}_1}\otimes \ketbra{\Phi^+_d}{\Phi^+_d}_{\mathcal{I}_1\mathcal{F}}.
\end{align}
Moreover, for $K\in\mathcal{L}(\mathcal{I}_1\to\mathcal{O}_1)$, we obtain
\begin{align}
    S\star \dketbra{K}_{\mathcal{I}_1\mathcal{O}_1}=p_{\mathrm{succ}}\dketbra{K^T}_{\mathcal{P}\mathcal{F}}.
\end{align}
Thus, we show $\doublewidetilde{\mathcal{S}^{}}(\widetilde{\Lambda})=p\widetilde{\Lambda}^T$ for any map $\widetilde{\Lambda}: \mathcal{L}(\mathbb{C}^d)\to \mathcal{L}(\mathbb{C}^D)$.
\end{proof}

\begin{proof}[Proof of Theorem~\ref{thm:isometry_transposition_optimality}]
Let $\{\doublewidetilde{\mathcal{S}^{}}, \doublewidetilde{\mathcal{F}^{}}\}$ be a single-input superinstrument which implements a probabilistic exact isometry transposition protocol with a success probability $p_{\mathrm{succ}}$.
Let the  Choi operator of $\doublewidetilde{\mathcal{S}^{}}$ and $\doublewidetilde{\mathcal{F}^{}}$ be $S, F\in \mathcal{L}(\mathcal{P}\otimes \mathcal{I}_1\otimes \mathcal{O}_1\otimes \mathcal{F})$, respectively, and we set an operator $C\coloneqq S+F$. From Lemma~\ref{lem:uniqueness_transpose}, the Choi operator $S$ is obtained by Eq.~(\ref{eq:isometry_transpose_uniqueness}).
Since $\{\doublewidetilde{\mathcal{S}^{}}, \doublewidetilde{\mathcal{F}^{}}\}$ is a single-input superinstrument, the conditions for $C$
\begin{align}
    C&\geq S,\\
    \Tr_{\mathcal{F}}C&=\Tr_{\mathcal{O}_1\mathcal{F}}C\otimes \frac{\1_{\mathcal{O}_1}}{d_{\mathcal{O}_1}},\\
    \Tr_{\mathcal{I}_1\mathcal{O}_1 \mathcal{F}}C&=d_{\mathcal{O}_1}\1_{\mathcal{P}}.
\end{align}
have to be satisfied.
Then, we obtain
\begin{align}
    \Tr_{\mathcal{I}_1\mathcal{F}}C=\Tr_{\mathcal{I}_1\mathcal{O}_1\mathcal{F}}C\otimes \frac{\1_{\mathcal{O}_1}}{d_{\mathcal{O}_1}}=\1_{\mathcal{P}}\otimes \1_{\mathcal{O}_1}.
\end{align}
Since $\Tr_{\mathcal{I}_1\mathcal{F}}C\geq \Tr_{\mathcal{I}_1\mathcal{F}}S=p_{\mathrm{succ}}Dd\ketbra{\Phi^+_D}{\Phi^+_D}_{\mathcal{P}\mathcal{O}_1}$ holds, we have
\begin{align}
    p_{\mathrm{succ}}Dd\ketbra{\Phi^+_D}{\Phi^+_D}_{\mathcal{P}\mathcal{O}_1}\leq \1_{\mathcal{P}}\otimes \1_{\mathcal{O}_1}.
\end{align}
Thus, we obtain $p_{\mathrm{succ}}\leq 1/(Dd)$, namely $p_{\mathrm{opt}}=1/(Dd)$.
\end{proof}

 \section{SDP to obtain the maximal success probability of isometry inversion, (pseudo) complex conjugation, and transposition}
\label{appendix:sdp}

Using the Choi representation of a supermap introduced in Appendix \ref{sec:supermap_superinstrument_higher-order}, a superinstrument $\{\doublewidetilde{\mathcal{S}^{}}, \doublewidetilde{\mathcal{F}^{}}\}: \bigotimes_{i=1}^{k}[\mathcal{L}(\mathcal{I}_i) \to \mathcal{L}(\mathcal{O}_i)] \to [\mathcal{L}(\mathcal{P})\to \mathcal{L}(\mathcal{F})]$ can be represented by $S, F \in \mathcal{L}(\mathcal{P}\otimes \mathcal{I}\otimes \mathcal{O}\otimes \mathcal{F})$, where $\mathcal{I}\coloneqq \bigotimes_{i=1}^{k}\mathcal{I}_i$ and $\mathcal{O}\coloneqq \bigotimes_{i=1}^{k}\mathcal{O}_i$. The Choi operators $\{S, F\}$ of a $k$-input parallel/sequential/general superinstrument is characterized by
\begin{align}
    &S, F\geq 0,\\
    &C\coloneqq S+F \text{ is a parallel/sequential/general superchannel},\label{cond:superchannel}
\end{align}
where the second condition (\ref{cond:superchannel}) is composed of the positivity of $C$ and linear constraints on $C$ (see Section II of Ref.~\cite{Quintino2019Probabilistic}).

This Appendix shows how to formulate the problem of finding the optimal success probability of isometry inversion, isometry (pseudo) complex conjugation, isometry transposition, and ``success-or-draw'' isometry inversion as SDP using the Choi representation.

\subsection{Isometry inversion}
The superinstrument $\{\doublewidetilde{\mathcal{S}^{}}, \doublewidetilde{\mathcal{F}^{}}\}$ implements isometry inversion with the success probability $p$ if and only if
\begin{align}
    \doublewidetilde{\mathcal{S}^{}}(\widetilde{\mathcal{V}}^{\otimes k})\circ \widetilde{\mathcal{V}} = p \widetilde{\1}_d
\end{align}
holds for all $V\in \mathbb{V}_{\text{iso}}(d,D)$. This condition can be written in the Choi representation as
\begin{align}
    S\star \dketbra{V}^{\otimes k}_{\mathcal{I}\mathcal{O}} \star \dketbra{V}_{\mathcal{P}'\mathcal{P}} = p \dketbra{\1_d}_{\mathcal{P}'\mathcal{F}},
\end{align}
where $\mathcal{P}'=\mathbb{C}^d$ and $\dket{V}\coloneqq \sum_i \ket{i}\otimes V\ket{i}$ for the computational basis $\{\ket{i}\}$ of $\mathbb{C}^d$.
 Since $\Tr[\dketbra{V}^{\otimes k+1}]=d^{k+1}$ is constant for all $V\in\mathbb{V}_{\text{iso}}(d,D)$, this condition is equivalent to the following condition:
\begin{align}
    S\star \dketbra{V}^{\otimes k}_{\mathcal{I}\mathcal{O}} \star \dketbra{V}_{\mathcal{P}'\mathcal{P}} = p \dketbra{\1_d}_{\mathcal{P}'\mathcal{F}}\times \frac{\Tr[\dketbra{V}^{\otimes k+1}]}{d^{k+1}},\label{eq:cond_isometry_inversion_choi}
\end{align}
Due to the linearity of Eq.~(\ref{eq:cond_isometry_inversion_choi}) with respect to $\dketbra{V}^{\otimes k+1}$, Eq.~(\ref{eq:cond_isometry_inversion_choi}) holds for all $V\in\mathbb{V}_{\text{iso}}(d,D)$ if and only if it holds for finite set of isometries $\{V_i\}_i \subset \mathbb{V}_{\text{iso}}(d,D)$ forming the basis $\{\dketbra{V_i}^{\otimes k+1}\}_i$ of the linear space $\mathcal{H}_{k+1}\coloneqq \text{span}\{\dketbra{V}^{\otimes k+1}|V\in\mathbb{V}_{\text{iso}}(d,D)\}$.

Therefore, the optimal success probability of isometry inversion in parallel/sequential/general protocols can be found as a solution of the following SDP:
\begin{align}
    &\max p\\
    \text{s.t. } & S, F \in \mathcal{L}(\mathcal{P}\otimes \mathcal{I}\otimes \mathcal{O} \otimes \mathcal{F}),\\
    & S, F\geq 0,\\
    & S\star \dketbra{V_i}^{\otimes k}_{\mathcal{I}\mathcal{O}} \star \dketbra{V_i}_{\mathcal{P}'\mathcal{P}} = p \dketbra{\1_d}_{\mathcal{P}'\mathcal{F}}\;\;\;\forall i,\\
    & C\coloneqq S+F \text{ is a parallel/sequential/general superchannel},
\end{align}
where $\{V_i\}_i$ is a finite set of isometries forming the basis $\{\dketbra{V_i}^{\otimes k+1}\}_i$ of the linear space $\mathcal{H}_{k+1}\coloneqq \text{span}\{\dketbra{V}^{\otimes k+1}|V\in\mathbb{V}_{\text{iso}}(d,D)\}$.

\subsection{Isometry pseudo complex conjugation}
The superinstrument $\{\doublewidetilde{\mathcal{S}^{}}, \doublewidetilde{\mathcal{F}^{}}\}$ implements isometry inversion with the success probability $p$ if and only if
\begin{align}
    [\doublewidetilde{\mathcal{S}^{}}(\widetilde{\mathcal{V}}^{\otimes k})]^T\circ \widetilde{\mathcal{V}} = p \widetilde{\1}_d\label{eq:cond_isometry_pss}
\end{align}
holds for all $V\in \mathbb{V}_{\text{iso}}(d,D)$. For any linear map $\widetilde{\Lambda}: \mathcal{L}(\mathcal{I}) \to \mathcal{L}(\mathcal{O})$, the Choi operator of the transposed map $\widetilde{\Lambda}^T$ is expressed as the Choi operator of $\widetilde{\Lambda}$ since
\begin{align}
    J_{\widetilde{\Lambda}^T}
    &= \sum_{i,j} \ketbra{i}{j}_{\mathcal{O}} \otimes \widetilde{\Lambda}^T (\ket{i}{j})_{\mathcal{I}}\\
    &= \sum_{i',j'} \widetilde{\Lambda}(\ketbra{i'}{j'})_{\mathcal{O}} \otimes \ketbra{i'}{j'}_{\mathcal{I}}\\
    &=J_{\widetilde{\Lambda}}
\end{align}
holds for the computational bases $\{\ket{i}\}$ and $\{\ket{i'}\}$ of $\mathcal{O}$ and $\mathcal{I}$, respectively. Then, the condition (\ref{eq:cond_isometry_pss}) can be written in the Choi representation as
\begin{align}
    S\star \dketbra{V}^{\otimes k}_{\mathcal{I}\mathcal{O}} \star \dketbra{V}_{\mathcal{F}'\mathcal{F}} = p \dketbra{\1_d}_{\mathcal{F}'\mathcal{P}},\label{eq:cond_isometry_pss_choi}
\end{align}
where $\mathcal{F}'=\mathbb{C}^d$. Similarly to the case with isometry inversion, Eq.~(\ref{eq:cond_isometry_pss_choi}) holds for all $V\in\mathbb{V}_{\text{iso}}(d,D)$ if and only if it holds for finite set of isometries $\{V_i\}_i \subset \mathbb{V}_{\text{iso}}(d,D)$ forming the basis $\{\dketbra{V_i}^{\otimes k+1}\}_i$ of the linear space $\text{span}\{\dketbra{V}^{\otimes k+1}|V\in\mathbb{V}_{\text{iso}}(d,D)\}$.

Therefore, the optimal success probability of isometry pseudo complex conjugation in parallel/sequential/general protocols can be found as a solution of the following SDP:
\begin{align}
    &\max p\\
    \text{s.t. } & S, F \in \mathcal{L}(\mathcal{P}\otimes \mathcal{I}\otimes \mathcal{O} \otimes \mathcal{F}),\\
    & S, F\geq 0,\\
    & S\star \dketbra{V_i}^{\otimes k}_{\mathcal{I}\mathcal{O}} \star \dketbra{V_i}_{\mathcal{F}'\mathcal{F}} = p \dketbra{\1_d}_{\mathcal{F}'\mathcal{P}}\;\;\;\forall i,\\
    & C\coloneqq S+F \text{ is a parallel/sequential/general superchannel},
\end{align}
where $\{V_i\}_i$ is a finite set of isometries forming the basis $\{\dketbra{V_i}^{\otimes k+1}\}_i$ of the linear space $\mathcal{H}_{k+1}\coloneqq \text{span}\{\dketbra{V}^{\otimes k+1}|V\in\mathbb{V}_{\text{iso}}(d,D)\}$.

\subsection{Isometry complex conjugation and isometry transposition}
The superinstrument $\{\doublewidetilde{\mathcal{S}^{}}, \doublewidetilde{\mathcal{F}^{}}\}$ implements isometry complex conjugation or isometry transposition with the success probability $p$ if and only if
\begin{align}
    \doublewidetilde{\mathcal{S}^{}}(\widetilde{\mathcal{V}}^{\otimes k}) = p  \widetilde{f(V)}
\end{align}
holds for all $V\in \mathbb{V}_{\text{iso}}(d,D)$, where $ \widetilde{f(V)} = \widetilde{\mathcal{V}}^*$ or $ \widetilde{f(V)} = \widetilde{\mathcal{V}}^T$. This condition can be written in the Choi representation as
\begin{align}
    S\star \dketbra{V}^{\otimes k}_{\mathcal{I}\mathcal{O}} = p \dketbra{f(V)}_{\mathcal{P}\mathcal{F}},\label{eq:cond_isometry_fV_old}
\end{align}
where $f(V)=V^*$ or $f(V)=V^T$. The mapping $\dketbra{V} \mapsto \dketbra{f(V)}$ is linear  for $f(V)=V^*$ or $f(V)=V^T$ since
\begin{align}
    &\dket{V^T}=\sum_{i}\ket{i}\otimes V^T\ket{i}=\sum_{i}V\ket{i}\otimes \ket{i}=\dket{V},\\
    &\dketbra{V^*}=\dketbra{V}^*=\dketbra{V}^T,\label{eq:dketbraVstar}
\end{align}
hold, where the hermicity of $\dketbra{V}$ is used to derive Eq.~(\ref{eq:dketbraVstar}). Similarly to the case with isometry inversion, Eq.~(\ref{eq:cond_isometry_fV_old}) can be rewritten as
\begin{align}
    S\star \dketbra{V}^{\otimes k}_{\mathcal{I}\mathcal{O}} = p \dketbra{f(V)}_{\mathcal{P}\mathcal{F}}\times \frac{\Tr[\dketbra{V}^{\otimes k-1}]}{d^{k-1}}.\label{eq:cond_isometry_fV}
\end{align}
Due to the linearity of Eq.~(\ref{eq:cond_isometry_fV}) with respect to $\dketbra{V}^{\otimes k}$, the condition (\ref{eq:cond_isometry_fV}) holds for all $V\in\mathbb{V}_{\text{iso}}(d,D)$ if and only if it holds for finite set of isometries $\{V_j\}_j \subset \mathbb{V}_{\text{iso}}(d,D)$ forming the basis $\{\dketbra{V_j}^{\otimes k}\}_j$ of the linear space $\mathcal{H}_{k}\coloneqq \text{span}\{\dketbra{V}^{\otimes k}|V\in\mathbb{V}_{\text{iso}}(d,D)\}$.

Therefore, the optimal success probability of isometry complex conjugation or transposition in parallel/sequential/general protocols can be found as a solution of the following SDP:
\begin{align}
    &\max p\\
    \text{s.t. } & S, F \in \mathcal{L}(\mathcal{P}\otimes \mathcal{I}\otimes \mathcal{O} \otimes \mathcal{F}),\\
    & S, F\geq 0,\\
    & S\star \dketbra{V_j}^{\otimes k}_{\mathcal{I}\mathcal{O}} = p \dketbra{f(V_j)}_{\mathcal{P}\mathcal{F}}\;\;\;\forall j,\\
    & C\coloneqq S+F \text{ is a parallel/sequential/general superchannel},
\end{align}
where $f(V)=V^*$ or $f(V)=V^T$, and $\{V_j\}_j$ is a finite set of isometries forming the basis $\{\dketbra{V_j}^{\otimes k}\}_j$ of the linear space $\mathcal{H}_{k}\coloneqq \text{span}\{\dketbra{V}^{\otimes k}|V\in\mathbb{V}_{\text{iso}}(d,D)\}$.

\subsection{``Success-or-draw'' isometry inversion}
The superinstrument $\{\doublewidetilde{\mathcal{S}^{}}, \doublewidetilde{\mathcal{F}^{}}\}$ implements ``success-or-draw'' isometry inversion with the success probability $p$ if and only if
\begin{align}
    \doublewidetilde{\mathcal{S}^{}}(\widetilde{\mathcal{V}}^{\otimes k}) \circ \widetilde{\mathcal{V}} &= p \widetilde{\mathcal{V}}^{\text{embed}},\\
    \doublewidetilde{\mathcal{F}^{}}(\widetilde{\mathcal{V}}^{\otimes k}) \circ \widetilde{\mathcal{V}} &= p \widetilde{\mathcal{V}},
\end{align}
holds for all $V\in\mathbb{V}_{\text{iso}}(d,D)$, where $V^{\text{embed}}$ is a natural embedding of $\mathbb{C}^d$ to $\mathbb{C}^D$ defined in Eq.~(\ref{eq:def_embedding}). Note that $V^{\text{embed}}$ is introduced to adjust the dimension of the output space $\mathcal{F}$ of the supermap $\doublewidetilde{\mathcal{S}^{}}$ to that of the supermap $\doublewidetilde{\mathcal{F}^{}}$. These conditions can be written in the Choi representation as
\begin{align}
    S\star \dketbra{V}^{\otimes k}_{\mathcal{I}\mathcal{O}} \star \dketbra{V}_{\mathcal{P}'\mathcal{P}} &= p \dketbra{V^{\text{embed}}}_{\mathcal{P}'\mathcal{F}},\label{eq:isometry_inversion_sod_choi1}\\
    F\star \dketbra{V}^{\otimes k}_{\mathcal{I}\mathcal{O}} \star \dketbra{V}_{\mathcal{P}'\mathcal{P}} &= (1-p) \dketbra{V}_{\mathcal{P}'\mathcal{F}},\label{eq:isometry_inversion_sod_choi2}
\end{align}
where $\mathcal{P}'=\mathbb{C}^d$.
Due to the linearity of the conditions (\ref{eq:isometry_inversion_sod_choi1}) and (\ref{eq:isometry_inversion_sod_choi2}) with respect to $\dketbra{V}^{\otimes k+1}$, the conditions (\ref{eq:isometry_inversion_sod_choi1}) and (\ref{eq:isometry_inversion_sod_choi2}) hold for all $V\in\mathbb{V}_{\text{iso}}(d,D)$ if and only if they hold for finite set of isometries $\{V_i\}_i \subset \mathbb{V}_{\text{iso}}(d,D)$ forming the basis $\{\dketbra{V_i}^{\otimes k+1}\}_i$ of the linear space $\mathcal{H}_{k+1}\coloneqq \text{span}\{\dketbra{V}^{\otimes k+1}|V\in\mathbb{V}_{\text{iso}}(d,D)\}$.

Therefore, the optimal success probability of ``success-or-draw'' isometry inversion in parallel/sequential/general protocols can be found as a solution of the following SDP:
\begin{align}
    &\max p\\
    \text{s.t. } & S, F \in \mathcal{L}(\mathcal{P}\otimes \mathcal{I}\otimes \mathcal{O} \otimes \mathcal{F}),\\
    & S, F\geq 0,\\
    & S\star \dketbra{V_i}^{\otimes k}_{\mathcal{I}\mathcal{O}} \star \dketbra{V_i}_{\mathcal{P}'\mathcal{P}} = p \dketbra{V^{\text{embed}}}_{\mathcal{P}'\mathcal{F}}\;\;\;\forall i,\\
    & F\star \dketbra{V_i}^{\otimes k}_{\mathcal{I}\mathcal{O}} \star \dketbra{V_i}_{\mathcal{P}'\mathcal{P}} = (1-p) \dketbra{V_i}_{\mathcal{P}'\mathcal{F}}\;\;\;\forall i,\\
    & C\coloneqq S+F \text{ is a parallel/sequential/general superchannel},
\end{align}
where $V^{\text{embed}}$ is a natural embedding of $\mathbb{C}^d$ to $\mathbb{C}^D$ defined in Eq.~(\ref{eq:def_embedding}), and $\{V_i\}_i$ is a finite set of isometries forming the basis $\{\dketbra{V_i}^{\otimes k+1}\}_i$ of the linear space $\mathcal{H}_{k+1}\coloneqq \text{span}\{\dketbra{V}^{\otimes k+1}|V\in\mathbb{V}_{\text{iso}}(d,D)\}$.

\section{Proof of Theorem~\ref{theorem:isometry_inverse}}
\label{sec:appendix_isometry_inverse}
By linearity, it is sufficient to consider a pure input state $\rho_{\mathrm{in}}=\ketbra{\psi_\mathrm{in}}{\psi_\mathrm{in}} \in  \mathcal{L}(\mathcal{P})=\mathcal{L}( \mathbb{C}^D)$ for the isometry inversion protocol of $V \in \mathbb{V}_\mathrm{iso} (d, D)$. We decompose $\ket{\psi_\mathrm{in}}$ as
\begin{align}
    \ket{\psi_\mathrm{in}}=\ket{\psi^{\parallel}}+\ket{\psi^{\perp}},
\end{align}
where $\ket{\psi^{\parallel}}\in \Im V$ and $ \ket{\psi^{\perp}}\in (\Im V)^{\perp}$ can be unnormalized.
For the computational basis of $\mathbb{C}^d$ given by $\{\ket{i} \}_{i=0}^{d-1}$, a set of vectors $\{V\ket{i} \}_{i=0}^{d-1} \subset \mathbb{C}^D$ satisfies 
\begin{align}
    \bra{i}V^\dagger V \ket{j}&=\braket{i}{j}=\delta_{i,j},\\
    \bra{\psi^{\perp}}V\ket{i}&=0,
\end{align}
for $i, j\in \{0, \cdots, d-1\}$,
where $\delta_{i, j}$ is Kronecker's delta given by
\begin{align}
    \delta_{i, j}\coloneqq 
    \begin{cases}
    1 & (i=j)\\
    0 & (i\neq j)
    \end{cases}.\label{eq:def_Kronecker_delta}
\end{align}
Therefore, there exists an orthonormal basis $\{\ket{v(i)}\}_{i=0}^{D-1}$ of $\mathcal{P}=\mathbb{C}^D$ such that the first $d+1$ elements satisfy
\begin{align}
    \ket{v(i)}&=V\ket{i}\;\;\;(i\in\{0, \cdots, d-1\}),\\
    \ket{v(d)}&\parallel \ket{\psi^{\perp}}.
\end{align}
Then, the remaining part of the orthonormal basis $\{\ket{v(i)}\}_{i=0}^{D-1}$, namely, each $\ket{v(i)}$ for all $i \in\{d+1, \cdots, D-1\}$ satisfies
\begin{align}
    \braket{v(i)}{\psi_{\mathrm{in}}}=0.
\end{align}

For $0\leq j_1<\cdots <j_d\leq D-1$ and $\vec{j}=(j_1, \cdots, j_d)$, we define a totally antisymmetric state $\ket{a^v_{\vec{j}}}\in \mathcal{P}\otimes \mathcal{O}= (\mathbb{C}^D)^{\otimes d}$ by
\begin{align}
    \ket{a^v_{\vec{j}}}\coloneqq \sum_{\vec{k}\in\{1, \cdots, d\}^{d}}\frac{\epsilon_{\vec{k}}}{\sqrt{d!}}\ket{v(j_{k_1})\cdots v(j_{k_d})},
\end{align}
where $\epsilon_{\vec{k}}$ is the antisymmetric tensor with rank $d$.
The projector $\Pi^{\mathrm{a.s.}}_{\mathcal{PO}}$ on $\mathcal{P}\otimes \bigotimes_{i=1}^{d-1}\mathcal{O}_i$ onto its subspace spanned by totally antisymmetric states satisfies
\begin{align}
    \Pi^{\mathrm{a.s.}}_{\mathcal{PO}}=\sum_{0\leq j_1<\cdots <j_d \leq D-1} \ketbra{a^v_{\vec{j}}}{a^v_{\vec{j}}}.
\end{align}

We calculate the output state $\rho'_{\mathrm{out}} \in  \mathcal{L}(\mathcal{F})=\mathcal{L}(\mathbb{C}^d)$ after obtaining the measurement outcome $a=1$ of $\mathcal{M}$ and the probability $p_{a=1}$ to obtain the measurement outcome $a=1$. To this end, we calculate the unnormalized operator $\rho_{\mathrm{out}}=p_{a=1}\rho'_{\mathrm{out}}$. First, we obtain 
\begin{align}
    \rho_{\mathrm{out}}=\Tr_{\mathcal{PO}}(\ketbra{\phi'}{\phi'}_{\mathcal{POF}}),
\end{align}
where $\ket{\phi'}$ is defined by
\begin{align}
    \ket{\phi'}\coloneqq \Pi^{\mathrm{a.s.}}_{\mathcal{PO}}\otimes \1_{\mathcal{F}}
    \left[
    \ket{\psi_\mathrm{in}}_{\mathcal{P}}\otimes (V^{\otimes d-1}_{\mathcal{I}\to \mathcal{O}}\otimes \1_{\mathcal{F}})\ket{A_d}_{\mathcal{IF}}
    \right].
\end{align}
The vector $\ket{\phi'}$ is calculated as
\begin{align}
    \ket{\phi'}
    &=\sum_{\vec{j}}\ketbra{a^v_{\vec{j}}}{a^v_{\vec{j}}}_{\mathcal{PO}}\otimes \1_{\mathcal{F}} 
    \left\{
    \ket{\psi_\mathrm{in}}_{\mathcal{P}}\otimes 
    \left[
    (V^{\otimes d-1}_{\mathcal{I}\to \mathcal{O}}\otimes \1_{\mathcal{F}})\ket{A_d}_{\mathcal{IF}}
    \right]
    \right\}\\
    &=\sum_{\vec{j}, \vec{k}, \vec{k}'}\ket{a^v_{\vec{j}}}_{\mathcal{PO}}\otimes 
    \frac{\epsilon_{\vec{k}}\epsilon_{\vec{k}'}}{d!}\delta_{j_{k_2},k'_1}\cdots \delta_{j_{k_{d}},k'_{d-1}}\ket{k'_d}_{\mathcal{F}}\braket{v(j_{k_1})}{\psi_\mathrm{in}},\label{eq:b6}
\end{align}
where the summation is taken over $0\leq j_1<\cdots < j_d\leq D-1$, $\vec{k}\in\{1, \cdots, d\}^{d}$ and $\vec{k}'\in\{0, \cdots, d-1\}^{d}$. Since the summand in Eq.~(\ref{eq:b6}) is non-zero only for $(j_1, \cdots, j_d)=(0, \cdots, d-1)$ or $(j_1, \cdots, j_d)=(0, \cdots, j-1, j+1, \cdots, d)$ for $j\in \{0, \cdots, d-1\}$, we obtain
\begin{align}
    \ket{\phi'}
    &=\ket{a^v_{0\cdots d-1}}_{\mathcal{PO}}\otimes \sum_{j=0}^{d-1}\frac{(-1)^d}{d}\ket{j}_{\mathcal{F}}\braket{v(j)}{\psi_\mathrm{in}}+\sum_{j=0}^{d-1}\ket{a^v_{0\cdots j-1 j+1\cdots d}}_{\mathcal{PO}}\otimes \frac{(-1)^d}{d}\ket{j}_{\mathcal{F}}\braket{v(d)}{\psi_\mathrm{in}}\\
    &=\ket{a^v_{0\cdots d-1}}_{\mathcal{PO}}\otimes \frac{(-1)^d}{d}V^{\dagger}_{\mathcal{P}\to\mathcal{F}}\ket{\psi_\mathrm{in}}_{\mathcal{P}}+\sum_{j=0}^{d-1}\ket{a^v_{0\cdots j-1 j+1\cdots d}}_{\mathcal{PO}}\otimes \frac{(-1)^d}{d}\ket{j}_{\mathcal{F}}\braket{v(d)}{\psi_\mathrm{in}}.
\end{align}

Since $\{\ket{a^v_{\vec{j}}}\}$ are orthogonal to each other, the successful output state multiplied by the probability to obtain the measurement outcome $a=1$, denoted by $\rho_{\mathrm{out}}$, can be calculated as
\begin{align}
    \rho_\mathrm{out}
    &=\Tr_{\mathcal{PO}}(\ketbra{\phi'}{\phi'}_{\mathcal{POF}})\\
    &=\frac{1}{d^2}V^{\dagger}\ketbra{\psi_\mathrm{in}}{\psi_\mathrm{in}}V+\frac{1}{d^2}\sum_{j=0}^{d-1}\ketbra{j}{j}\left|\braket{v(d)}{\psi_\mathrm{in}}\right|^2\\
    &=\frac{1}{d^2}V^{\dagger}\ketbra{\psi_\mathrm{in}}{\psi_\mathrm{in}}V+\frac{\pi_d}{d}
    \Tr
    \left[
    \Pi_{(\Im V)^{\perp}}\ketbra{\psi_\mathrm{in}}{\psi_\mathrm{in}}
    \right]\\
    &=\frac{1}{d^2}\widetilde{\mathcal{V}}^\prime(\ketbra{\psi_\mathrm{in}}{\psi_\mathrm{in}}),
\end{align}

by defining a CPTP map $\widetilde{\mathcal{V}}^\prime: \mathcal{L}(\mathbb{C}^D) \rightarrow \mathcal{L}(\mathbb{C}^d)$ as $\widetilde{\mathcal{V}}^\prime(\rho_\mathrm{in})\coloneqq V^{\dagger} \rho_{\mathrm{in}} V+I_d\Tr(\Pi_{(\Im V)^{\perp}}\rho_\mathrm{in})$. Thus, a protocol shown in Figure~\ref{fig:isometry_inversion_protocols}~(c) implements a pseudo complex conjugate map $\widetilde{\mathcal{V}}^\prime$ with a success probability $p_{\mathrm{succ}}=1/d^2$. 

\section{Proof of Lemma~\ref{lem:Lambda}}
\label{sec:appendix_Lambda}
Suppose $\mathcal{P}= \mathbb{C}^D$, $\mathcal{P}'= \mathbb{C}^d$, $\mathcal{P}'= \mathbb{C}^d$, $\mathcal{I}_i= \mathbb{C}^d$, $\mathcal{O}_i= \mathbb{C}^D$, and $\mathcal{O}'_i= \mathbb{C}^d$ for $i\in\{1, \cdots, k\}$ and define the joint Hilbert space by $\mathcal{I}\coloneqq \bigotimes_{i=1}^{k}\mathcal{I}_i$, $\mathcal{O}\coloneqq \bigotimes_{i=1}^{k}\mathcal{O}_i$, and $\mathcal{O}'\coloneqq \bigotimes_{i=1}^{k}\mathcal{O}_i$. For $V\in \mathbb{V}_\mathrm{iso} (d,D)$, we define $\widetilde{\Lambda}_V:\mathcal{L}(\mathcal{P}''\otimes \mathcal{I})\to \mathcal{L}(\mathcal{P}'\otimes \mathcal{O}'')$ by
\begin{align}
    \widetilde{\Lambda}_V(\rho)\coloneqq 
    \left(
    \widetilde{\Psi}_{\mathcal{P}\mathcal{O}\to\mathcal{P}'\mathcal{O}'}
    \circ \widetilde{\mathcal{V}}^{\otimes k+1}_{\mathcal{P}''\mathcal{I}\to \mathcal{P}\mathcal{O}}
    \right)
    (\rho)
\end{align}
using the CPTP map $\widetilde{\Psi}$ given by Eq.~(\ref{eq:def_Psi}).
From Eqs.~(\ref{eq:isometry_schur_weyl}) and (\ref{eq:def_Psi}), we obtain

\begin{align}
    \widetilde{\Lambda}_V(\rho)
    &=\bigoplus_{\mu\vdash k+1}\frac{\1_{\mathcal{U}^{(d)}_{\mu, \mathcal{P}'\mathcal{O}'}}}{d_{\mathcal{U}_{\mu}^{(d)}}}\otimes
    \left\{
    \widetilde{\mathcal{I}}_{\mathcal{S}^{(k+1)}_{\mathcal{PO}}\to \mathcal{S}^{(k+1)}_{\mathcal{P}'\mathcal{O}'}}
    \left(
    \Tr_{\mathcal{U}_{\mu, \mathcal{P}\mathcal{O}}^{(D)}}
    \left\{
    \Pi_{\mu, \mathcal{PO}}
    \left[
    \widetilde{\mathcal{V}}_\mu\otimes \widetilde{\1}_{\mathcal{S}^{(k+1)}_{\mu, \mathcal{PO}}}(\rho)
    \right]
    \right\}
    \right)
    \right\}\\
    &=\bigoplus_{\mu\vdash k+1}\frac{\1_{\mathcal{U}^{(d)}_{\mu, \mathcal{P}'\mathcal{O}'}}}{d_{\mathcal{U}_{\mu}^{(d)}}}\otimes
    \left\{
    \widetilde{\mathcal{I}}_{\mathcal{S}^{(k+1)}_{\mathcal{P}''\mathcal{I}}\to \mathcal{S}^{(k+1)}_{\mathcal{P}'\mathcal{O}'}}\left[\Tr_{\mathcal{U}_{\mu, \mathcal{P}''\mathcal{I}}^{(d)}}(\Pi_{\mu, \mathcal{P}''\mathcal{I}}\rho)\right]
    \right\}.\label{eq:e3}
\end{align}
On the other hand, we define $\widetilde{\Lambda}_U :\mathcal{L}(\mathcal{P}''\otimes \mathcal{I})\to \mathcal{L}(\mathcal{P}'\otimes \mathcal{O}'')$ by
\begin{align}
    \widetilde{\Lambda}_U(\rho)= \int \dd U \mathcal{U}^{\otimes k+1}_{\mathcal{P}''\mathcal{I}\to\mathcal{P}'\mathcal{O}'}(\rho),
\end{align}
where $\dd U$ is the Haar measure on $\mathbb{U}(D)$.
From Eqs.~(\ref{eq:haar_unitary_1}), (\ref{eq:haar_unitary_2}) and (\ref{eq:e3}), we obtain $\widetilde{\Lambda}_V=\widetilde{\Lambda}_U$.

\section{Proof of Theorem~\ref{thm:conjugate_modoki}}
\label{sec:appendix_isometry_conjugate_modoki}
For $\vec{j}=(j_1,\cdots, j_d)\in\{0,\cdots, D-1\}^d$, we define totally antisymmetric states in  $(\mathbb{C}^D)^{\otimes d}$ as
\begin{align}
    \ket{a_{\vec{j}}}&\coloneqq \sum_{\vec{k}\in\{1, \cdots, d\}^{d}}\frac{\epsilon_{\vec{k}}}{\sqrt{d!}}\ket{j_{k_1}\cdots j_{k_d}},\\
    \ket{a^{v*}_{\vec{j}}}&\coloneqq \sum_{\vec{k}\in\{1, \cdots, d\}^{d}}\frac{\epsilon_{\vec{k}}}{\sqrt{d!}}\ket{v(j_{k_1})^{*}\cdots v(j_{k_d})^{*}},
\end{align}
where $\ket{v(j)}$ is defined in Appendix \ref{sec:appendix_isometry_inverse} and $\ket{v(j)^{*}}$ is the complex conjugate of $\ket{v(j)}$ in terms of the computational basis. For simplicity, we introduce short-hand notations
\begin{align}
    \ket{\vec{j}_{d-1}}&\coloneqq \ket{j_1}\otimes\cdots\otimes\ket{j_{d-1}},\\
    \ket{j_{\vec{k}_{d-1}}}&\coloneqq \ket{j_{k_1}}\otimes\cdots\otimes\ket{j_{k_{d-1}}},\\
    \ket{v(j_{\vec{k}_{d-1}})}&\coloneqq \ket{v(j_{k_1})}\otimes \cdots \otimes \ket{v(j_{k_{d-1}})},\\
    \ket{v(j_{\vec{k}})^*}&\coloneqq \ket{v(j_{k_1})^*}\otimes \cdots \otimes \ket{v(j_{k_{d}})^*},
\end{align}
where the vectors are defined by $\vec{j}_{d-1}=(j_1, \cdots, j_{d-1})$,  $\vec{k}_{d-1}=(k_1, \cdots, k_{d-1})$ and $\vec{k}=(k_1, \cdots, k_{d})$.

First, we calculate the  Choi operator of $\widetilde{\Lambda}$ as
\begin{align}
    J_{\widetilde{\Lambda}}
    &=c\sum_{\vec{j}_{d-1}, \vec{j}'_{d-1}, \vec{j}''}\ketbra{\vec{j}_{d-1}}{\vec{j}'_{d-1}}_{\mathcal{O}}\otimes
    A_{\vec{j}''}\ketbra{\vec{j}_{d-1}}{\vec{j}'_{d-1}} A_{\vec{j}''}^{\dagger}  + (\1_{\mathcal{O}}-\Pi_{\mathcal{O}}^{\text{a.s.}}) \otimes \frac{\1_{\mathcal{F}}}{D}\label{eq:j5}\\
    &=c\sum_{\vec{j}'', \vec{k}, \vec{k}'}\frac{\epsilon_{\vec{k}}\epsilon_{\vec{k}'}}{(d-1)!}\ketbra{j''_{\vec{k}_{d-1}}}{j''_{\vec{k}'_{d-1}}}_{\mathcal{O}}\otimes \ketbra{j''_{k_d}}{j''_{k'_d}}_{\mathcal{F}} + (\1_{\mathcal{O}}-\Pi_{\mathcal{O}}^{\text{a.s.}}) \otimes \frac{\1_{\mathcal{F}}}{D}\label{eq:j6}\\
    &=cd\sum_{0\leq j_1<\cdots < j_d\leq D-1}\ketbra{a_{\vec{j}}}{a_{\vec{j}}}_{\mathcal{OF}} + (\1_{\mathcal{O}}-\Pi_{\mathcal{O}}^{\text{a.s.}}) \otimes \frac{\1_{\mathcal{F}}}{D}\\
    &=cd\Pi^{\mathrm{a.s.}}_{\mathcal{OF}} + (\1_{\mathcal{O}}-\Pi_{\mathcal{O}}^{\text{a.s.}}) \otimes \frac{\1_{\mathcal{F}}}{D},
\end{align}
where a coefficient $c$ is given by $c\coloneqq 1/(D-d+1)$ and the summation in Eqs.~(\ref{eq:j5}) and (\ref{eq:j6}) are taken over $\vec{j}_{d-1}, \vec{j}'_{d-1}\in \{0, \cdots, D-1\}^{d-1}$, $0\leq j''_1<\cdots < j''_d\leq D-1$ and $\vec{k}, \vec{k}'\in\{1, \cdots, d\}^d$. Since $\Pi^{\mathrm{a.s.}}$ is invariant under the tensor product $U^{\otimes d}$ of a unitary operator $U\in\mathbb{U}(D)$, we obtain
\begin{align}
    J_{\widetilde{\Lambda}}
    &=cd\sum_{0\leq j_1<\cdots < j_d\leq D-1}\ketbra{a^{v*}_{\vec{j}}}{a^{v*}_{\vec{j}}}_{\mathcal{OF}} + (\1_{\mathcal{O}}-\Pi_{\mathcal{O}}^{\text{a.s.}}) \otimes \frac{\1_{\mathcal{F}}}{D}\\
    &=c\sum_{\vec{j}'', \vec{k}, \vec{k}'}\frac{\epsilon_{\vec{k}}\epsilon_{\vec{k}'}}{(d-1)!}\ketbra{v(j''_{\vec{k}})^*}{v(j''_{\vec{k}'})^*}_{\mathcal{OF}} + (\1_{\mathcal{O}}-\Pi_{\mathcal{O}}^{\text{a.s.}}) \otimes \frac{\1_{\mathcal{F}}}{D},\label{eq:j11}
\end{align}
where the summation in Eq.~(\ref{eq:j11}) is taken over $0\leq j''_1<\cdots < j''_d\leq D-1$ and $\vec{k}, \vec{k}'\in\{1, \cdots, d\}^d$.

Then, the output state $\rho_{\mathrm{out}}\in\mathcal{L}(\mathcal{F})$ of the pseudo complex conjugation protocol is calculated as
\begin{align}
    &\rho_{\mathrm{out}}\nonumber\\
    &=J_{\widetilde{\Lambda}}\star 
    \left[
    \left(
    \widetilde{\mathcal{V}}^{\otimes d-1}_{\mathcal{I}\to\mathcal{O}}\circ \widetilde{\mathcal{V}}^{\mathrm{a.s.}}_{\mathcal{P}\to\mathcal{I}}
    \right)
    (\rho_{\mathrm{in}})
    \right]\\
    &=c\sum_{\vec{j}, \vec{k}, \vec{k}'}\frac{\epsilon_{\vec{k}}\epsilon_{\vec{k}'}}{(d-1)!}
    \Tr
    \left\{
    \ketbra{v(j_{\vec{k}'_{d-1}})}{v(j_{\vec{k}_{d-1}})}_{\mathcal{O}} 
    \left[
    \left(
    \widetilde{\mathcal{V}}^{\otimes d-1}_{\mathcal{I}\to\mathcal{O}}\circ \widetilde{\mathcal{V}}^{\mathrm{a.s.}}_{\mathcal{P}\to\mathcal{I}}
    \right)
    (\rho_{\mathrm{in}})
    \right]
    \right\}
    \ketbra{v(j_{k_{d}})^*}{v(j_{k'_{d}})^*}_{\mathcal{F}}\\
    &=c\sum_{\vec{j}, \vec{k}, \vec{k}'}\frac{\epsilon_{\vec{k}}\epsilon_{\vec{k}'}}{(d-1)!}\Tr
    \left\{
    \left[
    \left(
    \widetilde{\mathcal{V}}^{\dagger \mathrm{a.s.}}_{\mathcal{I}\to\mathcal{P}}\circ \widetilde{\mathcal{V}}^{\dagger\otimes d-1}_{\mathcal{O}\to\mathcal{I}}
    \right)
    \Big(\ketbra{v(j_{\vec{k}'_{d-1}})}{v(j_{\vec{k}_{d-1}})}_{\mathcal{O}}
    \Big)
    \right]
    \rho_{\mathrm{in}}
    \right\}
    \ketbra{v(j_{k_{d}})^*}{v(j_{k'_{d}})^*}_{\mathcal{F}},\label{eq:i12}
\end{align}
where the summation is taken over $0\leq j_1<\cdots < j_d\leq D-1$ and $\vec{k}, \vec{k}'\in\{1, \cdots, d\}^d$. Since the summand in Eq.~(\ref{eq:i12}) is non-zero only when 
\begin{align}
    \begin{cases}
    \vec{j}=(0, \cdots, d-1) & \mathrm{or}\\
    \vec{j}=(0, \cdots, l-1, l+1, \cdots, d-1, m) & (l\in\{0, \cdots, d-1\}, m\in\{d, \cdots, D-1\})
    \end{cases},
\end{align}
we obtain
\begin{align}
    \rho_{\mathrm{out}}
    &=c\sum_{\vec{j}, \vec{j}'\in\{0, \cdots, d-1\}^{d}}\frac{\epsilon_{\vec{j}}\epsilon_{\vec{j}'}}{(d-1)!}
    \Tr
    \left\{
    \left[
    \widetilde{\mathcal{V}}^{\dagger\mathrm{a.s.}}_{\mathcal{I}\to\mathcal{P}}
    \Big(
    \ketbra{\vec{j}'_{d-1}}{\vec{j}_{d-1}}_{\mathcal{I}}
    \Big)
    \right]
    \rho_{\mathrm{in}}
    \right\}
    \ketbra{v(j_{d})^*}{v(j'_{d})^*}_{\mathcal{F}}\nonumber\\
    &\hspace{12pt}+c\sum_{l=0}^{d-1}\sum_{m=d}^{D-1}\sum_{\vec{k}_{d-1}, \vec{k}'_{d-1}\in \{1,\cdots, d-1\}^{d-1}}\frac{\epsilon_{\vec{k}_{d-1}}\epsilon_{\vec{k}'_{d-1}}}{(d-1)!}
    \Tr
    \left\{
    \left[
    \widetilde{\mathcal{V}}^{\dagger\mathrm{a.s.}}_{\mathcal{I}\to\mathcal{P}}
    \Big(
    \ketbra{j_{\vec{k}'_{d-1}}}{j_{\vec{k}_{d-1}}}_{\mathcal{I}}
    \Big)
    \right]
    \rho_{\mathrm{in}}
    \right\}\nonumber\\
    &\hspace{180pt}\times
    \ketbra{v(m)^*}{v(m)^*}_{\mathcal{F}},\label{eq:i13}
\end{align}
where $(j_1, \cdots, j_{d-1})$ in the second term is $(j_1, \cdots, j_{d-1})=(0, \cdots, l-1, l+1, \cdots, d-1)$. Then, we proceed the calculation as
\begin{align}
    \rho_{\mathrm{out}}
    &=c\sum_{\vec{j}, \vec{j}'\in\{0, \cdots, d-1\}^{d}}\frac{1}{[(d-1)!]^2}
    \Tr
    \left[
    \ketbra{j'_d}{j_d}_{\mathcal{P}}
    \rho_{\mathrm{in}}
    \right]
    \ketbra{v(j_{d})^*}{v(j'_{d})^*}_{\mathcal{F}}\nonumber\\
    &\hspace{12pt}+c\sum_{l=0}^{d-1}\sum_{m=d}^{D-1}\sum_{\vec{k}_{d-1}, \vec{k}'_{d-1}\in \{1,\cdots, d-1\}^{d-1}}\frac{1}{[(d-1)!]^2}\Tr[
    \ketbra{l}{l}_{\mathcal{P}}
    \rho_{\mathrm{in}}]\ketbra{v(m)^*}{v(m)^*}_{\mathcal{F}}\\
    &=c\sum_{j_d, j'_d=0}^{d-1}\Tr[\ketbra{j'_d}{k_d}_{\mathcal{P}}
    \rho_{\mathrm{in}}]\ketbra{v(j_{d})^*}{v(j'_{d})^*}_{\mathcal{F}}+c\sum_{l=0}^{d-1}\sum_{m=d}^{D-1}\Tr[
    \ketbra{l}{l}_{\mathcal{P}}
    \rho_{\mathrm{in}}]\ketbra{v(m)^*}{v(m)^*}_{\mathcal{F}}\\
    &=c\sum_{j_d, j'_d=0}^{d-1}
    \ketbra{v(j_d)^*}{j_d} \rho_{\mathrm{in}}\ketbra{j'_d}{v(j'_d)^*}
    +c\sum_{m=d}^{D-1}\sum_{l=0}^{d-1}
    \ketbra{v(m)^*}{v(m)^*}\bra{l}\rho_{\mathrm{in}}\ket{l}\\
    &=c
    \left[
    V^*\rho_{\mathrm{in}}(V^*)^{\dagger}+\Pi^*_{(\Im V)^{\perp}}\Tr(\rho_{\mathrm{in}})
    \right]\\
    &=\frac{1}{D-d+1}\widetilde{\mathcal{V}}''(\rho_{\mathrm{in}}).
\end{align}
Therefore, the protocol shown in Figure~\ref{fig:isometry_conjugate_modoki} implements an pseudo complex conjugate map $\widetilde{\mathcal{V}}''$ with a success probability $p_{\mathrm{succ}}=\frac{1}{D-d+1}$.

\bibliographystyle{apsrev4-2}
\bibliography{main}

\end{document}